%% file: main.tex
\documentclass[aps,prd,twocolumn,showkeys,amsmath,amssymb,longbibliography,superscriptaddress]{revtex4-2}
\usepackage[utf8]{inputenc}
\usepackage[english]{babel}
\usepackage{graphicx}
\usepackage{xcolor}
\usepackage{amsmath,amsfonts,amssymb}
\usepackage{amsthm}
\usepackage[hypertexnames=false,colorlinks]{hyperref}
\usepackage{braket}
\usepackage{url}
\usepackage{xfrac}
\usepackage{multirow}
\usepackage[caption=false]{subfig}
\usepackage{minitoc}
\usepackage{xpatch}
\usepackage{booktabs}
\usepackage{adjustbox}
\usepackage[qm]{qcircuit}
\usepackage[section]{placeins}

\newtheorem{theorem}{Theorem}[section]
\let\Pr\relax                    
\DeclareMathOperator{\Pr}{\mathbb{P}} 

\newtheorem{lemma}{Lemma}

\begin{document}

\title{How to find expressible and trainable parameterized quantum circuits?}

\author{Peter Röseler}
\email{p.roeseler@fz-juelich.de}
\affiliation{J\"ulich Supercomputing Centre, Forschungszentrum J\"ulich, 52425 J\"ulich, Germany}
\affiliation{Faculty of Mathematics and Natural Sciences, University of Cologne, 50923 K\"oln, Germany}

\author{Dennis Willsch}
\affiliation{J\"ulich Supercomputing Centre, Forschungszentrum J\"ulich, 52425 J\"ulich, Germany}
\affiliation{Faculty of Medical Engineering and Technomathematics, University of Applied Sciences Aachen,\\
52428 J\"ulich, Germany}

\author{Kristel~Michielsen}
\affiliation{J\"ulich Supercomputing Centre, Forschungszentrum J\"ulich, 52425 J\"ulich, Germany}
\affiliation{AIDAS, 52425 J\"ulich, Germany}
\affiliation{Faculty of Mathematics and Natural Sciences, University of Cologne, 50923 K\"oln, Germany}

\begin{abstract}
Whether parameterized quantum circuits (PQCs) can be systematically constructed to be both trainable and expressive remains an open question. Highly expressive PQCs often exhibit barren plateaus, while several trainable alternatives admit efficient classical simulation. We prove that deciding whether a PQC family exhibits a barren plateau is undecidable in general. However, we derive a practical finite-sample, dimension-independent concentration bound for estimating the variance of a PQC cost function, yielding explicit trainability guarantees. Across commonly used ansätze, we observe an anticorrelation between trainability and expressibility, consistent with theoretical insights. Building on this observation, we propose a property-based ansatz-search framework for identifying circuits that combine trainability and expressibility. We demonstrate its practical viability on a real quantum computer and apply it to variational quantum algorithms. We identify quantum neural network ansätze with improved effective dimension using over six times fewer parameters, and for VQE on $\mathrm{H}_2$ we achieve UCCSD-like accuracy at substantially reduced circuit complexity.
\end{abstract}

\date{\today}

\maketitle


\section{Introduction}

Parameterized quantum circuits (PQCs) are the core model class of variational quantum algorithms (VQAs) \cite{Cerezo2021VQA}. They underlie prominent examples such as the variational quantum eigensolver (VQE) \cite{Peruzzo2014VQE} and the quantum approximate optimization algorithm (QAOA) \cite{farhi2014quantumapproximateoptimizationalgorithm}, and they also form the basis of many quantum neural networks (QNNs) \cite{farhi2018classificationquantumneuralnetworks,Abbas2021PowerQNN,ji2026quantumdeeplearning}. In these methods, the choice of ansatz is a primary performance bottleneck. An ansatz that is too restricted may not represent the target solution (low expressibility), while an overly flexible circuit can induce unfavorable optimization landscapes (low trainability), in particular associated with gradients that are vanishing in high-dimensional space (barren plateaus \cite{McClean2018BarrenPlateaus}). 

An important open question is whether there exist expressible and trainable quantum circuits. \cite{ExprVSBP_Holmes} Recent works suggest that highly expressive families can be prone to barren plateaus \cite{ExprVSBP_Holmes}, while circuit classes that avoid or soften barren plateaus may do so by exploiting structures that correlate with efficient classical simulation \cite{Cerezo2025AbsenceBarrenPlateausSimulability,Leone2024practicalusefulness,gilfuster2025relationtrainabilitydequantizationvariational}. Understanding and navigating this tradeoff is crucial for scaling VQAs.

While expressibility admits a range of practical generalizable metrics \cite{Sim2019,du2022efficient}, \emph{quantifying trainability} in a way that is both informative and broadly applicable remains less settled \cite{Cunningham2025MitigatingBarrenPlateaus,Larocca2025BarrenPlateausReview,Qi2023BarrenPlateausQNN}. Barren plateau theory provides diagnostic and mitigation tools, including approaches based on dynamical Lie algebras \cite{Fontana_2024,Ragone2024LieAlgebraBarrenPlateaus} (and extensions thereof \cite{diaz2023showcasingbarrenplateautheory}), initialization strategies \cite{NEURIPS2022_7611a3cb,Grant2019initialization,9951195}, local observables \cite{Cerezo2021CostFunctionBarrenPlateaus}, diagnostic tools for problem-based ansätze such as QAOA \cite{Larocca2022diagnosingbarren}, variable structure ansätze \cite{Grimsley2019ADAPTVQE,Tang2021QubitADAPTVQE,Ramoa2025ADAPTVQE,Yordanov2021QEBAVQE}, and shallow/structured circuit families such as quantum convolutional neural networks \cite{PhysRevX.11.041011} or quantum tree tensor products and multiscale entanglement renormalization ansatz \cite{CerveroMartin2023barrenplateausin}. However, many of these approaches rely on specific circuit structure, initial-state requirements, or problem restrictions, and can be costly to apply at scale.

Because gradient variance is heavily influenced by both the Hamiltonian and the circuit structure \cite{Uvarov_2021}, a broadly applicable notion of trainability would be desirable to find and compare ansätze across diverse problem settings. However, we show that the barren plateau problem is undecidable. More precisely, given a program that, for every \(n\), outputs an \(n\)-qubit parameterized circuit and its cost observable, no algorithm can decide in general whether the resulting family satisfies the barren-plateau criterion. This holds even for depth-one product circuits with a single parameter \(\theta\) drawn uniformly from \([0,2\pi]\). 

We therefore introduce a \emph{finite-sample, dimension-free} concentration bound that enables estimating whether a given PQC, under arbitrary initialization and cost function, exhibits barren-plateau behavior. We use this bound to define a practical trainability objective through empirical gradient statistics.

Building on this, we pursue \emph{property-based} ansatz design. Instead of optimizing directly for a computationally expensive downstream objective, we optimize proxy properties---\emph{expressibility}, \emph{trainability}, \emph{entanglement}, and \emph{complexity}---to identify hardware-efficient, non-simulable, and optimizable circuits. This complements a broad literature on quantum architecture/ansatz search that primarily targets the final objective \cite{10821367}, including greedy methods \cite{PhysRevResearch.2.023074,Grimsley2019ADAPTVQE,Ostaszewski2021structure}, reinforcement learning \cite{kuo2021quantumarchitecturesearchdeep,Zhu2023QuantumArchitectureSearchPPO,patel2024curriculumreinforcementlearningquantum,Pirhooshyaran2021QuantumCircuitDesignSearch}, Bayesian optimization \cite{BenitezBuenache2025BPQCO,Pirhooshyaran2021QuantumCircuitDesignSearch,Duong2022QuantumNASBayesian}, performance predictors \cite{Zhang_2021,He2023GNNPredictorQAS}, Monte Carlo tree search \cite{10100911,9566740}, evolutionary algorithms \cite{9773233,Pirhooshyaran2021QuantumCircuitDesignSearch}, one-shot approaches \cite{Du2022QuantumCircuitArchitectureSearch}, and differentiable architecture search \cite{Zhang_2022,pmlr-v202-wu23v}. Although directly optimizing the downstream objective is conceptually aligned with the ultimate goal, it is often prohibitively expensive and can necessitate surrogate simplifications that may be misleading. 

We derive a practical criterion to diagnose barren-plateau behavior from empirical gradients, which we define with the expressibility, trainability, entanglement, and complexity metrics in Section \ref{sec:theory}. We answer the question of how to identify expressible and trainable circuits in simulation and on a real quantum device, how they compare to known ansätze, and if these properties lead to improved performance for VQA algorithms (Section \ref{sec:experiments}). In Section \ref{sec:summary}, we summarize our findings and discuss limitations and open questions of metric-based design, including the dependence on the chosen reference state and the possibility of \emph{metric drift} along the optimization trajectory.


\section{Metrics}
\label{sec:theory}
Many metrics have been proposed to characterize parameterized quantum circuits (PQCs)~\cite{illesova2025qmetric}. In this section, we focus on criteria that are informative for training variational quantum algorithms on NISQ hardware: expressibility, trainability, entanglement, and complexity. For trainability, we introduce an operational barren-plateau diagnostic based on a finite-sample concentration bound for estimating gradient fluctuations from bounded samples, while for the others, we build on established research. Although we emphasize these metrics in our experiments, the proposed ansatz-search framework can and should be extended to incorporate additional properties for specific applications.

\subsection{Expressibility}
To approximate a desired solution, a parameterized quantum circuit must be sufficiently expressive to represent the solution state (or a good approximation thereof). We define expressibility as the ability of a parameterized quantum circuit (PQC) to generate a representative set of quantum states across a target (sub)space of the Hilbert space given by $\operatorname{Pr}_\mathrm{target}(\Omega)$ quantified by the Kullback–Leibler (KL) divergence such that
\begin{equation}
\label{eq:generic_expr}
    \text{Expr} = \frac{1}{|\mathcal{S}_0|}\sum_{\ket{\psi_0}\in \mathcal{S}_0}D_{\mathrm{KL}}(\Pr_\mathrm{PQC}(\Omega;\Theta,\ket{\psi_0})||\Pr_\mathrm{target}(\Omega)),
\end{equation}
where $\mathcal{S}_0$ is the set of initial states the PQC is applied to and $\Theta$ the set of parameter samples that generate the probability distribution $\Pr_\mathrm{PQC}(\Omega;\Theta,\ket{\psi_0})$ from the PQC.

Although one may choose specific target distributions $\operatorname{Pr}_\mathrm{target}$ depending on the task (see \cite{roseler2025efficient}), we use Haar-random pure states as a standard, problem-independent baseline following \cite{Sim2019}. We emphasize that this choice is not meant to be optimal for any particular application domain, but rather to provide a common reference point (random states are also useful in many other contexts \cite{jin2021random}). Formally, let $U(\boldsymbol\theta)\ket{\psi_0}$ be a PQC acting on an initial state $\ket{\psi_0}$, and let $\boldsymbol\theta, \boldsymbol\theta'$ be uniformly drawn i.i.d from a parameter domain. Define the random states
\begin{equation}
    \ket{\psi_{\boldsymbol\theta}} := U(\boldsymbol\theta)\ket{\psi_0},
    \qquad
    \ket{\psi_{\boldsymbol\theta'}} := U(\boldsymbol\theta')\ket{\psi_0},
\end{equation}
and the associated fidelity random variable
\begin{equation}
    F \;:=\; \bigl|\braket{\psi_{\boldsymbol\theta}|{\psi_{\boldsymbol\theta'}}}\bigr|^2 \in [0,1].
\end{equation}
For fixed $\ket{\psi_0}$, the PQC induces a distribution $\Pr_{\mathrm{PQC}}(F\,;\ket{\psi_0})$. In practice, we estimate this distribution from a set $\Theta$ of sampled parameter pairs
$\{(\boldsymbol\theta^{(i)},\boldsymbol\theta'^{(i)})\}_{i=1}^{|\Theta|}$, yielding an empirical (histogram) distribution $\Pr_{\mathrm{PQC}}(F\,;\Theta,\ket{\psi_0})$. The fidelity density for two independent Haar-random pure states is 
\begin{equation}
    \label{eq:fid_distr}
    p_{\mathrm{Haar}}(F) \;=\; (N-1)(1-F)^{N-2}
\end{equation}
with $N$ being the dimension of the Hilbert space \cite{M_Kus_1988}. Quantifying expressibility using the KL divergence between these two distributions gives
\begin{equation}
    \text{Expr} = \frac{1}{|\mathcal{S}_0|}\sum_{\ket{\psi_0}\in \mathcal{S}_0}D_{\mathrm{KL}}(\Pr_\mathrm{PQC}(F;\Theta,\ket{\psi_0})||\Pr_\mathrm{Haar}(F)).
\end{equation}
Lower expressibility values indicate circuits whose generated state ensembles more closely resemble the Haar distribution, while higher values identify circuits with limited expressive power, for example those restricted to product-state manifolds. 

\paragraph*{$\varepsilon$-truncation.} For large $N$, the Haar fidelity density $p_{\mathrm{Haar}}(F)$ concentrates near $F=0$ with exponential decrease (Eq. \ref{eq:fid_distr}), so fixed uniform binning can allocate negligible Haar mass to the high-fidelity tail, which in turn leads to numerical instability in the KL divergence and can assign low KL divergence to highly restricted circuit families (see Appendix \ref{app:expr}). To mitigate this effect, we modify the discretization $\mathcal{B}=\{b_0,\ldots,b_B\}$ with $b_0=0$ and $b_B=1$ by choosing the last bin edge $b_{B-1}$ such that the Haar probability mass of the final bin is at least $\varepsilon$:
\begin{equation}
    \operatorname{Pr}_{\mathrm{Haar}}\!\bigl(F\in[b_{B-1},1]\bigr)
    \;=\;
    \int_{b_{B-1}}^{1} (N-1)(1-F)^{N-2}\,dF
    \;\ge\; \varepsilon .
\end{equation}
Equivalently, this yields
\begin{equation}
    b_{B-1} \le 1-\varepsilon^{\frac{1}{N-1}} .
\end{equation}

\paragraph*{Log-scaled expressibility loss.}
To reduce the influence of outliers and to normalize across problem instances, we use the log-scaled loss
\begin{equation}
\label{eq:expr_loss}
    \mathcal{L}_{\mathrm{Expr}} \;=\;
    \max\left\{
        \frac{\log\!\left( \tfrac{\mathrm{Expr}}{\tau_{\mathrm{Expr}}} \right)}
             {\log\!\left( \tfrac{\mathrm{Expr}_{\max}}{\tau_{\mathrm{Expr}}} \right)},
        \,0
    \right\},
\end{equation}
where $\tau_{\mathrm{Expr}}$ is a chosen threshold and $\mathrm{Expr}_{\max}$ denotes an upper reference value used for normalization. 

\subsection{Trainability}
\label{sec:trainability}
Highly expressible circuits often suffer from barren plateaus \cite{ExprVSBP_Holmes}, rendering optimization ineffective despite good coverage. The aim of this metric is to diagnose whether a given parameterized circuit $U(\boldsymbol\theta)$ acting on a fiducial state $\rho$ and cost function
\begin{equation}
\label{eq:cost_trace}
C(\boldsymbol\theta) \;:=\; \mathrm{Tr}\!\big[\rho(\boldsymbol\theta)\,O\big],
\qquad
\rho(\boldsymbol\theta) \;:=\; U(\boldsymbol\theta)\,\rho\,U^\dagger(\boldsymbol\theta),
\end{equation}
(where $O$ is a Hermitian observable) is prone to exhibit barren plateaus, i.e., regimes in which the optimization landscape becomes exponentially flat as the system size increases. Following \cite{Cerezo_2021, Qi2023BarrenPlateausQNN, Cunningham2025MitigatingBarrenPlateaus, Larocca2022diagnosingbarren} we define the barren plateau as
\begin{equation}
\label{eq:bp}
\operatorname{Var}_{\boldsymbol{\theta}}\!\left[\partial_\mu C(\boldsymbol{\theta})\right]
\le F(n),
\qquad
F(n)\in\mathcal{O}\!\left(\frac{1}{c^n}\right),
\qquad c>1 .
\end{equation}

Barren plateaus have been shown for circuits forming 2-designs \cite{McClean2018BarrenPlateaus}, while the expected gradient has been shown to vanish for 1-designs \cite{McClean2018BarrenPlateaus} and for common PQC ansätze \cite{ExprVSBP_Holmes}. We show that the barren plateau problem is undecidable in Theorem \ref{thm:bp_undecidable}.


\begin{theorem}
    \label{thm:bp_undecidable}
    No algorithm can decide in general whether a PQC family $\mathcal F = \{(\mathrm{PQC}_n,O_n):n\in\mathbb N\}$, consisting of an \(n\)-qubit PQC and its cost observable $O_n$ for every $n$, satisfies the barren-plateau criterion of Eq. \ref{eq:bp}. 
\end{theorem}

To prove Theorem \ref{thm:bp_undecidable}, we reduce the complement of the halting problem to the barren-plateau decision problem. Concretely, for a given Turing machine $M$, we construct a PQC family $\mathcal F(M)$ and show that deciding whether $\mathcal F(M)$ has a barren plateau is equivalent to solving the halting problem for $M$. The idea is that if $M$ halts, the variance becomes constant (so $\mathcal F(M)$ has no barren plateau), but if $M$ does not halt, the variance vanishes exponentially and $\mathcal F(M)$ has a barren plateau.

\begin{proof}
    Given a Turing machine \(M\) and input \(x\), construct a circuit generator that simulates \(M(x)\) for \(n\) steps. Starting from \(\lvert 0\rangle^{\otimes n}\), it applies \(R_y(\theta)\) to the first qubit. If the simulation has halted, it leaves the n-1 remaining qubits in \(\lvert 0\rangle\). Otherwise, it applies \(R_y(\pi/4)\) to each of them. For the observable \(O_n=Z^{\otimes n}\), the product structure gives \(C_n(\theta)=a_n\cos\theta\), where \(a_n=1\) in the halted case and \(a_n=2^{-(n-1)/2}\) otherwise. Hence \(\operatorname{Var}_{\theta}[\partial_\theta C_n]=a_n^2/2\). If \(M(x)\) never halts, this variance equals \(2^{-n}\) for every \(n\). If \(M(x)\) halts after \(T\) steps, the variance equals \(1/2\) for every \(n\ge T\). The generator terminates for every \(n\), yet its family satisfies Eq. \ref{eq:bp} exactly when \(M(x)\) never halts. A decision algorithm for Eq. \ref{eq:bp} would therefore decide the halting problem.
\end{proof}

However, since this definition is asymptotic in the number of qubits $n$, it is primarily suited to statements about scaling. In this work, we want to evaluate candidate ans\"atze at a fixed system size $n$. We therefore adopt an \emph{operational} notion of barren plateaus: for fixed $n$ we declare a circuit effectively non-trainable if the typical gradient fluctuations fall below a threshold $\tau_{\mathrm{BP}}\in\mathbb{R}^+$, i.e.,
\begin{equation}
    \operatorname{Var}_{\boldsymbol\theta}\!\left[\partial_\mu C({\boldsymbol\theta})\right]\leq \tau_{\mathrm{BP}}.
\end{equation}
This criterion raises two practical questions: 
\begin{enumerate}
\renewcommand{\labelenumi}{(\roman{enumi})}
    \item How to set $\tau_{\mathrm{BP}}$
    \item How to estimate $\operatorname{Var}_{\boldsymbol\theta}[\partial_\mu C({\boldsymbol\theta})]$
\end{enumerate}

To address (i), we assume that the algorithm employing the PQC is executed on noisy intermediate-scale quantum (NISQ) devices. If the gradient fluctuations are smaller than the estimation noise floor, optimization can become ineffective because updates are dominated by noise rather than signal. This regime is consistent with noise-induced barren plateaus for noisy VQAs \cite{Wang2021NoiseInducedBarrenPlateaus, Schumann2024EmergenceNoiseInducedBarrenPlateaus}. In general, we assume that the following inequality holds for (non-trivial) PQCs on NISQ devices when a barren plateau is present:
\begin{equation}
     \operatorname{Var}_{\boldsymbol\theta}[\partial_\mu C({\boldsymbol\theta})] \leq \tau_{\mathrm{BP}} \leq \operatorname{Pr}_\mathrm{PQC}(\mathrm{err}),
\end{equation}
where $\operatorname{Pr}_\mathrm{PQC}(\mathrm{err})$ is the probability that an error occurs for the PQC (see Appendix \ref{app:train_err}). Motivated by signal-to-noise analyses of gradient-based learning in stochastic optimization \cite{NIPS2008_8df707a9,mccandlish2018empiricalmodellargebatchtraining}, we therefore deem a PQC trainable if its gradient-to-noise ratio exceeds the threshold $\tau_{\mathrm{BP}}$:
\begin{equation}
    \frac{\operatorname{Var}_{\boldsymbol\theta}[\partial_\mu C({\boldsymbol\theta})]}{\operatorname{Pr}_\mathrm{PQC}(\mathrm{err})}>\tau_{\mathrm{BP}}.
\end{equation}
In other words, if the ratio exceeds $\tau_{\mathrm{BP}}$, the circuit is expected to avoid noise-dominated optimization and the barren-plateau regime. We emphasize that noise does not necessarily imply poor optimization. Stochasticity can improve training dynamics in classical learning \cite{pmlr-v119-smith20a} and has been shown to be beneficial for avoiding strict saddle points in variational quantum algorithms \cite{NoiseOpt2025Junyu}.

Secondly, to address (ii), we estimate $\operatorname{Var}_{\boldsymbol\theta}[\partial_\mu C({\boldsymbol\theta})]$ from finitely many gradient samples using a concentration inequality for self-bounding random variables \cite[Thm.~7]{maurer2009empirical} (derived from \cite[Thm.~13]{Maurer2006ConcentrationInequalities}). This requires uniformly bounded samples. We therefore bound the sampled gradients via \emph{gradient clipping}, a standard technique in machine learning \cite{pmlr-v28-pascanu13,Goodfellow-et-al-2016} and also a key ingredient in proving the differential privacy guarantee for SGD \cite{10.1145/2976749.2978318}. Beyond its original use as a stabilization heuristic (e.g., to mitigate exploding gradients), clipping has also been observed to improve optimization behavior, enabling stable traversal of non-smooth regions and, in some settings, faster convergence \cite{zhang2019gradient,pmlr-v202-koloskova23a}.

Concretely, we enforce bounded gradient samples via element-wise clipping \cite{pmlr-v28-pascanu13} so that
$X_i\in[L,U]$, which matches the boundedness assumption in Theorem~\ref{thm:main}. In our experiments we use the
normalized choice $[L,U]=[-1,1]$ for simplicity. We note that alternative choices, such as $\ell_2$-norm clipping, are
also common and left for future work \cite{Oh2025DiscoveringSOTARLAlgorithms,10.1145/2976749.2978318}.

\begin{theorem}
\label{thm:main}
Let $L,U\in\mathbb R$ with $L<U$ and $R:=U-L$, $\mathbf{X}=(X_1,\dots,X_m)$ be independent random gradient samples with $X_i\in[L,U]$, and $m\ge 3$. Define $Z(\mathbf{X}):=\frac{m}{R^2}s_m^2$, where the sample variance is    
\[
s_m^2:=\frac{1}{m-1}\sum_{i=1}^m (X_i-\bar X)^2,
\qquad
\bar X:=\frac{1}{m}\sum_{i=1}^m X_i.
\]
 For $k\in\{1,\dots,m\}$ and $y\in[L,U]$, let $\mathbf X_{y,k}$ be obtained from $\mathbf X$ by replacing $X_k$ with $y$, and set
\[
\Delta_k(\mathbf X):=Z(\mathbf X)-\inf_{y\in[L,U]} Z(\mathbf X_{y,k}).
\]
Then almost surely
\begin{equation}
    \Delta_k(\mathbf X)\le 1 \quad \forall k,
\end{equation}
\begin{equation}
\sum_{k=1}^m \Delta_k(\mathbf X)^2 \le a\, Z(\mathbf X)
\end{equation}
with $a=\frac{m}{m-1}$. Therefore, for all $\delta>0$:

\begin{equation}
    \operatorname{Pr}\left(\bigl|\sqrt{s_m^2}-\sqrt{\mathbb E[s_m^2]}\bigr|
\le
\sqrt{\frac{2R^2\,\ln(2/\delta)}{(m-1)}}\right)\geq 1-\delta.
\end{equation}
\end{theorem}
\begin{proof}
    See Appendix \ref{app:main_proof}.
\end{proof}

To estimate $\operatorname{Var}(X)$ from Theorem~\ref{thm:main}, we use the unbiased sample variance $s_m^2$ (with prefactor $1/(m-1)$), so that $\mathbb E[s_m^2]=\operatorname{Var}(X)$ for i.i.d.\ samples. We report concentration in terms of
standard deviation since the bound controls $|\sqrt{s_m^2}-\sqrt{\mathbb E[s_m^2]}|\leq \varepsilon$. On the event that this bound holds, one obtains either an upper bound on the variance deviation via the identity
\begin{equation}
\bigl|s_m^2-\mathbb E[s_m^2]\bigr|
=
\bigl|\sqrt{s_m^2}-\sqrt{\mathbb E[s_m^2]}\bigr|
\left(\sqrt{s_m^2}+\sqrt{\mathbb E[s_m^2]}\right),
\end{equation}
or the conclusion that no barren plateau exists (Lemma \ref{lem:sd_implies_var}).

\begin{lemma}
\label{lem:sd_implies_var}
Let $s_m^2,\mathbb{E}[s_m^2]\in[0,1]$ be fixed and $\varepsilon\ll 1$ with $\bigl|\sqrt{s_m^2}-\sqrt{\mathbb E[s_m^2]}\bigr|\leq \varepsilon$ then either 
\begin{equation}
        \bigl|s_m^2-\mathbb E[s_m^2]\bigr|
\leq
\bigl|\sqrt{s_m^2}-\sqrt{\mathbb E[s_m^2]}\bigr|    
\end{equation}
or
\begin{equation}
    \mathbb E[s_m^2] >
\left(\frac{1-\varepsilon}{2}\right)^2
\end{equation}
holds true.
\begin{proof}
Set $a:=\sqrt{\mathbb{E}[s_m^2]}$ and $b:=\sqrt{s_m^2}$. If $a+b\le 1$, then multiplying by $|b-a|$ yields
$|(a+b)(b-a)|\le |b-a|$, and since $(a+b)(b-a)=b^2-a^2=s_m^2-\mathbb{E}[s_m^2]$, we obtain
$\bigl|s_m^2-\mathbb{E}[s_m^2]\bigr|\le |b-a|$. If instead $a+b>1$, then $b>1-a$.
Using $|b-a|\le \varepsilon$ gives $b\le a+\varepsilon$, hence $1-a<b\le a+\varepsilon$, which implies
$2a+\varepsilon>1$, i.e.\ $a>(1-\varepsilon)/2$. Therefore,
$\mathbb{E}[s_m^2]=a^2>\bigl((1-\varepsilon)/2\bigr)^2$.
\end{proof}
\end{lemma}

Sharper (problem-dependent) bounds may be obtained by working directly with the variance in the small-variance regime or by incorporating assumptions on the sampling distribution. Nevertheless, our bound is dimension-free, i.e., the required number of samples does not depend on the number of qubits, and provides us with an accurate estimate of the variance, especially when the variance decreases to zero. 

\paragraph*{Trainability loss.} For simplicity, we define the trainability loss by averaging over the sampled initial states in $\mathcal{S}_0$ and over the $|\boldsymbol{\theta}|$ trainable parameters
\begin{equation}
\label{eq:train_loss}
\mathcal{L}_{\mathrm{Train}}
=
\max\left(
\frac{
\tau_{\mathrm{BP}}
-
\frac{1}{|\mathcal{S}_0||\boldsymbol{\theta}|}
\sum_{\ket{\psi_0}\in \mathcal{S}_0}
\sum_{i=1}^{|\boldsymbol{\theta}|}
\frac{s_{m,\theta_i,\ket{\psi_0}}^2}{\operatorname{Pr}_\mathrm{PQC}(\mathrm{err})}
}{
\tau_{\mathrm{BP}}
},
0
\right).
\end{equation}
where \(s_{m,\theta_i,\ket{\psi_0}}^2\) denotes the sample variance of \(\partial_{\theta_i} C(\boldsymbol{\theta})\) over the \(m\) sampled parameter values, for the fixed initial state \(\ket{\psi_0}\).

\subsection{Entanglement}
Trainability does not preclude classical tractability; indeed, several works note that well-trainable ans\"atze can often admit efficient classical simulation~\cite{Cerezo2025AbsenceBarrenPlateausSimulability,Leone2024practicalusefulness,gilfuster2025relationtrainabilitydequantizationvariational}. We therefore incorporate entanglement as a complementary indicator of nontrivial quantum correlations, which is commonly associated with increased classical simulation cost~\cite{EntglSimulable2003Guifr,EntglSimulable2003Josza,EntglSimulable2010Eisert}. We emphasize, however, that entanglement alone is not sufficient to guarantee classical hardness~\cite{gottesman1998theorem,ManyBody2022Liu}. To quantify entanglement, we employ the Meyer-Wallach measure \cite{meyer2002global}. For a pure $n$-qubit state $\ket{\psi}$, it is defined as
\begin{equation}
Q(\ket{\psi}) = \frac{4}{n}\sum_{j=1}^{n} D\!\left(\iota_j(0)\ket{\psi},\;\iota_j(1)\ket{\psi}\right),
\end{equation}
where $\iota_j(b)$ denotes the projection of $\ket{\psi}$ onto the subspace in which qubit $j$ is in the computational basis state $\ket{b}$, yielding a vector in $(\mathbb{C}^2)^{\otimes (n-1)}$, and
\begin{equation}
D(u,v) = \sum_{x<y} \left| u_x v_y - u_y v_x \right|^2
\end{equation}
is the squared norm of the wedge product of the vectors $u$ and $v$. The measure is invariant under local unitary transformations and vanishes if and only if $\ket{\psi}$ is fully separable.

Equivalently, the Meyer--Wallach entanglement can be expressed in terms of the single-qubit reduced density matrices $\rho_j = \mathrm{Tr}_{\bar{j}}(\ket{\psi}\!\bra{\psi})$ as
\begin{equation}
Q(\ket{\psi}) = \frac{2}{n} \sum_{j=1}^{n} \left(1 -\mathrm{Tr}\left[\rho_j^2\right]\right),
\end{equation}
which highlights its interpretation as the average linear entropy of the one-qubit subsystems. With this normalization, $Q(\ket{\psi}) \in [0,1]$, attaining its maximal value for globally entangled states such as GHZ states. We estimate the entanglement for a given PQC with
\begin{equation} 
    \text{Ent}=\frac{1}{|\mathcal{S}_0|\cdot|\Theta|}\sum_{\ket{\psi_0}\in \mathcal{S}_0}\sum_{{\boldsymbol\theta}\in \Theta}Q\!\left(U(\boldsymbol\theta)\ket{\psi_0}\right).
\end{equation}
\paragraph*{Entanglement loss.} The corresponding loss function is then given by
\begin{equation}
\label{eq:entgl_loss}
    \mathcal{L}_{\mathrm{Ent}}
    =
    \max\left(
        \frac{\tau_{\mathrm{Ent}} - \mathrm{Ent}}{\tau_{\mathrm{Ent}}},
        0
    \right),
\end{equation}
where $\tau_{\mathrm{Ent}}$ denotes the entanglement threshold.

\subsection{Complexity}
Finally, since these properties must be achieved under near-term hardware and optimization budgets, we include a simple complexity proxy based on parameters, depth, and gate count.

Circuit complexity is not uniquely defined and depends on both the hardware backend and the optimization routine. 
Since we focus on variational quantum algorithms, the number of trainable parameters ($|\boldsymbol{\theta}|$) is particularly relevant because it directly impacts the cost of gradient-based optimization (e.g., parameter-shift \cite{Parameter2018Mitarai,Parameter2019Schuld}). 
Moreover, to ensure that candidate PQCs remain feasible on near-term devices, we also account for circuit depth ($D$) and gate count ($G$) as proxies for noise exposure.

\paragraph*{Complexity loss.} We therefore define a simple, normalized complexity loss as
\begin{equation}
\label{eq:cmplx_loss}
\mathcal{L}_{\mathrm{Cmplx}}
\;=\;
\frac{|\boldsymbol\theta| + D + G}
{|\boldsymbol\theta|_{\max} + D_{\max} + G_{\max}},
\end{equation}
where $(|\boldsymbol\theta|_{\max},D_{\max},G_{\max})$ are reference maxima determined by the considered search space. 
This yields $\mathcal{L}_{\mathrm{Cmplx}}\in[0,1]$, with smaller values indicating less complex circuits.

\section{Experiments}
\label{sec:experiments}
In this section, we evaluate the proposed ansatz-search framework in a range of single- and multi-objective settings, and additionally demonstrate hardware-constrained search on the 5-qubit \emph{IQM Spark} \cite{Ronkko2024OnPremSuperconductingQC} JIQCER-5 located at Jülich Supercomputing Centre, Germany. Our goal is not to advocate a single optimization backend, but to provide a flexible metric-driven formulation that can be combined with different search strategies (e.g., Bayesian optimization, evolutionary methods, or reinforcement learning). In this work, we use Bayesian optimization \cite{NIPS2012_05311655}, motivated by its sample efficiency in expensive black-box settings and its ability to incorporate prior beliefs and uncertainty estimates over the objective landscape.

Unless stated otherwise, we run the optimizer for $10{,}000$ trials and compare against the benchmark circuit collection of \cite{Sim2019}. To prevent degenerate architectures, we require every generated PQC to contain at least one trainable parameter and to satisfy a minimal connectivity constraint, namely that each qubit must be connected to the rest of the circuit via at least one path. All solutions reported below are re-evaluated on an independent test set of random draws from the same distribution, with the same sample size as in training.

For metric estimation, we follow \cite{Sim2019} and use $5{,}000$ Monte Carlo samples for expressibility and entanglement. For trainability (Section \ref{sec:trainability}), we set target accuracy $\varepsilon=0.05$ and confidence level $1-\delta=0.95$, yielding the sufficient sample size
\begin{equation}
m \;\ge\; \frac{8\ln(2/\delta)}{\varepsilon^2}+1
\;=\; \frac{8\ln 40}{0.05^2}+1 \;\approx\; 11{,}806,
\end{equation}
which guarantees $|\sqrt{s_m^2}-\sqrt{\mathbb E[s_m^2]}|\le 0.05$ with probability at least $0.95$.

Additional setup information and details on the results presented in this section are provided in Appendix~\ref{app:setup_details} and Appendix~\ref{app:extended_results}.

\subsection{Single-objective optimization}
\label{sec:single_objective}

To validate the proposed framework, we first consider \emph{single-objective} searches. Concretely, we optimize one of the metric losses introduced in Section~\ref{sec:theory}: expressibility (Eq.~\ref{eq:expr_loss}), trainability (Eq.~\ref{eq:train_loss}), or entanglement (Eq.~\ref{eq:entgl_loss}). Circuit complexity is not treated as a standalone primary objective in this section. Instead, we include it as a strictly secondary criterion via a hierarchical cost,
\begin{equation}
\label{eq:lexicographic_cost_single}
\mathcal{L} =
\begin{cases}
\mathcal{L}_{\mathrm{obj}} + 1, & \mathcal{L}_{\mathrm{obj}} > 0,\\[4pt]
\mathcal{L}_{\mathrm{Cmplx}}, & \text{otherwise},
\end{cases}
\end{equation}
where $\mathcal{L}_{\mathrm{obj}}\in\{\mathcal{L}_{\mathrm{Expr}}, \mathcal{L}_{\mathrm{Train}}, \mathcal{L}_{\mathrm{Ent}}\}$.
This construction enforces a clear search priority. As long as the circuit violates the target threshold (i.e., $\mathcal{L}_{\mathrm{obj}}>0$), the optimizer is driven exclusively by $\mathcal{L}_{\mathrm{obj}}$. Only once the threshold is satisfied does the search minimize complexity.

We emphasize that Eq.~\eqref{eq:lexicographic_cost_single} is merely one convenient choice. Any objective function compatible with the search procedure could be used. In particular, hierarchical objectives are typically nonsmooth at the switching boundary and may therefore require tailored handling in gradient-based search algorithms \cite{Clarke1990,Burke2005}. We set the search space and the metric thresholds based on the best-performing benchmark PQC for the corresponding metric (details in Appendix~\ref{app:setup_details} Table \ref{tab:config}). 

\begin{figure*}[t]
    \centering
    \subfloat[]{%
        \includegraphics[width=0.47\textwidth]{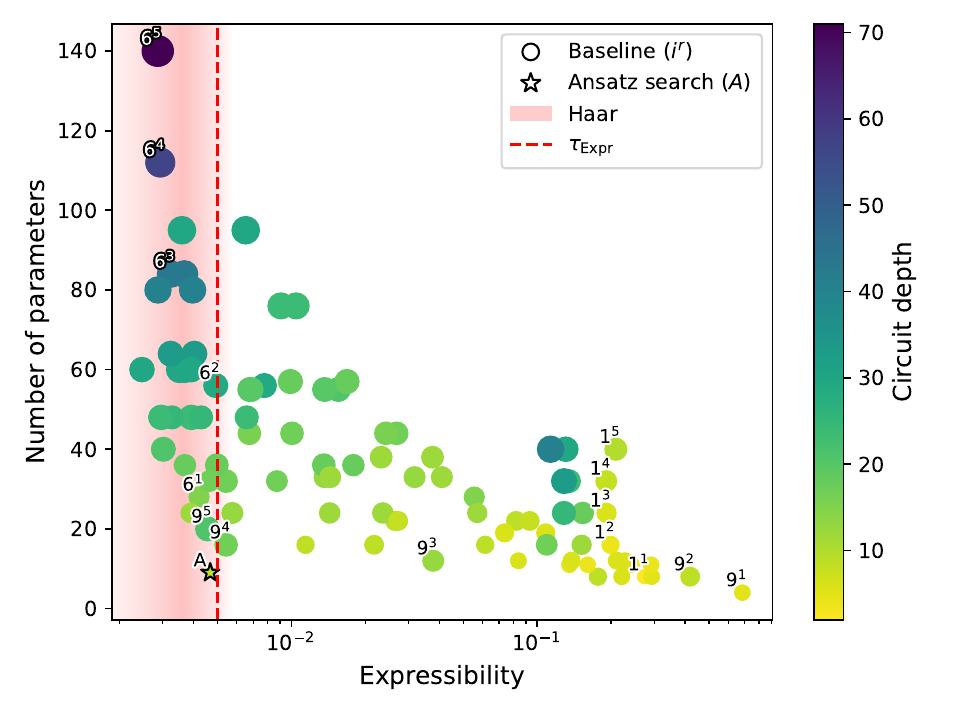}
    }\hfill
    \subfloat[]{%
        \includegraphics[width=0.47\textwidth]{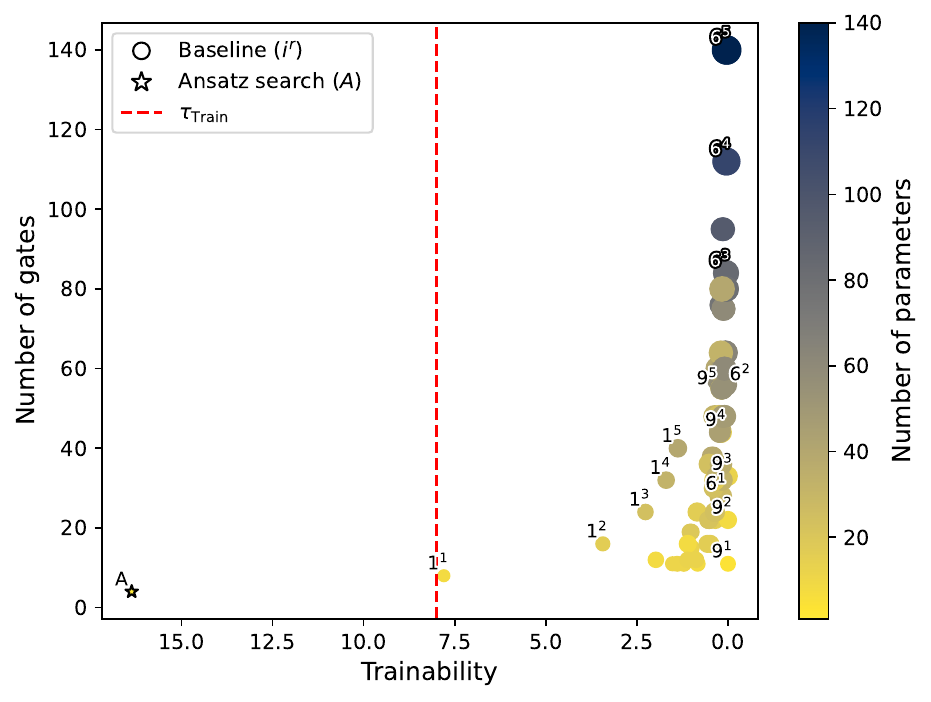}
    }\\[0.6em]
    \subfloat[]{%
        \includegraphics[width=0.47\textwidth]{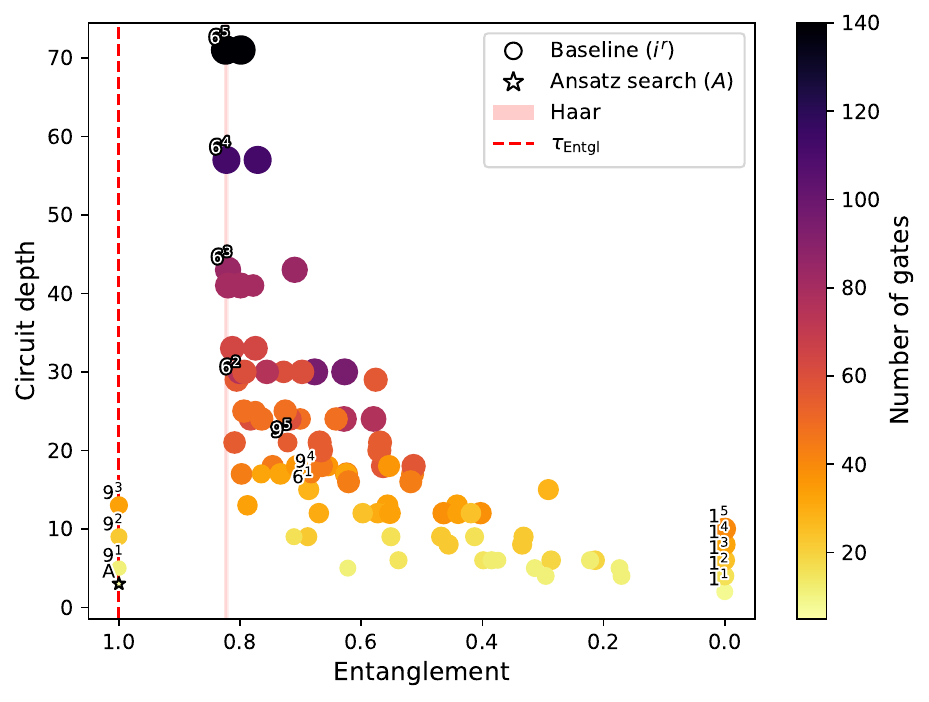}
    }\hfill
    \subfloat[]{%
        \includegraphics[width=0.47\textwidth]{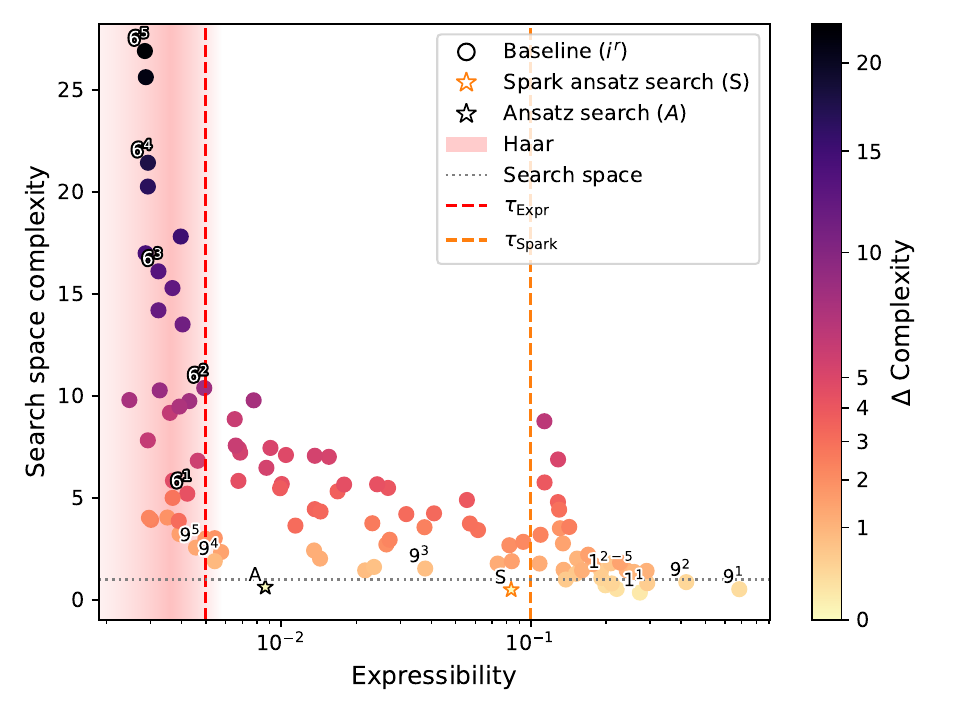} 
    } 
    \caption{Single-metric ansatz search. (a) Expressibility vs number of parameters ($|\boldsymbol{\theta}|$). (b) Trainability vs number of gates ($G$). (c) Entanglement vs circuit depth ($D$). (d) Expressibility vs search space complexity. Markers denote quantum circuit ansätze (see legend), comparing the circuit found by the ansatz search ($A$) to the benchmark circuits $i^r$ from \cite{Sim2019}, where $r$ is the number of repetitions. Where applicable, the horizontal band labeled \textit{Haar} indicates a $\pm 3\sigma$ interval estimated from $5{,}000$ (pair) Haar-random states for the corresponding metric. Marker size encodes a proxy for circuit complexity: in (a) $\propto \sqrt{G}$, in (b) $\propto \sqrt{D}$ (depth), and in (c) $\propto \sqrt{|\boldsymbol\theta|}$. Panel (d) shows the hardware-constrained expressibility search on the 5-qubit IQM Spark device (Spark ansatz search) using its native gate set and star connectivity.}
    \label{fig:single_obj}
\end{figure*} 

Figure~\ref{fig:single_obj} shows that the framework reliably finds ansätze that satisfy a single metric threshold. For expressibility, the discovered circuits are substantially more compact than the benchmark circuits while remaining Haar-like within the expected finite-sample fluctuations. Notably, the optimized solutions also reduce the number of trainable parameters, a practically important outcome because parameter count directly impacts the cost of gradient-based training. Analogous behavior is observed for entanglement and trainability. We find circuits that meet the respective thresholds with less complexity than the strongest benchmark circuit for that metric. 

In addition to simulations, we performed an expressibility search on the 5-qubit \emph{IQM Spark} \cite{Ronkko2024OnPremSuperconductingQC} JIQCER-5 quantum processor (Figure~\ref{fig:single_obj}d). To exploit hardware constraints, the search space was restricted to the native gate set and topology of the device, i.e., a star connectivity graph with parameterized single-qubit $R_z$ rotations and two-qubit $CZ$ gates. We executed 50 optimizer iterations on hardware. Nevertheless, the discovered ansatz improves over all benchmark circuits within the considered search space. This improvement is achieved despite the low circuit fidelities reported for random circuits on superconducting hardware \cite{Arute2019QuantumSupremacy}.

To test whether this hardware-constrained setup can still yield near-Haar expressibility, we additionally ran a 10,000-trial simulation using the same gate set and topology. In contrast to Figure~\ref{fig:single_obj}a we do not observe convergence to a Haar-like expressibility baseline in the noiseless simulation. We attribute the deviation from the expressibility baseline to restricting the circuit family to parameterized $R_z(\theta)$ gates only, i.e., we disallow fixed-angle rotations (constant $R_z$ offsets).  

Importantly, incorporating such fixed-angle $R_z$ rotations would not introduce additional transpilation overhead. By contrast, compiling the benchmark templates to the hardware-native basis incurs substantial complexity overhead by up to a factor of 20 relative to their original representation. This complexity efficiency is another reason why using an ansatz search might be favorable for executing a VQA on a quantum processor.

More broadly, we observe that highly complex circuits tend to become more Haar-like, consistent with deep random circuits approaching approximate $t$-design behavior \cite{Brandao2016LocalRandomCircuitsDesigns,Haferkamp2022randomquantum}. For example, the most expressive circuit at zero repetition (circuit 6) can remain near-Haar in expressibility for small repetitions and can approach the Haar entanglement baseline, whereas circuits that start highly entangling (circuit 9) may trade entanglement for expressibility as repetitions increase. For trainability, circuits degrade with repetitions, such as circuit 1 (see Figure~\ref{fig:single_obj}b). Interestingly, circuit 1, consisting only of single-qubit gates \cite{Sim2019}, exhibits the largest trainability while simultaneously being poorly expressive and having no entanglement. This shows why single-metric optima may not align with good ansätze for VQAs.

\begin{figure*}[t]
    \centering
    \subfloat[]{%
        \includegraphics[width=0.45\textwidth]{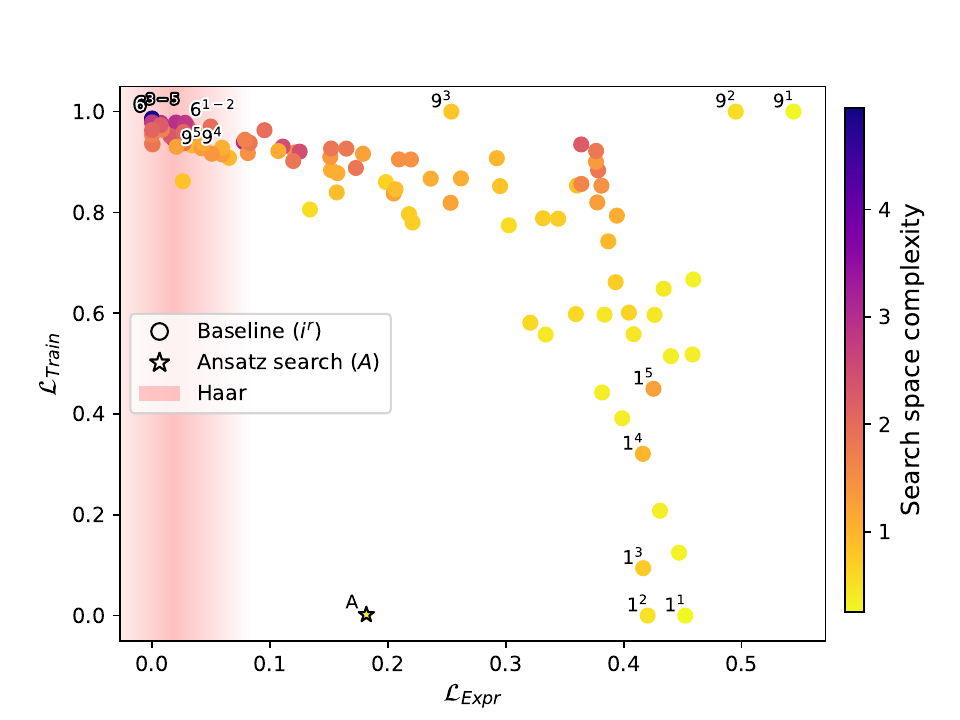}
    }\hfill
    \subfloat[]{%
        \includegraphics[width=0.45\textwidth]{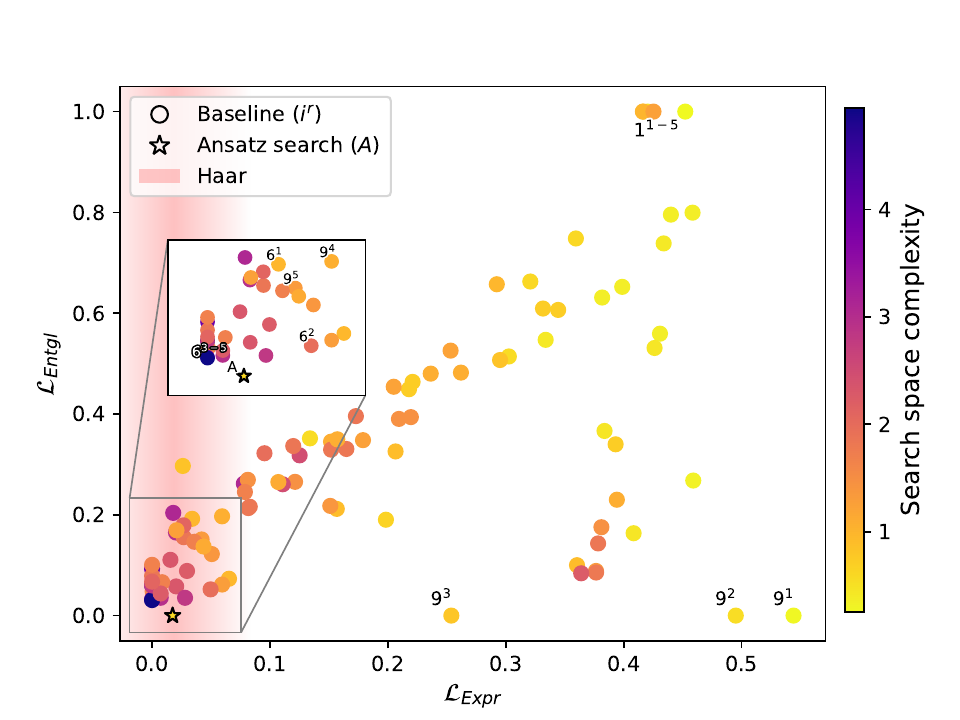}
    }
    \caption{Multi-objective ansatz search. Each point corresponds to a circuit $i^r$ from the benchmark families of \cite{Sim2019} (across repetitions $r$) and the circuit discovered by our search ($A$). The horizontal band labeled \textit{Haar} indicates a $\pm 3\sigma$ interval estimated from 5{,}000 pair Haar-random states for the expressibility metric. (a) Shows the expressibility vs.\ trainability and (b) expressibility vs.\ entanglement. For visualization purposes, we set the expressibility threshold to $0.003$, the trainability threshold to $2.5$, and the entanglement threshold to $0.85$.}    
    \label{fig:multi_obj}
\end{figure*}

This trade-off is also reflected in the circuits returned by the single-objective ansatz search (Appendix~\ref{app:single_obj}). For the trainability objective, the search similarly favors very shallow circuits, which can score highly under the trainability proxy while remaining poorly expressive and thus uninformative for VQA performance. For the entanglement objective, we find solutions that satisfy the entanglement threshold by appending parameterized gates whose optimized angles have negligible effect on the prepared state, effectively ``padding'' the circuit to meet the PQC structural constraints without improving expressibility or trainability. These artifacts underscore that single-metric optima need not correspond to practically useful ansätze and motivate interpreting PQCs jointly across metrics.

\subsection{Multi-objective optimization}
\label{sec:multi_obj}
Meaningful PQCs for practical applications typically need to satisfy multiple desirable properties simultaneously. As in the previous section, we treat circuit complexity as a secondary criterion and optimize the primary objectives until the chosen thresholds are met. Concretely, we use the hierarchical cost
\begin{equation}
\label{eq:lexicographic_cost_multi}
\mathcal{L} =
\begin{cases}
\mathcal{L}_{\mathrm{obj}_1} + \mathcal{L}_{\mathrm{obj}_2} + 1, & \mathcal{L}_{\mathrm{obj}_1} + \mathcal{L}_{\mathrm{obj}_2} > 0,\\[4pt]
\mathcal{L}_{\mathrm{Cmplx}}, & \text{otherwise},
\end{cases}
\end{equation}
where $\mathcal{L}_{\mathrm{obj}_i}\in\{\mathcal{L}_{\mathrm{Expr}}, \mathcal{L}_{\mathrm{Train}}, \mathcal{L}_{\mathrm{Ent}}\}$.
We emphasize that this is, as in Section~\ref{sec:single_objective}, a pragmatic choice. Depending on the target application, it can be beneficial to smooth the switching boundary or to enforce feasibility with a stricter aggregation (e.g., $\max(\mathcal{L}_{\mathrm{obj}_1},\mathcal{L}_{\mathrm{obj}_2})$) to prevent compensation between objectives.

We first address a central question of our paper: can one find circuits that are both expressive \emph{and} trainable? Figure~\ref{fig:multi_obj}a shows that circuits can satisfy moderate thresholds for both metrics, and that our search identifies ans\"atze in this region. At the same time, the region of simultaneously expressible and trainable circuits appears largely inaccessible for commonly used circuit families \cite{Sim2019}. This reproduces the theoretical expectation that expressibility and trainability are typically anticorrelated. Ensembles close to Haar randomness, or even approximate 2-designs, exhibit strong gradient concentration and thus barren-plateau--type behavior \cite{McClean2018BarrenPlateaus}.

A natural follow-up question is whether circuits that are expressive and trainable remain classically simulable. As one proxy for nontrivial many-body structure, we consider entanglement. Figure~\ref{fig:multi_obj}b indicates a strong positive correlation between expressibility and entanglement across the explored circuit families. In our setting, this suggests that an explicit entanglement constraint may not always be necessary, since high expressibility often coincides with high entanglement. Moreover, Figure~\ref{fig:multi_obj}b also shows that the search can identify circuits that simultaneously satisfy thresholds for expressibility and entanglement while remaining competitive with the best benchmark circuits, at lower complexity.

\begin{figure*}[t]
    \centering
    \subfloat[]{%
        \includegraphics[width=0.313\textwidth]{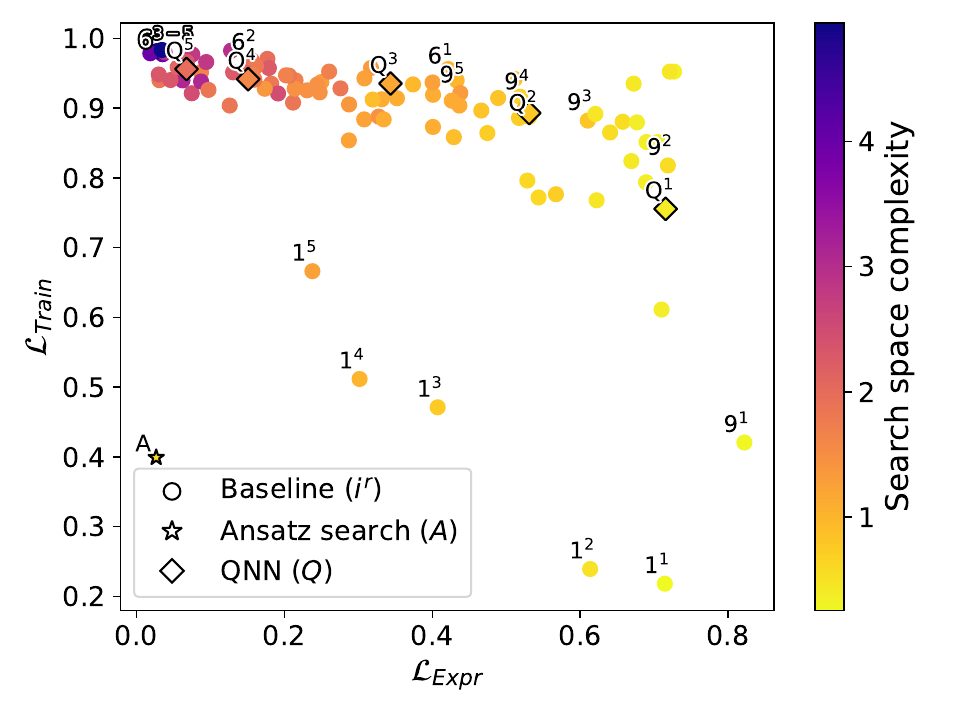}
    }\hfill
    \subfloat[]{%
        \includegraphics[width=0.333\textwidth]{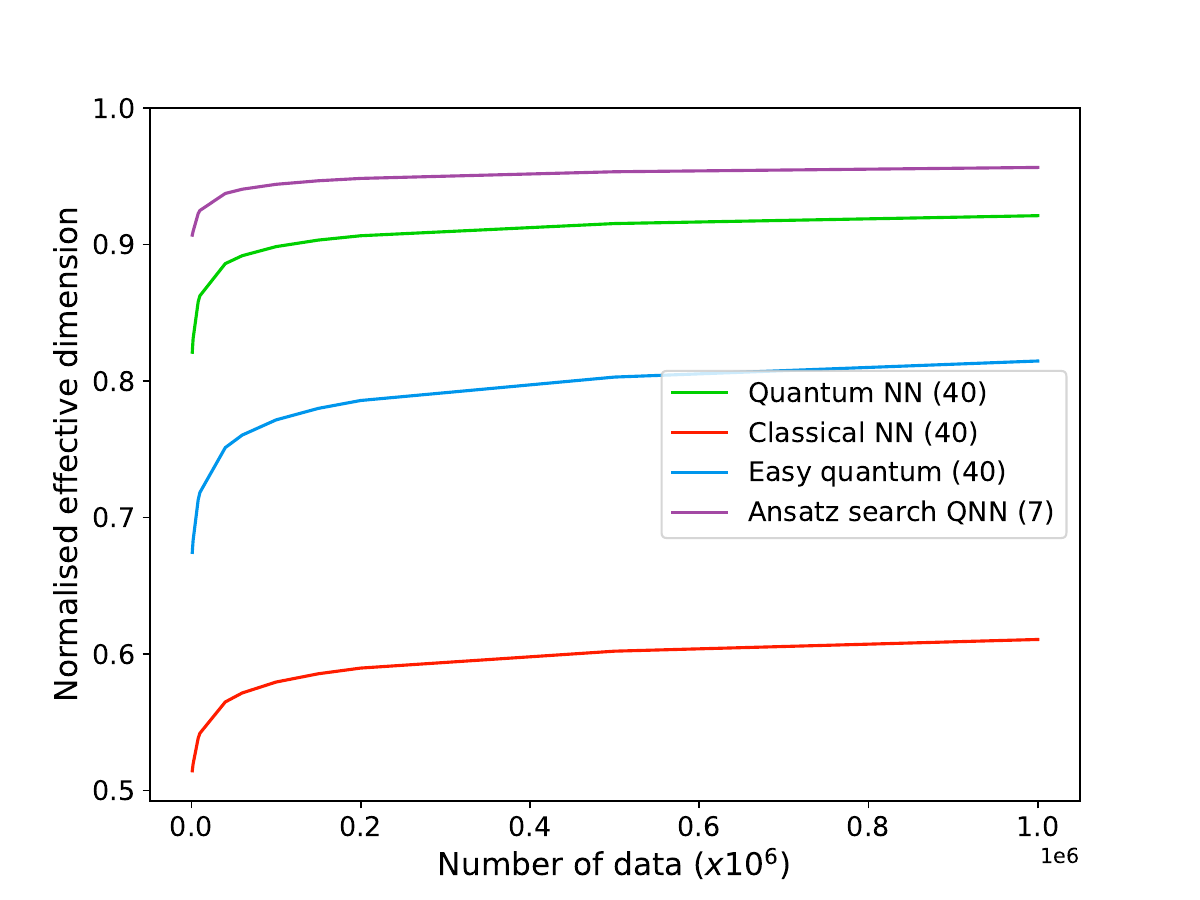}
    }\hfill
    \subfloat[]{%
        \includegraphics[width=0.333\textwidth]{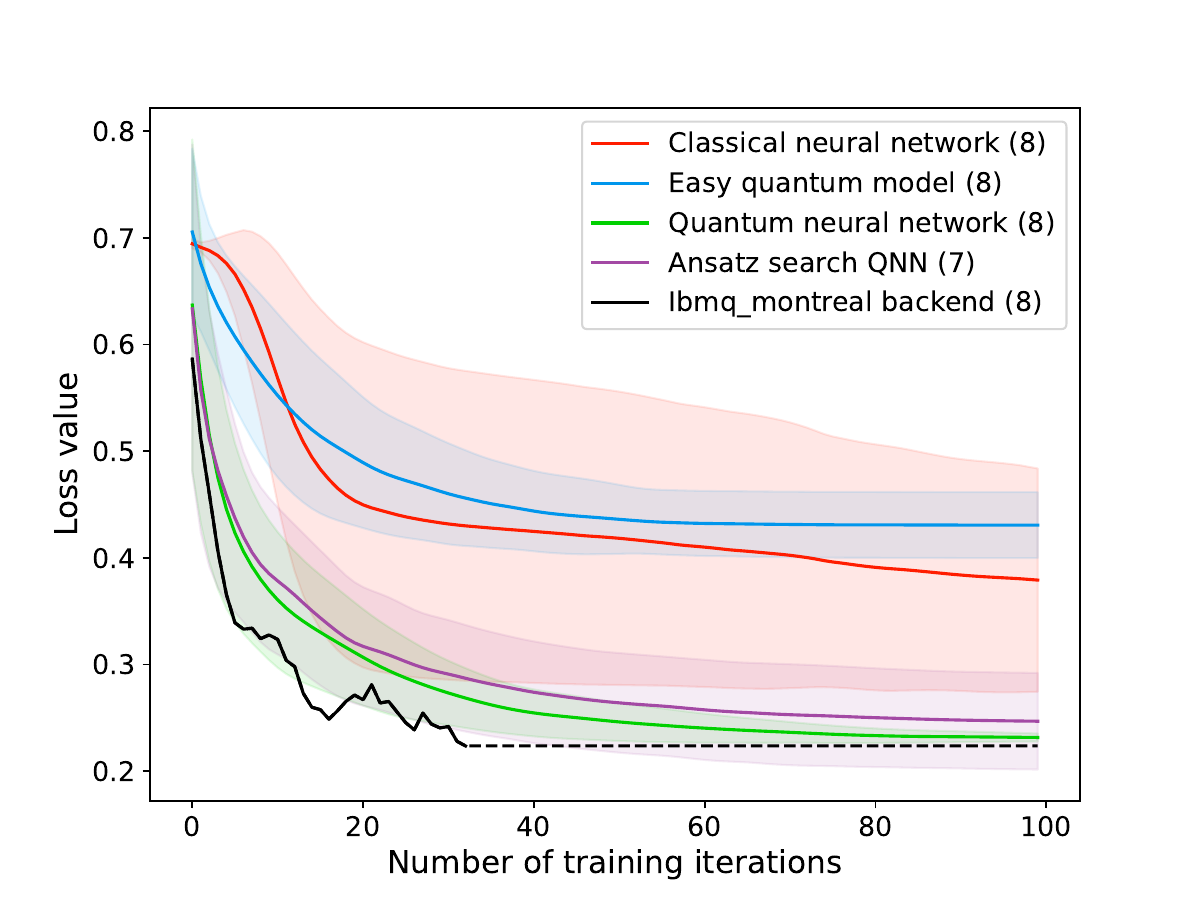} 
    }
    \caption{QNN case study following the setup of \cite{Abbas2021PowerQNN}. (a) Search outcomes in the expressibility--trainability plane for the considered search space. (b) Estimated effective dimension for the reference QNN architecture from \cite{Abbas2021PowerQNN} and for the discovered PQC. (c) Training performance on the dataset of \cite{Abbas2021PowerQNN}, using the same optimizer and protocol, comparing the reference architecture to our discovered PQC. For visualization purposes, we set the expressibility and trainability thresholds to \(0.003\) and \(2.5\), respectively.}
    \label{fig:qnn_training}
\end{figure*}

Finally, we note that the best-performing circuits found by the search are not necessarily ``good-looking'' in the sense of exhibiting simple symmetries (Appendix~\ref{app:extended_results}). Many benchmark templates are highly structured and visually appealing, yet not optimal for applications. Moreover, metric values can depend strongly on the \emph{input state} used for evaluation. In practice, a PQC used in an iterative optimization or as a repeated layer is generally \emph{not} applied to $\ket{0}^{\otimes n}$ after the first iteration/layer, so circuits assessed only on $\ket{0}^{\otimes n}$ may score differently under the states encountered during training. For instance, an $R_xR_z$ layer can explore the same set of single-qubit states as $(R_xR_z)^2$ when composed during training, while exhibiting noticeably different expressibility when evaluated only on $\ket{0}$ (for which the first $R_z$ gate has no effect). 

These observations motivate an application-dependent search that incorporates, where possible, domain knowledge. For the following sections, we explicitly note that we do not claim domain-optimal designs and that domain experts may achieve substantially better task-level performance.

\subsection{Quantum neural networks}
\label{sec:power_of_qnn}
A prominent application of parameterized quantum circuits (PQCs) is as variational models in quantum neural networks (QNNs) \cite{farhi2018classificationquantumneuralnetworks,Benedetti_2019,Beer2020TrainingDeepQNN,Abbas2021PowerQNN}. To assess whether metric-guided ansatz discovery can yield competitive QNN architectures, we adopt the experimental protocol of \cite{Abbas2021PowerQNN} and only replace the QNN circuit architecture with a PQC found by our framework.

For metric evaluation, we estimate expressibility and trainability by averaging over a small set of randomly chosen input states. Concretely, we sample five initial states, which we found sufficient for stable expressibility estimates in this setting (Appendix~\ref{app:qnn} Figure~\ref{fig:app_init}).
Following \cite{Abbas2021PowerQNN}, circuit parameters used during metric evaluation are drawn i.i.d.\ uniformly from $[-1,1]^{|\boldsymbol{\theta}|}$.
We perform a bi-objective search targeting expressibility and trainability, with circuit complexity minimized hierarchically once both metric thresholds are satisfied (Eq.~\ref{eq:lexicographic_cost_multi}). As in Figure~\ref{fig:multi_obj}a, the explored circuit families in Figure~\ref{fig:qnn_training}a again exhibit an anticorrelation trend between expressibility and trainability. The discovered ansatz again lies in the regime where this anticorrelation is observed, whereas the reference QNN exhibits relatively high trainability at one iteration but then degrades rapidly as the number of iterations increases, while its expressibility increases.

We first test the found ansatz on the \emph{effective dimension} \cite{Abbas2021PowerQNN}, an information-geometric complexity measure that estimates the size the model occupies in function (model) space at finite data resolution, derived from the Fisher information. In Figure~\ref{fig:qnn_training}b one can see, that the discovered PQC achieves a higher effective dimension than the five-layer reference QNN from \cite{Abbas2021PowerQNN}, while using substantially fewer trainable parameters ($7$ versus $40$). However, when following \cite{Abbas2021PowerQNN} for training on the first two classes of the Iris dataset, we do not observe an improvement in training performance over the reference model in this particular benchmark. One plausible explanation is that the dataset is comparatively easy, so that the additional functional capacity captured by the effective dimension does not translate into better generalization or lower training loss under the fixed optimization protocol. 

We additionally observe that the expressibility of circuit 1 increases with the number of repetitions, even though more than 2 repetitions do not enlarge the set of reachable states (Figure~\ref{fig:qnn_training}a). A plausible explanation is that our expressibility estimate, given a feature encoded initial state, by repeating the same block, can act as an additional “scrambling” step that drives the output ensemble closer to Haar-like statistics. This is consistent with results showing that (local) random circuits converge toward approximate unitary designs with increasing depth \cite{Brandao2016LocalRandomCircuitsDesigns}. We therefore interpret the observed increase in expressibility primarily as improved mixing toward Haar for the relevant input ensemble, not as evidence that more distinct states become accessible.

\begin{figure*}[t]
    \centering
    \subfloat[]{%
        \includegraphics[width=0.35\textwidth]{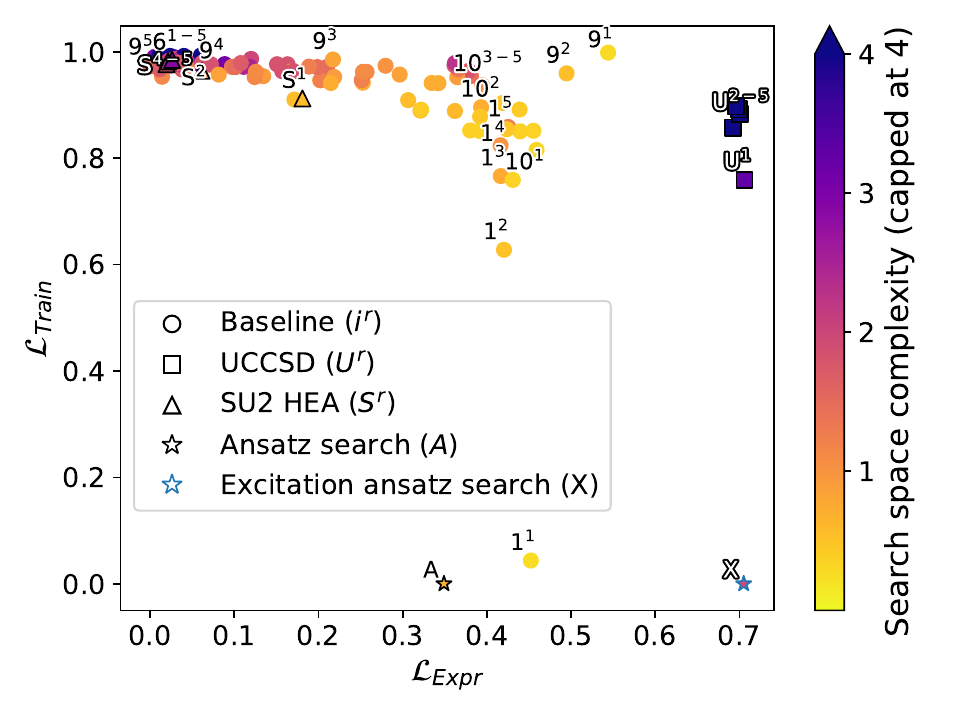}
    }\hfill
    \subfloat[]{%
        \includegraphics[width=0.31\textwidth]{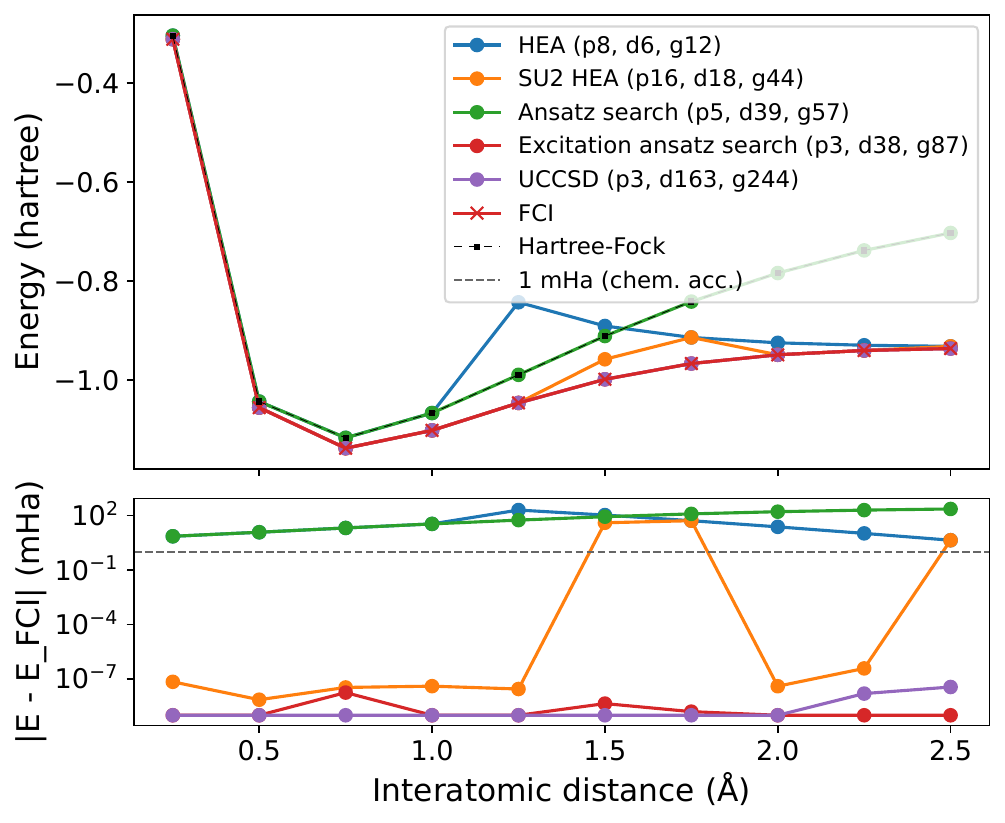}
    }\hfill
    \subfloat[]{%
        \includegraphics[width=0.31\textwidth]{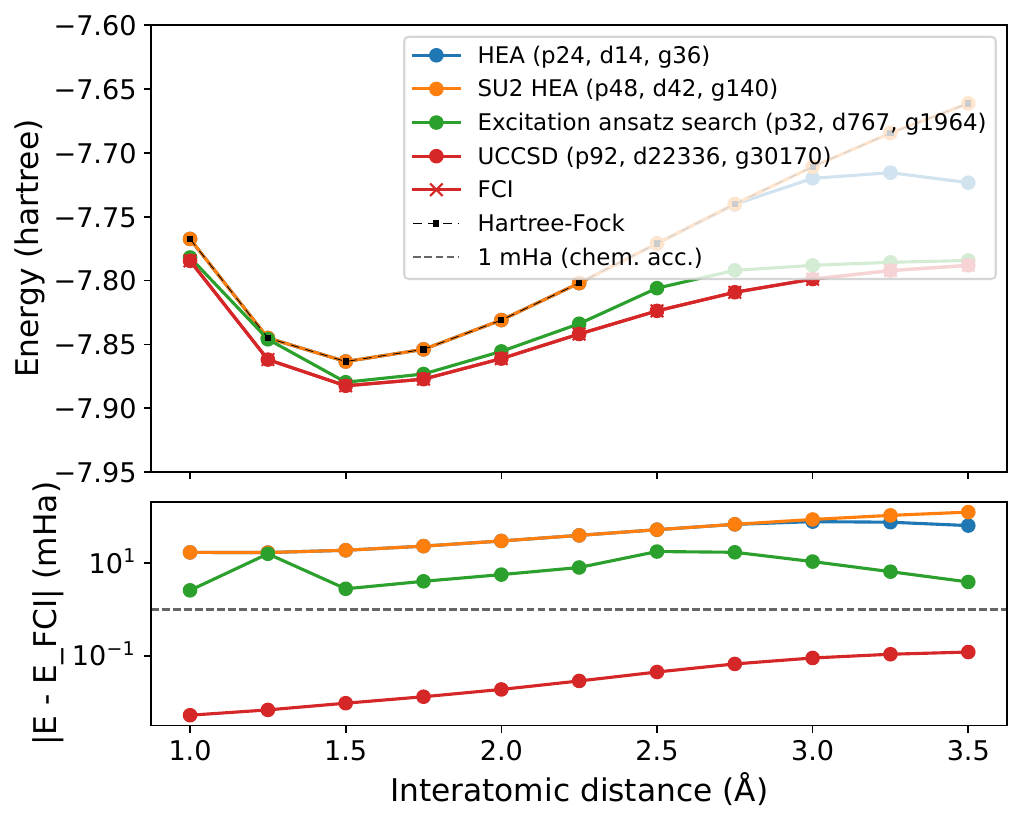}
    }
    \caption{VQE evaluation of ans\"atze discovered by metric-driven search.
(a) Metric values of the discovered ansatz and reference circuits for $\mathrm{H}_2$ at $R=0.75\,\text{\AA}$, including expressibility, trainability, and circuit complexity.
(b) VQE energy convergence for $\mathrm{H}_2$ across multiple interatomic distances using a single fixed ansatz optimized at $R=0.75\,\text{\AA}$, compared against UCCSD and hardware-efficient baselines.
(c) Proof-of-concept LiH VQE at $R=1.5\,\text{\AA}$ (100 Bayesian-optimization iterations). For visualization purposes, we set the expressibility and trainability thresholds to \(0.003\) and \(2.5\), respectively.}
    \label{fig:vqe_obj}
\end{figure*}

More broadly, this experiment highlights an important limitation: higher metric scores (or proxies such as effective dimension) do not automatically imply improved downstream performance. In addition, it remains unclear how the discovered architecture should be scaled with system size and depth. In particular, the dependence of both metric values and task-level performance on (i) the number of qubits and (ii) circuit repetitions/layers is nontrivial, and we leave a systematic scaling study to future work. Nevertheless, the experiment suggests that improved QNN architectures may be attainable, and that metric-driven design, as explored here, is a promising route toward them.

\subsection{Variational quantum eigensolver}
\label{sec:vqe}
The variational quantum eigensolver (VQE) is a hybrid quantum--classical algorithm for approximating the lowest eigenvalues of a Hamiltonian $\hat{H}$ by variationally optimizing a parametrized trial state \cite{Fedorov2022VQESurvey,Tilly2022VQEReview}. By the variational principle, the energy expectation value
\begin{equation}
\label{eq:vqe_opt}
E(\boldsymbol{\theta})
= \bra{\psi(\boldsymbol{\theta})}\hat{H}\ket{\psi(\boldsymbol{\theta})}
\ge E_0
\end{equation}
upper-bounds the ground-state energy $E_0$. VQE seeks parameters $\boldsymbol{\theta}$ that minimize $E(\boldsymbol{\theta})$, where the trial state is prepared by a parametrized quantum circuit
$\ket{\psi(\boldsymbol{\theta})}=U(\boldsymbol{\theta})\ket{\psi_{\mathrm{ref}}}$.

In this work, we consider many-body fermionic Hamiltonians arising in electronic-structure problems \cite{QChem2020McArdle,QChem2019Cao}. After mapping fermionic operators to qubits (via Jordan--Wigner \cite{JordanWigner1928PauliEquivalence}), the Hamiltonian can be expressed as a weighted sum of Pauli strings,
\begin{equation}
\label{eq:hamiltonian_decomp}
\hat{H} = \sum_{j} c_j P_j,
\qquad
c_j \in \mathbb{R},\;
P_j \in \{I,X,Y,Z\}^{\otimes n}.
\end{equation}
Accordingly, the VQE objective decomposes as
\begin{equation}
\label{eq:vqe_energy_sum}
E(\boldsymbol{\theta})
= \sum_j c_j \,\langle P_j\rangle_{\boldsymbol{\theta}},
\qquad
\langle P_j\rangle_{\boldsymbol{\theta}}
:= \bra{\psi(\boldsymbol{\theta})}P_j\ket{\psi(\boldsymbol{\theta})},
\end{equation}
and each $\langle P_j\rangle_{\boldsymbol{\theta}}$ is estimated from measurements in Pauli bases.

The performance of VQE depends critically on the choice of ansatz family $\{U(\boldsymbol{\theta})\}$.
A chemically motivated and widely used choice is the unitary coupled-cluster singles and doubles (UCCSD) ansatz \cite{Peruzzo2014VQE}, typically built on a Hartree--Fock (HF) \cite{QChem2020McArdle} reference state $\ket{\psi_{\mathrm{ref}}}=\ket{\psi_{\mathrm{HF}}}$. While UCCSD can achieve good accuracy for weakly correlated regimes, its parameter count and circuit depth scale unfavorably with system size, which limits its applicability on near-term hardware \cite{Kandala2017HardwareEfficientVQE,Grimsley2019ADAPTVQE,Tang2021QubitADAPTVQE,Yordanov2021QEBAVQE,Ramoa2025ADAPTVQE}. Hardware-efficient ans\"atze (HEAs) instead leverage device-native gate sets and connectivity \cite{Kandala2017HardwareEfficientVQE}, but HEAs are known to exhibit barren-plateau behavior in relevant regimes \cite{McClean2018BarrenPlateaus}.

We apply our ansatz-search framework to identify circuits that are simultaneously expressive, trainable, and of low complexity, using the same hierarchical objective as introduced in Eq.~\ref{eq:lexicographic_cost_multi}. As an initial guess, we initialize from the HF reference state as in \cite{Peruzzo2014VQE,Grimsley2019ADAPTVQE,Yordanov2021QEBAVQE}. We optimize the metrics for $\mathrm{H}_2$ at a bond length close to equilibrium, $R=0.75\text{\AA}$. We compare the discovered ansatz against the benchmark circuits from \cite{Sim2019}, including circuit~10 as a hardware-efficient ansatz (HEA), as well as against UCCSD and a standard two-local SU(2) hardware-efficient ansatz (see Appendix~\ref{app:vqe} Figure~\ref{fig:expr-pqc-efficientSU2}). All circuits are evaluated at the metric level, while UCCSD and the hardware-efficient ansätze are additionally assessed in terms of VQE performance for $\mathrm{H}_2$. For the VQE experiments, UCCSD is initialized from the Hartree–Fock reference state with initial parameters set to zero, whereas the hardware-efficient ansätze are initialized in the computational zero state with parameters sampled uniformly from $[0,2\pi]$, following standard practice \cite{Peruzzo2014VQE,Kandala2017HardwareEfficientVQE}.


The ansatz search circuit achieves favorable metric values (Figure~\ref{fig:vqe_obj}a). However, when deployed in VQE for $\mathrm{H}_2$ it fails to reliably converge to the ground-state energy (Figure~\ref{fig:vqe_obj}b). In contrast, UCCSD converges despite comparatively poor expressibility. This highlights that high ``global'' expressibility is not necessarily helpful for chemical ground state estimation. The relevant search manifold is constrained by particle number, spin, and by the physically relevant sector of Fock space \cite{QChem2019Cao,QChem2020McArdle}, rather than by the full Hilbert space. Consequently, expressibility must be interpreted relative to the task and its symmetry-constrained subspace. It is also interesting that UCCSD, with a single iteration, can exhibit strong trainability among the considered circuits.

Motivated by this mismatch, we replace the gate pool by qubit-excitation operators as used in QEB-ADAPT-VQE \cite{Yordanov2021QEBAVQE} (details in Appendix~\ref{app:vqe}). These operators are designed to generate chemically relevant excitations while preserving problem symmetries (e.g., particle number), and admit circuit-efficient implementations. Note that the expressibility is limited here by the subspace we explore. Using this pool, our search discovers an ansatz that successfully reaches the $\mathrm{H}_2$ ground state (Figure~\ref{fig:vqe_obj}b).

Searching for a separate circuit for each geometry can be costly. We therefore test whether an ansatz optimized at $R=0.75\,\text{\AA}$ generalizes across bond lengths. The same circuit achieves ground-state convergence for multiple interatomic distances (Figure~\ref{fig:vqe_obj}b). However, the underlying metrics are not stationary. While expressibility is geometry-independent, trainability depends on the objective landscape induced by the Hamiltonian coefficients. Writing the geometry-dependent Hamiltonian as $H(R)=\sum_i c_i(R) P_i$, the objective derivative reads
\begin{equation}
\label{eq:vqe_geometry}
\partial_{\theta_\mu} E(\boldsymbol{\theta};R)
=
\sum_i c_i(R)\,
\partial_{\theta_\mu}
\bra{\psi_{\mathrm{ref}}}
U^\dagger(\boldsymbol{\theta})\, P_i\, U(\boldsymbol{\theta})
\ket{\psi_{\mathrm{ref}}}.
\end{equation}
Therefore, gradients at different interatomic distances are not independent. They differ only through the coefficients $c_i(R)$, while they share the same underlying parameter derivatives of the projected Pauli expectation values. Importantly, evaluating trainability across several distances incurs only constant overhead. We leave a joint multi-interatomic distance objective for future work.

As a proof of concept, we also consider LiH at a bond length close to equilibrium, $R=1.5\text{\AA}$, and perform 100 iterations of Bayesian optimization. In this case, a single fixed ansatz found by our search does not consistently converge to the ground state across geometries (Figure~\ref{fig:vqe_obj}c). While most experiments required far fewer than 10,000 iterations, this harder problem may require a larger optimization budget to obtain meaningful results, which we leave for future work. More broadly, tackling chemically harder instances will likely benefit from improved computational efficiency and stronger chemistry/optimization priors. In particular, instead of comparing against Haar randomness over the full Hilbert space, we plan to introduce a more task-aligned expressibility definition in which  $\Pr_{\mathrm{target}}$ in Eq.~\ref{eq:generic_expr} represents a physically relevant subspace motivated by chemically relevant manifolds. 

Note, our framework is metric-agnostic: users can plug in task-specific objectives and constraints, and the search returns ansätze tailored to those choices.

\section{Summary and outlook}
\label{sec:summary}
We have shown that the barren plateau problem is undecidable. Additionally, we have presented a metric-guided ansatz-search framework for discovering parameterized quantum circuits (PQCs) under near-term constraints. On the methodological side, we have introduced an operational barren-plateau diagnostic by estimating gradient fluctuations from finitely many bounded (clipped) samples. Using a concentration bound, we obtain a dimension-free sample complexity requirement for trainability estimation. 

The ansatz search discovers circuits that satisfy the respective metric thresholds while substantially reducing complexity relative to the strongest benchmark circuits for that metric. In multi-objective searches, we reproduce the expected anticorrelation between expressibility and trainability. Highly expressive circuits tend to be poorly trainable, and the region of simultaneously expressive and trainable circuits appears difficult to access within common circuit families. Nevertheless, our results show that the framework can navigate this trade-off in a controlled way and identify feasible circuits in the joint expressibility--trainability regime. We have further observed a strong positive correlation between expressibility and entanglement in the explored families and identified circuits that satisfy joint thresholds while remaining competitive with the best benchmarks at substantially lower complexity.

We have also demonstrated hardware-constrained search on the 5-qubit IQM Spark device \cite{Ronkko2024OnPremSuperconductingQC} by restricting the search space to the native gate set and star connectivity. Already with 50 hardware iterations, the discovered ansatz improves over all benchmark circuits within the considered hardware-constrained search space. Together with the substantial overhead incurred when compiling structured templates into the native basis, this motivates direct hardware-aware ansatz discovery.

To connect metric-level optimization to downstream performance, we have evaluated discovered circuits in a QNN case study and in VQE. In the QNN setting (following \cite{Abbas2021PowerQNN}), the discovered PQC achieves a higher effective dimension while using substantially fewer trainable parameters than the reference architecture. In VQE, we find that favorable ``global'' metrics (in particular Haar-based expressibility) are not sufficient to guarantee ground-state convergence. Replacing the gate pool by symmetry-preserving excitation operators (as in QEB-ADAPT-VQE \cite{Yordanov2021QEBAVQE}) resolves this mismatch, underscoring the importance of aligning the ansatz with the target manifold.


\paragraph*{Future work.} For future work, we propose to study the effect of evaluating the metrics on a larger set of initializations (potentially including randomization), instead of fixed reference input states. The reason is that in an actual parameter optimization, the state may quickly move away from the fixed reference manifold. Understanding and leveraging this metric drift seems key to scaling property-based ansatz design. Beyond scaling via repetitions, we propose to study how discovered architectures scale with qubit number, add further metrics (e.g., magic/non-stabilizerness–inspired measures), and seek structural principles for circuits that are both expressive and trainable—potentially even crafting ansätze tied to physical phenomena to aid explanation and understanding. Finally, we note that our current implementation and optimization budgets are not yet efficiency-tuned, and we encourage domain experts to improve on these results.

\section*{Acknowledgements}
D.W. acknowledges support from the project J\"ulich UNified Infrastructure for Quantum computing (JUNIQ) that has received funding from the Federal Ministry of Research, Technology and Space (BMFTR) and the Ministry of Culture and Science of the State of North Rhine-Westphalia.
The authors gratefully acknowledge the Gauss Centre for Supercomputing e.V. (www.gauss-centre.eu) for funding this project by providing computing time on the GCS Supercomputer JUWELS at J\"ulich Supercomputing Centre (JSC).

\section*{Data availability}
The datasets generated during and/or analyzed during the current study are available from the corresponding author on reasonable request.

\appendix

\section{Expressibility at scale}
\label{app:expr}
The expressibility of a parameterized quantum circuit (PQC), when quantified via fidelity-distribution matching, is not solely a property of the gate template; it also depends on the input state ensemble on which the circuit acts. In the main text we therefore evaluate expressibility under multiple \emph{initializations} that serve as proxies for realistic inputs, e.g., states that arise as outputs of earlier layers (see Section~\ref{sec:power_of_qnn}). This is important because the same PQC can exhibit markedly different expressibility when applied to different input ensembles, as illustrated in Figure~\ref{fig:app_expr_rz}a. Consequently, in settings where a circuit block is reused across layers, one should not assume the input is $\ket{0}^{\otimes n}$; instead, expressibility should be assessed under the relevant intermediate-state distribution. This dependence is particularly relevant when the same block is repeated across layers, since repetition changes the distribution of intermediate states and can thereby alter the observed expressibility \cite{Sim2019}.

We emphasize that this initialization-based analysis does not constitute a direct, quantitative characterization of how expressibility changes as a function of the number of repetitions of the same block. A systematic metric for “expressibility growth under repetition” is left for future work.

Finally, fidelity-histogram expressibility can be artificially inflated by discretization effects. As a concrete example, consider a circuit consisting only of single-qubit $R_z$ rotations acting on the product input $\ket{+}^{\otimes n}$ with 75 fidelity bins (Figure~\ref{fig:app_expr_rz}b). This family explores only a low-dimensional manifold of product states and is far from Haar-expressive. However, for large $n$ the Haar fidelity distribution concentrates near $F\approx 0$ (with $N=2^n$), and with a fixed, coarse binning on $[0,1]$ the empirical histogram of such non-Haar ensembles can become indistinguishable from the discretized Haar reference within numerical precision. In this regime, a small binned discrepancy (e.g., $\mathrm{Expr}\approx 0$) reflects limited histogram resolution rather than genuine Haar-like coverage.

\section{Trainability error}
\label{app:train_err}
As a lightweight noise proxy, we account only for the physical error rates of one-qubit ($p_1$) and two-qubit ($p_2$) gates and their respective counts. This yields the heuristic circuit error probability
\begin{equation}
\operatorname{Pr}_{\mathrm{PQC}}(\mathrm{err})
   =
   1 -
   \bigl(1-\,p_1\bigr)^{N_1}\;
   \bigl(1-\,p_2\bigr)^{N_2},
\end{equation}
where $N_1$ and $N_2$ are the numbers of one-qubit and two-qubit gates in the circuit, respectively. 

We emphasize that more refined approaches include gate-dependent noise-channel simulations (e.g., Pauli/depolarizing or amplitude-damping noise with readout error), device-level characterization of gate errors via randomized benchmarking \cite{ErrComput2012Easwar,Wallman2014RandomizedBenchmarkingConfidence} or gate-set tomography \cite{BlumeKohout2017GSTFaultToleranceThreshold}, and holistic circuit benchmarks such as cycle benchmarking \cite{Erhard2019CycleBenchmarking}, cross-entropy benchmarking \cite{Boixo2018QuantumSupremacyCharacterization,Arute2019QuantumSupremacy}, or quantum volume \cite{PhysRevA.100.032328}.

\section{Ansatz-search setup}
\label{app:setup_details}

The ansatz search is carried out using Qiskit \cite{qiskit2024} and Bayesian optimization implemented with Optuna \cite{optuna}. 
The candidate circuits are assembled from the same gate pool used in our benchmark ansätze,
\[
\{\mathrm{H}, \mathrm{R_X}, \mathrm{R_Y}, \mathrm{R_Z}, \mathrm{CX}, \mathrm{CZ}, \mathrm{CRX}, \mathrm{CRZ}\},
\]
when not otherwise stated. Circuit construction proceeds sequentially: for each position in the circuit, the optimizer selects both a gate type and the qubit index or qubit pair on which that gate is applied. In this manner, the full ansatz is generated step by step. To reduce redundant exploration of the search space, candidate models that have been proposed more than $10$ times are pruned.

\begin{figure*}[t]
    \centering
    \begin{minipage}[t]{0.34\textwidth}
        \centering
        {\small (a)}\\[-0.0em]
        \includegraphics[width=0.45\linewidth]{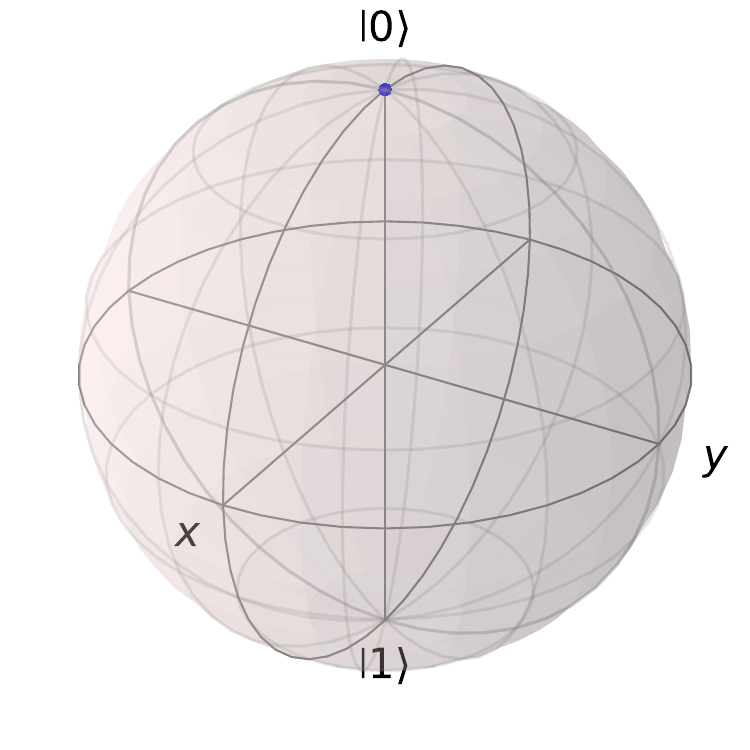}\hfill
        \includegraphics[width=0.45\linewidth]{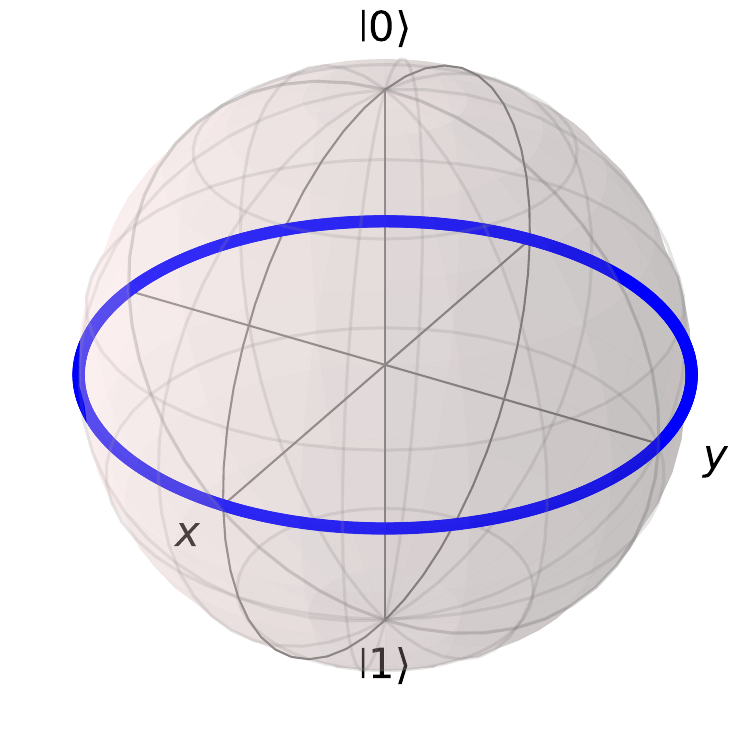}
        \label{fig:app_expr_rz_bloch}
    \end{minipage}\hfill
    \begin{minipage}[t]{0.64\textwidth}
        \centering
        {\small (b)}\\[-0.0em]
        \includegraphics[width=0.295\linewidth]{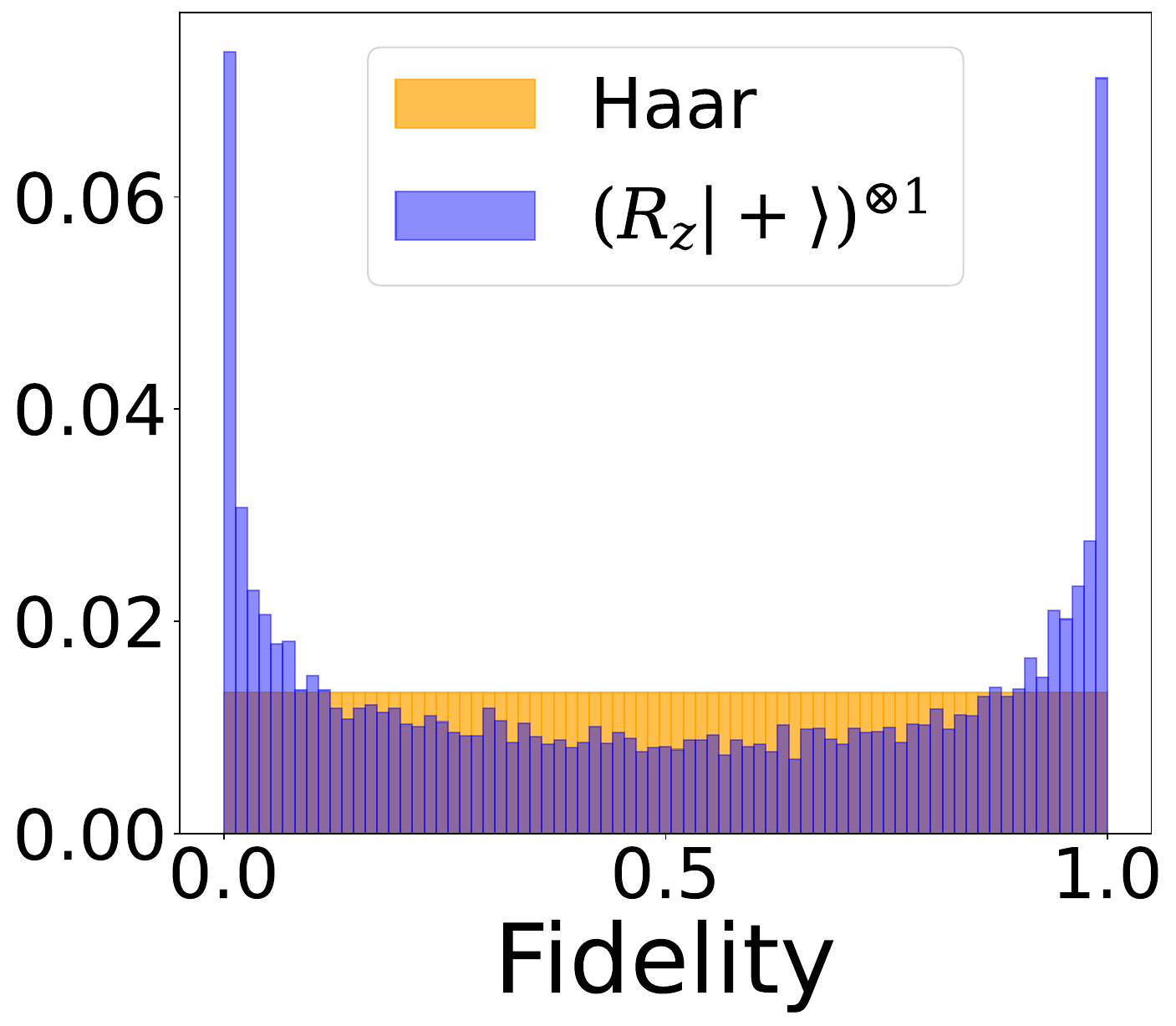}\hfill
        \includegraphics[width=0.285\linewidth]{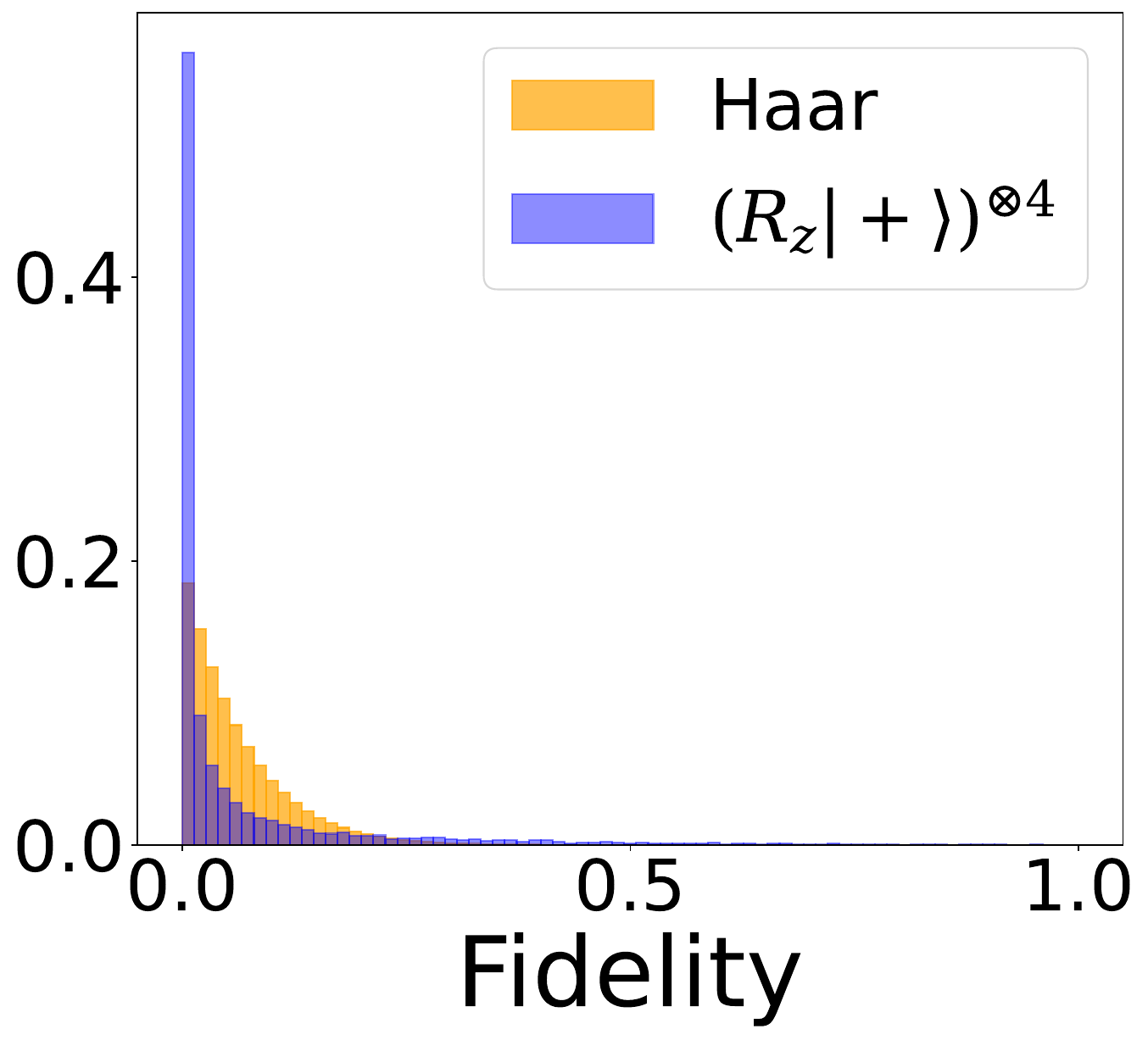}\hfill
        \includegraphics[width=0.285\linewidth]{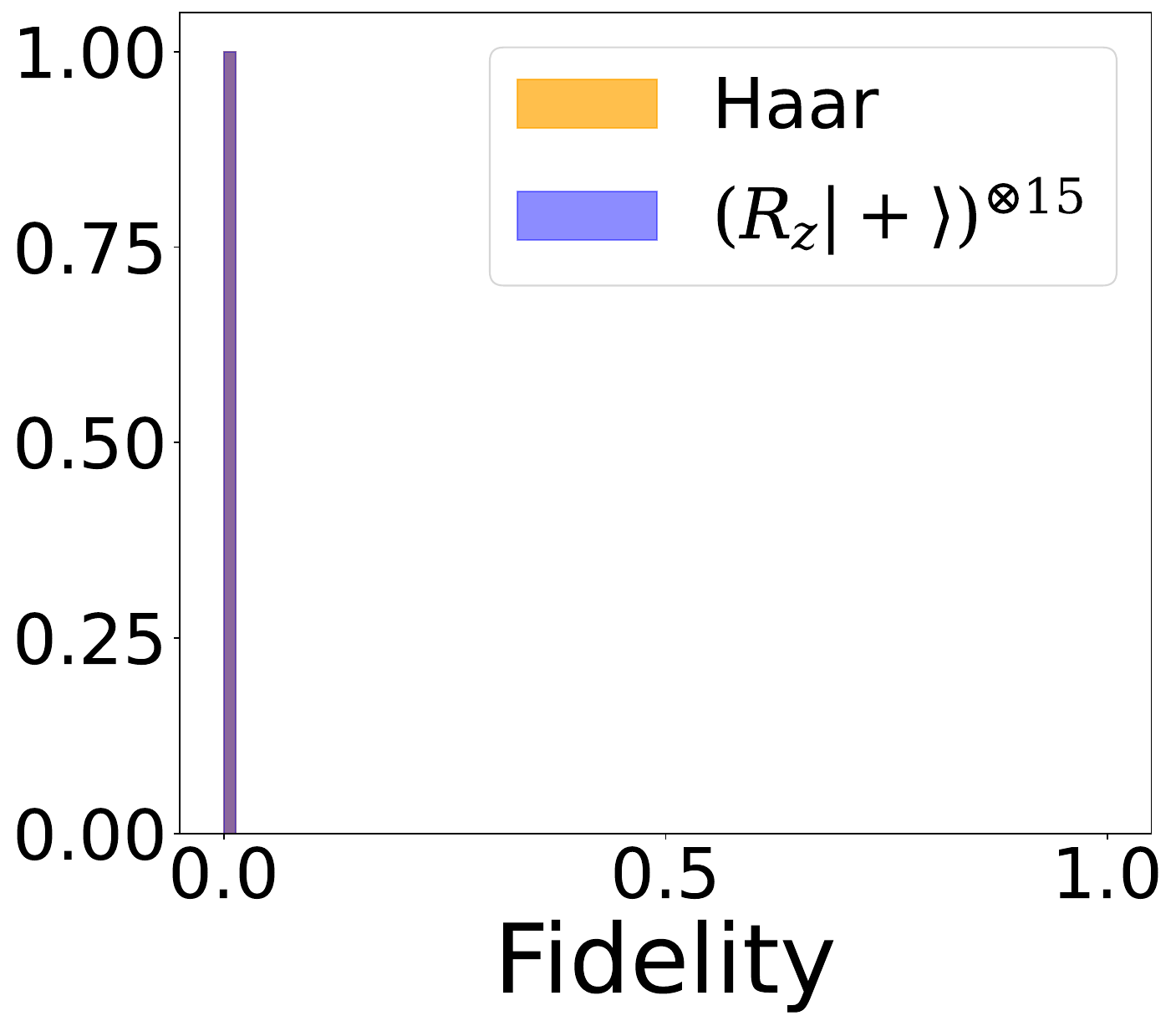}
        \label{fig:rz_fid_scale_hist}
    \end{minipage}

    \caption{%
    (a) Bloch-sphere intuition for the single-qubit family generated by $R_z(\theta)$ under two inputs: starting from $\ket{0}$, the state remains fixed on the Bloch sphere (localized ensemble), whereas starting from $\ket{+}$ the state traces the equator as $\theta$ varies (spread ensemble).
    (b) Binned pairwise-fidelity histograms for the corresponding $n$-qubit product ensemble produced by applying $R_z(\theta_i)$ on each qubit $i$ starting from $\ket{+}$ (with 75 bins), compared against the discretized Haar reference as $n$ increases.%
    }
    \label{fig:app_expr_rz}
\end{figure*}

All variational parameters are initialized by sampling uniformly from the interval $[0,2\pi]$. This provides a simple and uniform parameter domain for all considered circuit families. In principle, other initialization schemes could also be employed. For instance, sampling from narrower or Gaussian-distributed domains may be advantageous in settings where one aims to reduce the tendency toward barren-plateau-prone initializations \cite{NEURIPS2022_7611a3cb}.

Since our framework is intended to remain applicable to general gate sets, gradients are estimated using finite differences with step size $10^{-7}$. For specific gate libraries, analytic gradient techniques such as the parameter-shift rule \cite{Parameter2018Mitarai,Parameter2019Schuld} are preferable and should be used whenever available. Accordingly, the computational cost of our trainability metric depends directly on the efficiency with which gradients can be obtained. For the trainability evaluations in Sections~\ref{sec:single_objective} and~\ref{sec:multi_obj}, we use the local Pauli-\(Z\) observable
\(
O = I \otimes I \otimes I \otimes Z
\).

To ensure a fair complexity-controlled comparison, we further restrict the search space using the most efficient benchmark circuit that achieves the target objective (Table~\ref{tab:config}). Concretely, the maximum number of parameters, circuit depth, and total number of gates are bounded by those of the least complex successful benchmark ansatz. During sequential circuit construction, a trial is terminated as soon as one of these limits is reached. The candidate circuit is subsequently evaluated only if it defines a valid ansatz, meaning that it is connected and contains at least one trainable parameter; otherwise, the trial is pruned. 

For the quantum hardware expressibility experiment and the VQE experiments, all circuits are transpiled with the Qiskit transpiler at optimization level 3 to the native gate set $\{\mathrm{CZ}, \mathrm{RZ}\}$. In the quantum hardware experiment, transpilation additionally respects the hardware topology.

\begin{table*}[t]
\centering
\small
\resizebox{\linewidth}{!}{%
\input{tab/setup.tex}}
\caption{Search-space constraints and metric thresholds for all experiments. For each experiment, the table lists the number of qubits ($n$), the allowed gate set ($\mathcal{G}$), the connectivity/topology constraint ($\mathcal{T}$) when applicable, the bounds imposed on parameter count ($|\boldsymbol{\theta}|$), gate count ($G_{\max}$), and circuit depth ($D_{\max}$) during ansatz search, as well as the corresponding number of optimization trials ($N_{\text{trials}}$), metric thresholds used in the objective ($\tau_{\mathrm{Expr}},\tau_{\mathrm{BP}},\tau_{\mathrm{Ent}}$), and the epsilon truncation ($\varepsilon_{Expr}$).}
\label{tab:config}
\end{table*}

\subsection{Threshold selection.}
Selecting appropriate thresholds for the evaluation metrics considered in this work is a nontrivial task. A seemingly natural choice would be to define idealized target values, such as setting the expressibility threshold to the expected statistical noise level obtained from sampling Haar-random states, or defining the trainability threshold as the ratio between the maximal attainable gradient magnitude and the minimal circuit complexity required to achieve it.

However, such idealized thresholds are of limited practical use. In particular, it is generally unknown \emph{a priori} whether parametrized circuits can attain multiple of these values at all, nor how the required circuit depth or gate count scales in order to approach them. 

For these reasons, we adopt a data-driven thresholding strategy based on benchmark circuits. Concretely, for each metric, we define the threshold as the best value observed across the fixed set of benchmark circuits evaluated under identical conditions (Section~\ref{app:extended_results}). These benchmarks serve as reference points representing empirically attainable performance within the considered design space (Table~\ref{tab:config}).

We emphasize that this choice does not imply that the resulting thresholds are optimal or fundamental. Rather, they provide a consistent and operationally meaningful baseline that enables relative comparisons between different ansätze. A systematic investigation of theoretically motivated or problem-dependent threshold definitions is left for future work.

\begin{figure*}[t]
    \centering
    \subfloat[]{%
        \includegraphics[width=0.31\textwidth]{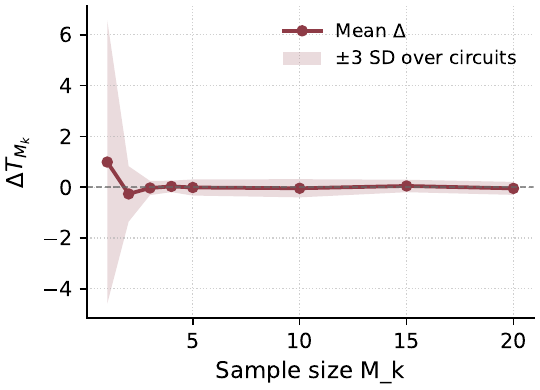}
    }\hspace{0.2\textwidth}
    \subfloat[]{%
        \includegraphics[width=0.31\textwidth]{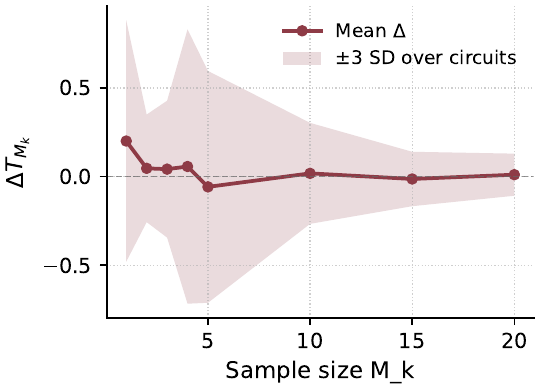}
    }
    \caption{Sensitivity of metric estimates to the number of sampled initial states in the QNN setting. For all baseline circuits and the reference QNN of \cite{Abbas2021PowerQNN}, metric values are computed by averaging over $|\mathcal{S}_0|=M_k$ randomly drawn (feature-map encoded) input states with $(M_k)=(1,2,3,4,5,10,15,20)$ and over circuit parameters uniformly sampled i.i.d.\ from $[-1,1]^{|\boldsymbol{\theta}|}$ following \cite{Abbas2021PowerQNN}. (a) Stepwise change in expressibility $E$ between consecutive sample sizes, $\Delta E_k := E(M_k)-E(M_{k-1})$. (b) Stepwise change in trainability $T$, $\Delta T_k := T(M_k)-T(M_{k-1})$.}
    \label{fig:app_init}
\end{figure*}

\subsection{QNN}
\label{app:qnn}
Measuring all qubits in the computational ($Z$) basis yields bitstrings
$x \in \{0,1\}^n$ distributed as $x \sim p_{\boldsymbol\theta}(x)$.
Any model output used for learning is then a classical post-processing $f(x)$
of these samples. For a shot outcome $x$, the eigenvalue of $Z_i$ is
$(-1)^{x_i}$, hence $\langle Z_i\rangle = \mathbb{E}_{x\sim p_{\boldsymbol\theta}}[(-1)^{x_i}]$.
In particular, parity classification corresponds to
\begin{equation}
f(x) = (-1)^{\sum_{i=1}^n x_i},
\end{equation}
whose expectation equals the parity correlator
\begin{equation}
\mathbb{E}_{x\sim p_{\boldsymbol\theta}}\!\left[f(x)\right]
= \left\langle \bigotimes_{i=1}^n Z_i \right\rangle
= \langle Z^{\otimes n} \rangle.
\end{equation}
Moreover, the even/odd probabilities are given by
\begin{equation}
p_{\mathrm{even}} = \frac{1 + \langle Z^{\otimes n}\rangle}{2},
\qquad
p_{\mathrm{odd}}  = \frac{1 - \langle Z^{\otimes n}\rangle}{2}.
\end{equation}
To diagnose barren plateaus for this parity-based readout, we study the scaling
of the gradient variance
\begin{equation}
\operatorname{Var}_{\boldsymbol\theta}\!\left[\partial_{\theta_k}\langle Z^{\otimes n}\rangle\right],
\end{equation}
with optional averaging over parameters $k$ and input samples.


\subsection{VQE}
\label{app:vqe}
We work in the Jordan--Wigner encoding, where each qubit represents the
occupation of one molecular \emph{spin orbital} and $\ket{1}$ ($\ket{0}$)
denotes an occupied (unoccupied) spin orbital. We use a blocked spin ordering:
the first $n_{\mathrm{spatial}}$ qubits correspond to $\alpha$-spin orbitals,
and the remaining $n_{\mathrm{spatial}}$ qubits correspond to $\beta$-spin
orbitals. Thus, qubits $p$ and $p+n_{\mathrm{spatial}}$ correspond to the
$\alpha$- and $\beta$-spin orbitals of the same spatial orbital $p$,
respectively. We denote the spin associated with index $q$ by
$\sigma_q \in \{\alpha,\beta\}$ according to the spin block to which $q$
belongs. 

\paragraph*{Single excitations (SE).}
The gate $\tilde A_{ik}(\theta)$ implements a number-preserving hopping between
spin orbitals $i$ and $k$:
\begin{equation*}
\tilde{A}_{ik}(\theta)
=
\exp\!\left(
i\,\frac{\theta}{2}
\left(
X_i Y_k - Y_i X_k
\right)
\right).
\end{equation*}
It performs an $\mathrm{SU}(2)$ rotation in the two-qubit subspace spanned by
$\{\ket{10},\ket{01}\}$, while leaving $\ket{00}$ and $\ket{11}$ invariant.
To conserve the spin projection $S_z$, we restrict to spin-conserving single
excitations, i.e.\ $\sigma_i=\sigma_k$.

\paragraph*{Double excitations (DE).}
The gate $\tilde A_{ijkl}(\theta)$ implements a number-preserving two-electron
excitation that mixes configurations in which two occupations are transferred
between two pairs of spin orbitals (for example,
$\ket{1100}\leftrightarrow\ket{0011}$ in a suitable ordering):
\begin{equation*}
\scalebox{0.85}{$
\begin{aligned}
\tilde{A}_{ijkl}(\theta)
&=
\exp\!\Bigg(
i\,\frac{\theta}{8}
\Big(
  X_i Y_j X_k X_l
+ Y_i X_j X_k X_l
+ Y_i Y_j Y_k X_l
+ Y_i Y_j X_k Y_l \\
&\hspace{2.2em}
- X_i X_j Y_k X_l
- X_i X_j X_k Y_l
- Y_i X_j Y_k Y_l
- X_i Y_j Y_k Y_l
\Big)
\Bigg).
\end{aligned}
$}
\end{equation*}
We restrict to spin-conserving double excitations by allowing only index combinations for which the total spin projection is preserved.

\begin{figure*}[tbp]
  \centering
  \begingroup
  \input{pqc/efficientSU2}
  \endgroup
  \caption{SU2 HEA.}
  \label{fig:expr-pqc-efficientSU2}
\end{figure*}

\onecolumngrid

\section{Theorem \ref{thm:main} proof}
\label{app:main_proof}
\begin{theorem}
Let $L,U\in\mathbb R$ with $L<U$ and $R:=U-L$, $\mathbf{X}=(X_1,\dots,X_m)$ be independent random gradient samples with $X_i\in[L,U]$, and $m\ge 3$. Define $Z(\mathbf{X}):=\frac{m}{R^2}s_m^2$, where the sample variance is    
\[
s_m^2:=\frac{1}{m-1}\sum_{i=1}^m (X_i-\bar X)^2,
\qquad
\bar X:=\frac{1}{m}\sum_{i=1}^m X_i.
\]
 For $k\in\{1,\dots,m\}$ and $y\in[L,U]$, let $\mathbf X_{y,k}$ be obtained from $\mathbf X$ by replacing $X_k$ with $y$, and set
\[
\Delta_k(\mathbf X):=Z(\mathbf X)-\inf_{y\in[L,U]} Z(\mathbf X_{y,k}).
\]
Then almost surely
\begin{equation}
    \Delta_k(\mathbf X)\le 1 \quad \forall k,
\end{equation}
\begin{equation}
\sum_{k=1}^m \Delta_k(\mathbf X)^2 \le a\, Z(\mathbf X)
\end{equation}
with $a=\frac{m}{m-1}$. Therefore, for all $\delta>0$:

\begin{equation}
    \operatorname{Pr}\left(\bigl|\sqrt{s_m^2}-\sqrt{\mathbb E[s_m^2]}\bigr|
\le
\sqrt{\frac{2R^2\,\ln(2/\delta)}{(m-1)}}\right)\geq 1-\delta.
\end{equation}
\end{theorem}
\begin{proof}

\begin{enumerate}
\item \label{itm:singlecoord} \textbf{Single--coordinate influence.}  Denote the leave‑one‑out mean by $\displaystyle\bar X_{(k)}:=\frac1{m-1}\sum_{i\ne k}X_i$ and variance by $s_{m-1,(k)}^2:=\frac{1}{m-2}\sum_{i\neq k}\bigl(X_i-\bar X_{(k)}\bigr)^2$.  The variance‐update identity gives
  \[
    s_m^2=\frac{m-2}{m-1}s_{m-1,(k)}^2+\frac{(X_k-\bar X_{(k)})^2}{m},
  \]
  which yields an expression for $Z(\mathbf X) = \frac{m}{R^2} s_m^2$. To obtain $\inf_{y\in[L,U]}Z(\mathbf X_{y,k})$, we note that the first term $\frac{m-2}{m-1}s_{m-1,(k)}^2$ does not depend on $y$, so the minimum is at $y=\bar X_{(k)}$ for which last term vanishes. Hence
  \[
    \Delta_k=Z(\mathbf X)-\inf_{y}Z(\mathbf X_{y,k})=\frac{(X_k-\bar X_{(k)})^2}{R^2}.
  \]
  Because both numbers lie in $[L,U]$, $|X_k-\bar X_{(k)}|\le R$ and thus $\Delta_k\le1$.    

\item \label{itm:sumsqinf} \textbf{Sum of squared influences.}  First, let us note that we can rewrite  $X_k - \bar X_{(k)}$:

\begin{align}
m\bar X
      &= \sum_{i\ne k} X_i + X_k
       = (m-1)\bar X_{(k)} + X_k, \notag\\
\Longrightarrow\quad
\bar X_{(k)}
      &= \frac{m\bar X - X_k}{m-1}, \notag\\
\Longrightarrow\quad
X_k - \bar X_{(k)}
      &= X_k - \frac{m\bar X - X_k}{m-1}
       = \frac{m}{m-1}\bigl(X_k - \bar X\bigr).
\label{eq:xk-minus-xbar-k}
\end{align}

Second, since all $X_i\in[L,U]$, we have $\bar X_{(k)}\in[L,U]$. Rearranging \eqref{eq:xk-minus-xbar-k} yields
\begin{equation}
\label{eq:xk-minus-xbar-bound}
  |X_k-\bar X| 
  \overset{\eqref{eq:xk-minus-xbar-k}}{=} | X_k - \bar{X}_{(k)}|\Bigl(\frac{m-1}{m}\Bigr) 
  \le (U-L)\Bigl(\frac{m-1}{m}\Bigr) ,
\end{equation}

which is attained for the configuration $X_k=U$ and $X_i=L$ for all $i\neq k$ (or vice versa). Hence,
  \begin{align*}
    \sum_{k=1}^m \Delta_k^{2}&\overset{\eqref{eq:xk-minus-xbar-k}}{=}\frac{m^{4}}{R^4(m-1)^{4}}\sum_{k=1}^m(X_k-\bar X)^{4}\overset{\eqref{eq:xk-minus-xbar-bound}}{
    \le}\frac{m^{4}}{R^4(m-1)^{4}}\,R^2\left(\frac{m-1}{m}\right)^2\sum_{k=1}^m(X_k-\bar X)^{2}\\
    &=\frac{m^{2}}{R^2(m-1)}\,s_m^{2}=\frac{m}{(m-1)}\,Z(\mathbf X).
  \end{align*}
\\

\item \textbf{Concentration bound.} Given the single-coordinate influence (\ref{itm:singlecoord}) and sum of squared influences (\ref{itm:sumsqinf}) with $a=\frac{m}{(m-1)}$ we get for any $t>0$ with Theorem~7 in \cite{maurer2009empirical} (or Theorem~13 in \cite{Maurer2006ConcentrationInequalities}) :
\begin{align}
\operatorname{Pr}(\mathbb E[s_m^2]-s_m^2>t)=\operatorname{Pr}(\mathbb E[Z]-Z>\tfrac{m}{R^2}t)
&\le \exp\!\left(-\frac{m\,t^{2}}{2R^2a\,\mathbb E[s_m^2]}\right),
\label{eq:lowerV}\\
\operatorname{Pr}\{s_m^2-\mathbb E[s_m^2]>t\}=\operatorname{Pr}(Z-\mathbb E[Z]>\tfrac{m}{R^2}t)
&\le \exp\!\left(-\frac{m\,t^{2}}{R^2a\,(2\mathbb E[s_m^2]+t)}\right).
\label{eq:upperV}
\end{align}
Fix $\eta>0$ and set $L:=\ln(1/\eta)$. For the lower tail bound \eqref{eq:lowerV} we have,
\[
\operatorname{Pr}\{\mathbb E[s_m^2]-s_m^2>t\}\le \eta
\quad\text{whenever}\quad
t \ge \sqrt{\frac{2R^2a\,\mathbb E[s_m^2]\,L}{m}}.
\]
Let $\alpha:=\sqrt{\frac{2R^2aL}{m}}$, then with probability at least $1-\eta$,
\[
\mathbb E[s_m^2] - \alpha\sqrt{\mathbb E[s_m^2]} \le s_m^2,
\]
\[
\Longrightarrow \left(\sqrt{\mathbb E[s_m^2]}-\frac{\alpha}{2}\right)^2 \le s_m^2 + \frac{\alpha^2}{4},
\]
\[
\Longrightarrow\sqrt{\mathbb E[s_m^2]} \le \sqrt{s_m^2 + \frac{\alpha^2}{4}} + \frac{\alpha}{2}\leq \sqrt{s_m^2}+\alpha, \qquad \text{(using } \sqrt{u+v}\le \sqrt u+\sqrt v\text{)}
\]
i.e., with probability at least $1-\eta$,
\begin{equation}
\sqrt{\mathbb E[s_m^2]} \le \sqrt{s_m^2} + \sqrt{\frac{2R^2a\,\ln(1/\eta)}{m}}.
\label{eq:rootlower}
\end{equation}
For the upper tail bound \eqref{eq:upperV}, requiring $\operatorname{Pr}\{s_m^2-\mathbb E[s_m^2]>t\}\le \eta$ is ensured by
\[
\frac{m\,t^{2}}{R^2a(2\mathbb E[s_m^2]+t)} \ge L
\quad\Longleftrightarrow\quad
m t^2 - R^2aLt - 2R^2aL\,\mathbb E[s_m^2] \ge 0.
\]
The smallest $t\ge 0$ satisfying this quadratic inequality is
\[
t_\eta
=\frac{R^2aL+\sqrt{R^4a^2L^2+8R^2aL\,m\,\mathbb E[s_m^2]}}{2m}
\;\le\;
\frac{R^2aL}{m} + \sqrt{\frac{2R^2a\,\mathbb E[s_m^2]\,L}{m}}.
\]
Hence, with probability at least $1-\eta$,
\[
s_m^2 \le \mathbb E[s_m^2] + \sqrt{\frac{2R^2a\,\mathbb E[s_m^2]\,L}{m}} + \frac{R^2aL}{m}.
\]
Writing $\beta:=\frac{R^2aL}{m}$, this can be rewritten as
\[
s_m^2 \le \Bigl(\sqrt{\mathbb E[s_m^2]}+\sqrt{\beta/2}\Bigr)^2 + \frac{\beta}{2}.
\]
Taking square roots and using $\sqrt{u+v}\le \sqrt{u}+\sqrt{v}$ again yields
\begin{equation}
\sqrt{s_m^2(\mathbf X)} \le \sqrt{\mathbb E[s_m^2]} + \sqrt{\frac{2R^2a\,\ln(1/\eta)}{m}}.
\label{eq:rootupper}
\end{equation}
Combining \eqref{eq:rootlower} and \eqref{eq:rootupper}, and applying a union bound (with $\eta=\delta/2$) yields that for any $\delta>0$:
\[
\operatorname{Pr}\left(\bigl|\sqrt{s_m^2(\mathbf X)}-\sqrt{\mathbb E[s_m^2]}\bigr|
\le
\sqrt{\frac{2R^2\,\ln(2/\delta)}{m-1}}\right)\geq 1-\delta.
\]
\end{enumerate}
\end{proof}

Alternatively, the result may be proved by mapping \(X\in[-1,1]\) to \(Y=(X+1)/2\in[0,1]\). Since \(s_m^2(X)=4\,s_m^2(Y)\), one can apply Theorem~10 of \cite{maurer2009empirical} to \(Y\). After reformulating the resulting inequality in terms of \(X\) and applying a union bound, one obtains the same bound. A tighter bound may be obtainable by exploiting additional distributional information or sharper estimates for the root inequalities.

\section{Extended results}
\label{app:extended_results}
This appendix collects supplementary quantitative results and circuit visualizations for the experiments discussed in Section~\ref{sec:experiments}. In particular, we provide the full benchmark comparisons underlying the single-objective, multi-objective, QNN, and VQE case studies, together with circuit diagrams of the ans\"atze returned by the search. The purpose of this appendix is to document the extended results referenced in the main text and to make the comparison across benchmark families and discovered circuits transparent. Numerical values are rounded to four decimal places.

\subsection{Single-objective optimization}
\label{app:single_obj}
In this subsection, we report the full results for the single-objective searches discussed in Section~\ref{sec:single_objective}. Table~\ref{tab:ext_single_obj} summarizes the metric values and complexity measures for all benchmark circuits and for the ans\"atze found by optimizing expressibility, entanglement, or trainability individually. The corresponding circuit diagrams are presented in Figures~\ref{fig:expr-pqc}, \ref{fig:train-pqc}, and \ref{fig:entgl-pqc}, together with the additional hardware-constrained expressibility results on IQM Spark reported in Table~\ref{tab:iqm_expr} and Figures~\ref{fig:iqm-expr-pqc} and \ref{fig:iqm-expr-pqc-10k}.
\begin{table*}[tbp]
\centering
\scriptsize
\begin{adjustbox}{max width=\textwidth}
\input{tab/single_metric.tex}
\end{adjustbox}
\caption{Detailed single-objective results. The results are shown for the benchmark circuits $i^r$ from \cite{Sim2019} and the ansätze found by our search ($\mathrm{AS}_{\mathrm{Expr}}$, $\mathrm{AS}_{\mathrm{Ent}}$, $\mathrm{AS}_{\mathrm{Train}}$).
For each circuit ID and repetition $r\in\{1,\dots,5\}$ we report parameter count $|{\boldsymbol\theta}|$, gate count $G$, depth $D$, and the achieved metric values (expressibility, entanglement, trainability).
Entries marked with ``--'' indicate metrics that are not applicable to the corresponding single-objective run.}
\label{tab:ext_single_obj}
\end{table*}

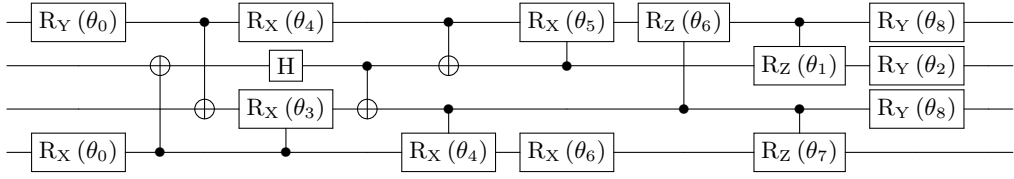
\begin{figure}[tbp]
  \centering
  \begingroup
  \input{pqc/expr_pqc}
  \endgroup
  \caption{Circuit diagram of the ansatz returned by the single-objective search optimizing expressibility ($\mathrm{AS}_{\mathrm{Expr}}$).}
  \label{fig:expr-pqc}
\end{figure}

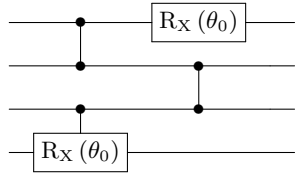
\begin{figure}[tbp]
  \centering
  \begingroup
  \input{pqc/train_pqc}
  \endgroup
  \caption{Circuit diagram of the ansatz returned by the single-objective search optimizing trainability ($\mathrm{AS}_{\mathrm{Train}}$).}
  \label{fig:train-pqc}
\end{figure}

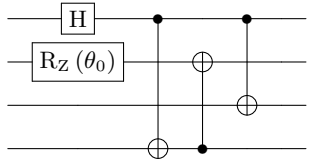
\begin{figure}[!tbp]
  \centering
  \begingroup
  \input{pqc/entgl_pqc}
  \endgroup
  \caption{Circuit diagram of the ansatz returned by the single-objective search optimizing entanglement ($\mathrm{AS}_{\mathrm{Ent}}$).}
  \label{fig:entgl-pqc}
\end{figure}

\begin{table}[!tbp]
\centering
\scriptsize
\begin{adjustbox}{width=0.6\textwidth}
\input{tab/iqm_metric}
\end{adjustbox}
\caption{Expressibility ansatz search on quantum hardware. The expressibility and effective circuit complexity are shown after compiling the benchmark circuits from \cite{Sim2019} and the discovered ansätze to the IQM Spark native basis and star connectivity. We report $|{\boldsymbol\theta}|$, gate count $G$, depth $D$, and expressibility across repetitions $r\in\{1,\dots,5\}$. Gate count and depth are reported after transpilation to the native gate set $\{\mathrm{CZ},\mathrm{RZ}\}$ and star topology. The last rows summarize the hardware-constrained ansatz-search result ($\mathrm{AS}_{\mathrm{Expr}}$) and the best circuit obtained from the 10{,}000-trial simulation in the same restricted gate set/topology ($\mathrm{AS}_{\mathrm{Expr,10k}}$).}
\label{tab:iqm_expr}
\end{table}

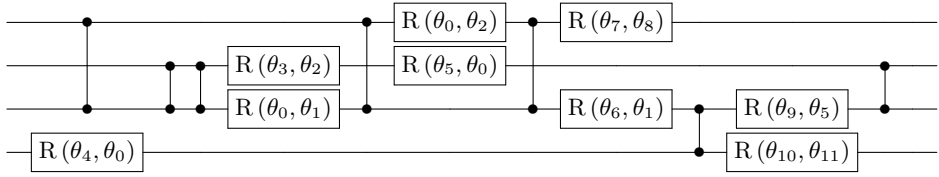
\begin{figure}[tbp]
  \centering
  \begingroup
  \input{pqc/iqm_expr_pqc}
  \endgroup
  \caption{Circuit returned by the hardware-constrained expressibility search on the 5-qubit IQM Spark device ($\mathrm{AS}_{\mathrm{Expr}}$).}
  \label{fig:iqm-expr-pqc}
\end{figure}

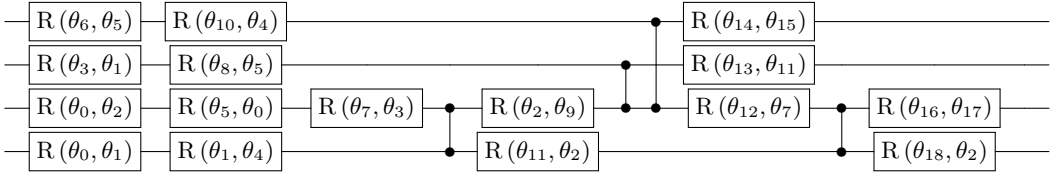
\begin{figure}[tbp]
  \centering
  \begingroup
  \input{pqc/iqm_sim_expr_pqc}
  \endgroup
  \caption{Circuit obtained from the 10{,}000-trial simulation using the same IQM Spark gate set and star topology ($\mathrm{AS}_{\mathrm{Expr,10k}}$).}
  \label{fig:iqm-expr-pqc-10k}
\end{figure}

\subsection{Multi-objective optimization}
Here we provide the complete results for the multi-objective searches from Section~\ref{sec:multi_obj}. Table~\ref{tab:multi_obj_ext} lists the benchmark circuits and the discovered ans\"atze for the joint expressibility--trainability and expressibility--entanglement objectives, while Figures~\ref{fig:expr-train-pqc} and \ref{fig:expr-entgl-pqc} present the corresponding returned circuit solutions.

\FloatBarrier

\begin{table}[!h]
\centering
\scriptsize
\begin{adjustbox}{width=1\textwidth}
\input{tab/multi_metric}
\end{adjustbox}
\caption{Detailed multi-objective results. The results are shown for the benchmark circuits $i^r$ from \cite{Sim2019} and the ansätze found by our search ($\mathrm{AS}_{\mathrm{Expr}-\mathrm{Ent}}$, $\mathrm{AS}_{\mathrm{Expr}-\mathrm{Train}}$).
For each circuit ID and repetition $r\in\{1,\dots,5\}$ we report parameter count $|{\boldsymbol\theta}|$, gate count $G$, depth $D$, and the achieved metric values (expressibility, entanglement, trainability).
Entries marked with ``--'' indicate metrics that are not applicable to the corresponding multi-objective run.}
\label{tab:multi_obj_ext}
\end{table}

\begin{figure}[!h]
  \centering
  \begingroup
  \input{pqc/expr_train_pqc}
  \endgroup
  \caption{Circuit diagram of the ansatz returned by the multi-objective expressibility and trainability search ($\mathrm{AS}_{\mathrm{Expr-Train}}$).}
  \label{fig:expr-train-pqc}
\end{figure}
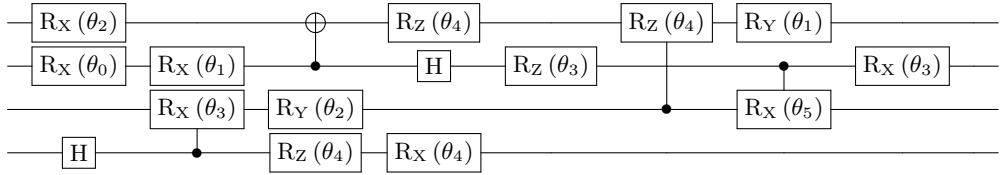

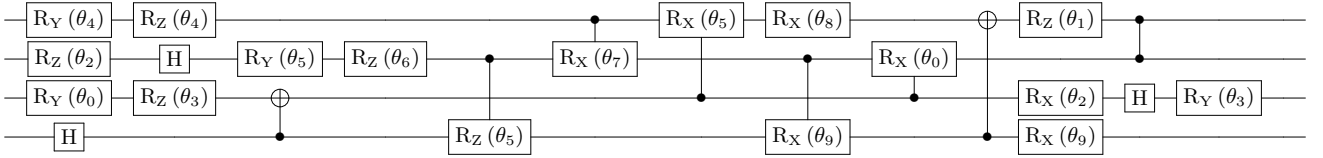
\begin{figure}[tbp]
  \centering
  \begingroup
  \input{pqc/expr_entgl_pqc}
  \endgroup
  \caption{Circuit diagram of the ansatz returned by the multi-objective expressibility and entanglement search ($\mathrm{AS}_{\mathrm{Expr-Ent}}$).}
  \label{fig:expr-entgl-pqc}
\end{figure}

\FloatBarrier

\subsection{QNN}
This subsection contains the extended quantitative results for the QNN case study of Section~\ref{sec:power_of_qnn}. Table~\ref{tab:qnn_ext} reports the benchmark comparisons, the reference architecture from \cite{Abbas2021PowerQNN}, and the ansatz returned by our search, while Fig.~\ref{fig:qnn-expr-entgl-pqc} shows the discovered circuit explicitly.
\begin{table}[!h]
\centering
\scriptsize
\begin{adjustbox}{width=0.75\textwidth}
\input{tab/power_of_qnn}
\end{adjustbox}
\caption{Detailed QNN results. The results are shown for the benchmark circuits $i^r$ from \cite{Sim2019}, the QNN from \cite{Abbas2021PowerQNN}, and the ansätze found by our search ($\mathrm{AS}$).
For each circuit ID and repetition $r\in\{1,\dots,5\}$ we report parameter count $|{\boldsymbol\theta}|$, gate count $G$, depth $D$, and the achieved metric values (expressibility, trainability).
Entries marked with ``--'' indicate metrics that are not applicable to the corresponding ansatz search run.}
\label{tab:qnn_ext}
\end{table}

\begin{figure}[!h]
  \centering
  \begingroup
  \input{pqc/qnn_pqc}
  \endgroup
  \caption{Circuit diagram of the QNN ansatz returned by the multi-objective expressibility and trainability search ($\mathrm{AS}$).}
  \label{fig:qnn-expr-entgl-pqc}
\end{figure}
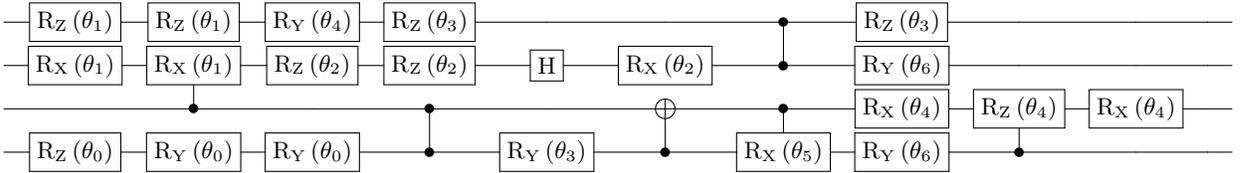

\subsection{H2}
We next report the detailed results for the $\mathrm{H}_2$ VQE study described in Section~\ref{sec:vqe}. Table~\ref{tab:h2_ext} includes the benchmark circuits, UCCSD, the SU2 hardware-efficient ansatz, and the ans\"atze returned by our search, together with their metric values and complexity measures. The circuit diagrams shown below distinguish the solution obtained in the generic gate pool (Figure~\ref{fig:h2-expr-entgl-pqc}) from the excitation-gate-based ansatz (Figure~\ref{fig:h2-excitation-expr-entgl-pqc}) used in the chemistry-informed setting.
\begin{table*}[tbp]
\centering
\scriptsize
\begin{adjustbox}{width=0.78\textwidth}
\input{tab/h2}
\end{adjustbox}
\caption{Detailed $\mathrm{H}_2$ results. The results are shown for the benchmark circuits $i^r$ from \cite{Sim2019} including circuit 10 as HEA, the HEA SU2 from Figure~\ref{fig:expr-pqc-efficientSU2}, the UCCSD \cite{Peruzzo2014VQE}, and the ansätze found by our search ($\mathrm{AS}$).
For each circuit ID and repetition $r\in\{1,\dots,5\}$ we report parameter count $|{\boldsymbol\theta}|$, gate count $G$, depth $D$, and the achieved metric values (expressibility, trainability). Gate count and depth are reported after transpilation to the native gate set $\{\mathrm{CZ},\mathrm{RZ}\}$.
Entries marked with ``--'' indicate metrics that are not applicable to the corresponding ansatz search run.}
\label{tab:h2_ext}
\end{table*}

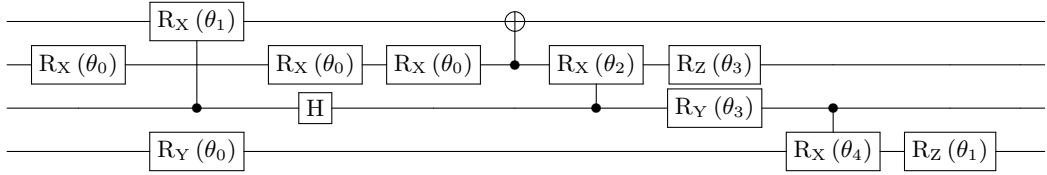
\begin{figure}[tbp]
  \centering
  \begingroup
  \input{pqc/vqe_h2}
  \endgroup
  \caption{Circuit diagram of the $\mathrm{H}_2$ ansatz returned by the multi-objective expressibility and trainability search ($\mathrm{AS}_{\mathrm{bench}}$).}
  \label{fig:h2-expr-entgl-pqc}
\end{figure}

\begin{figure}[!h]
  \centering
  \begingroup
  \input{pqc/vqe_h2_excitation}
  \endgroup
  \caption{Circuit diagram of the $\mathrm{H}_2$ ansatz returned by the multi-objective expressibility and trainability search ($\mathrm{AS}_{\mathrm{excitation}}$) with excitation gates.}
  \label{fig:h2-excitation-expr-entgl-pqc}
\end{figure}
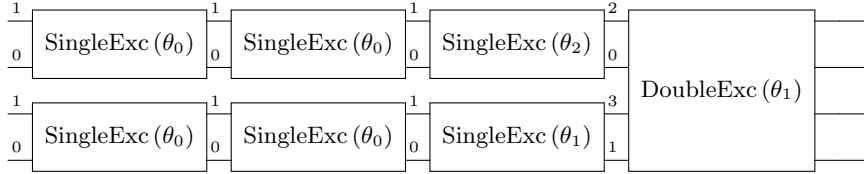

\subsection{LiH}
Finally, we provide the proof-of-concept LiH results corresponding to the experiment discussed in Section~\ref{sec:vqe}. Table~\ref{tab:lih_ext} summarizes the metric values and complexity of the considered circuits, and Fig.~\ref{fig:lih-expr-entgl-pqc} shows the ansatz returned by the corresponding search run.
\begin{table}[!h]
\centering
\scriptsize
\begin{adjustbox}{max width=\textwidth}
\input{tab/lih}
\end{adjustbox}
\caption{Detailed LiH results. The results are shown for the considered LiH ansätze together with the ansatz found by our search ($\mathrm{AS}$). We report parameter count $|{\boldsymbol\theta}|$, gate count $G$, depth $D$, and the achieved metric values (expressibility and trainability). Gate count and depth are reported after transpilation to the native gate set $\{\mathrm{CZ},\mathrm{RZ}\}$.}
\label{tab:lih_ext}
\end{table}

\begin{figure}[tbp]
  \centering
  \begingroup
  \input{pqc/lih_pqc}
  \endgroup
  \caption{Circuit diagram of the LiH ansatz returned by the multi-objective expressibility and trainability search ($\mathrm{AS}_{\mathrm{excitation}}$).}
  \label{fig:lih-expr-entgl-pqc}
\end{figure}
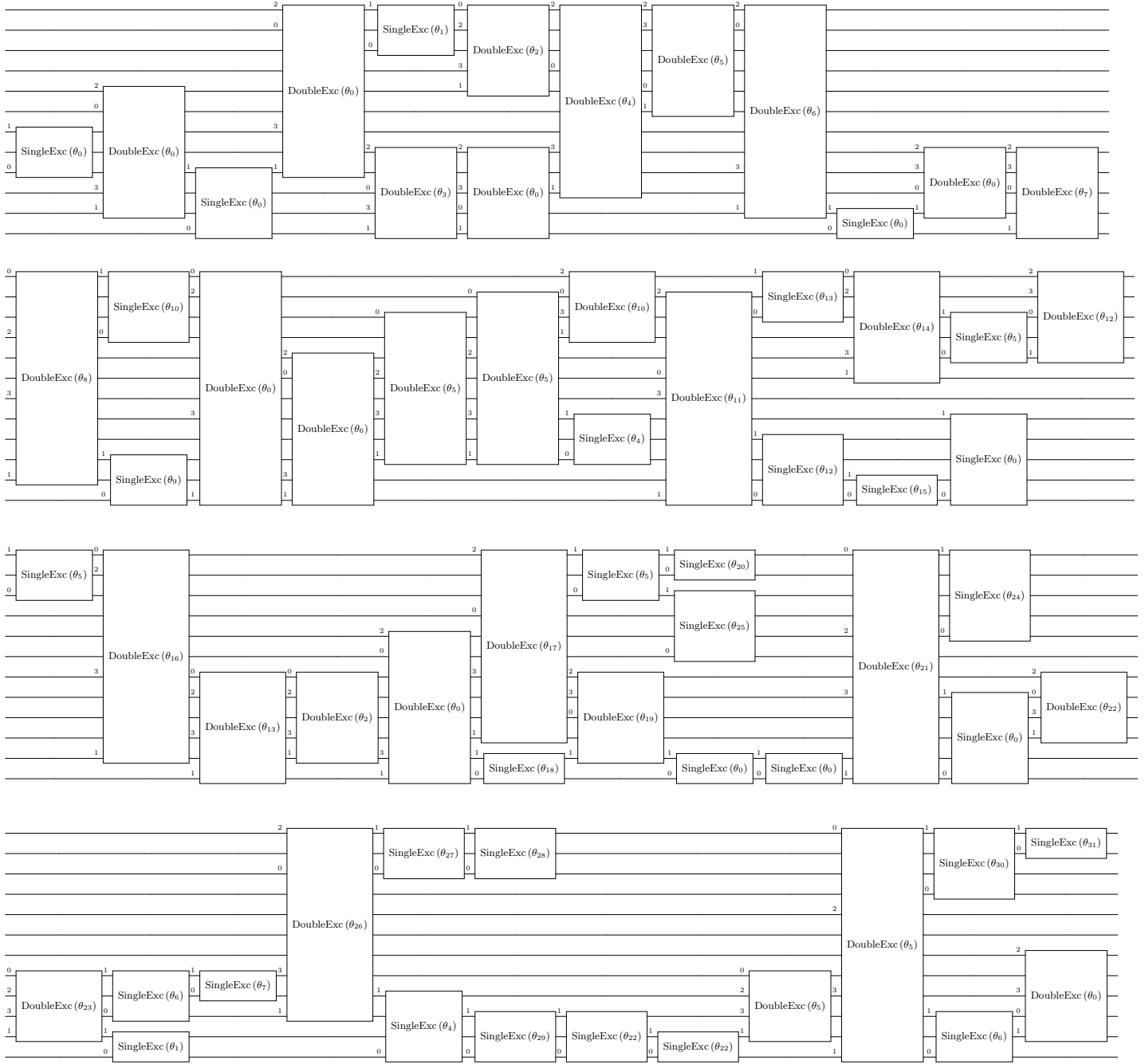

\twocolumngrid

\FloatBarrier

\bibliography{bib}

\end{document}

%% file: tab/setup.tex
\begin{tabular}{ccccccccccccc}
\toprule
Experiment & $n$ & $\mathcal{G}$ & $\mathcal{T}$ & $|\boldsymbol{\theta}|_{\max}$ & $G_{\max}$ & $D_{\max}$ & $|\boldsymbol{\theta}|_{\min}$ & $N_{\mathrm{trials}}$ & $\tau_{\mathrm{Expr}}$ & $\varepsilon_{\mathrm{Expr}}$ & $\tau_{\mathrm{BP}}$ & $\tau_{\mathrm{Ent}}$ \\
\midrule
Expr & 4 & H, RX, RY, RZ, CX, CZ, CRX, CRZ & -- & 28 & 28 & 15 & 0 & 10000 & 0.005 & 0 & -- & -- \\
Expr IQM & 4 & R, CZ & 0-2, 1-2, 2-3 & 12 & 41 & 25 & 0 & 50 & 0.1 & 0 & -- & -- \\
Expr IQM 10k & 4 & R, CZ & 0-2, 1-2, 2-3 & 20 & 115 & 65 & 0 & 10000 & 0.005 & 0 & -- & -- \\
Train & 4 & H, RX, RY, RZ, CX, CZ, CRX, CRZ & -- & 8 & 11 & 5 & 0 & 10000 & -- & -- & 8.0 & -- \\
Ent & 4 & H, RX, RY, RZ, CX, CZ, CRX, CRZ & -- & 4 & 11 & 5 & 0 & 10000 & -- & -- & -- & 0 \\
Expr vs Train & 4 & H, RX, RY, RZ, CX, CZ, CRX, CRZ & -- & 28 & 28 & 15 & 0 & 10000 & 0.005 & 0 & 2.0 & -- \\
Expr vs Ent & 4 & H, RX, RY, RZ, CX, CZ, CRX, CRZ & -- & 28 & 28 & 15 & 0 & 10000 & 0.005 & 0 & -- & 0.85 \\
QNN & 4 & H, RX, RY, RZ, CX, CZ, CRX, CRZ & -- & 28 & 28 & 15 & 0 & 10000 & 0.005 & 0 & 0.6 & -- \\
H2 benchmark & 4 & H, RX, RY, RZ, CX, CZ, CRX, CRZ & -- & 36 & 45 & 18 & 0 & 10000 & 0.005 & 0 & 2.5 & -- \\
H2 excitation & 4 & SE, DE & -- & 3 & 244 & 163 & 0 & 10000 & 0.005 & 0 & 0.5 & -- \\
LiH excitation & 12 & SE, DE & -- & 100 & 500 & 300 & 30 & 100 & 0.15 & $1.00\times 10^{-30}$ & 0.083 & -- \\
\bottomrule
\end{tabular}

%% file: pqc/efficientSU2.tex
\scalebox{1.0}{
\Qcircuit @C=1.0em @R=0.2em @!R { \\
	 	\nghost{} & \lstick{} & \gate{\mathrm{R_Y}\,(\mathrm{\theta_{0}})} & \gate{\mathrm{R_Z}\,(\mathrm{\theta_{4}})} & \ctrl{1} & \qw & \qw & \gate{\mathrm{R_Y}\,(\mathrm{\theta_{8}})} & \gate{\mathrm{R_Z}\,(\mathrm{\theta_{12}})} & \qw & \qw & \qw & \qw\\
	 	\nghost{} & \lstick{} & \gate{\mathrm{R_Y}\,(\mathrm{\theta_{1}})} & \gate{\mathrm{R_Z}\,(\mathrm{\theta_{5}})} & \targ & \ctrl{1} & \qw & \gate{\mathrm{R_Y}\,(\mathrm{\theta_{9}})} & \gate{\mathrm{R_Z}\,(\mathrm{\theta_{13}})} & \qw & \qw & \qw\\
	 	\nghost{} & \lstick{} & \gate{\mathrm{R_Y}\,(\mathrm{\theta_{2}})} & \gate{\mathrm{R_Z}\,(\mathrm{\theta_{6}})} & \qw & \targ & \ctrl{1} & \gate{\mathrm{R_Y}\,(\mathrm{\theta_{10}})} & \gate{\mathrm{R_Z}\,(\mathrm{\theta_{14}})} & \qw & \qw\\
	 	\nghost{} & \lstick{} & \gate{\mathrm{R_Y}\,(\mathrm{\theta_{3}})} & \gate{\mathrm{R_Z}\,(\mathrm{\theta_{7}})} & \qw & \qw & \targ & \gate{\mathrm{R_Y}\,(\mathrm{\theta_{11}})} & \gate{\mathrm{R_Z}\,(\mathrm{\theta_{15}})} & \qw & \qw\\
\\ }}

%% file: tab/single_metric.tex
\begin{tabular}{lcccccccccccccccccc}
\toprule
ID & $|\boldsymbol{\theta}|_{r=1}$ & $G_{r=1}$ & $D_{r=1}$ & $\mathrm{Expr}_{r=1}$ & $\mathrm{Expr}_{r=2}$ & $\mathrm{Expr}_{r=3}$ & $\mathrm{Expr}_{r=4}$ & $\mathrm{Expr}_{r=5}$ & $\mathrm{Ent}_{r=1}$ & $\mathrm{Ent}_{r=2}$ & $\mathrm{Ent}_{r=3}$ & $\mathrm{Ent}_{r=4}$ & $\mathrm{Ent}_{r=5}$ & $\mathrm{Train}_{r=1}$ & $\mathrm{Train}_{r=2}$ & $\mathrm{Train}_{r=3}$ & $\mathrm{Train}_{r=4}$ & $\mathrm{Train}_{r=5}$ \\
\midrule
1 & 8 & 8 & 2 & 0.2736 & 0.1990 & 0.1916 & 0.1911 & 0.2091 & 9.13e-18 & 4.29e-17 & 1.40e-18 & 8.06e-18 & -2.16e-17 & 7.7985 & 3.4361 & 2.2648 & 1.6976 & 1.3749 \\
2 & 8 & 11 & 5 & 0.2927 & 0.0217 & 0.0058 & 0.0054 & 0.0030 & 0.6222 & 0.6883 & 0.7878 & 0.7976 & 0.8091 & 0.8327 & 0.3497 & 0.2293 & 0.2121 & 0.1628 \\
3 & 11 & 11 & 5 & 0.2423 & 0.0821 & 0.0317 & 0.0268 & 0.0155 & 0.1738 & 0.3321 & 0.4419 & 0.5153 & 0.5690 & 1.2132 & 0.5291 & 0.3326 & 0.2375 & 0.1836 \\
4 & 11 & 11 & 5 & 0.1354 & 0.0264 & 0.0136 & 0.0101 & 0.0068 & 0.3135 & 0.4682 & 0.5567 & 0.6244 & 0.6687 & 1.3931 & 0.5088 & 0.2900 & 0.2025 & 0.1555 \\
5 & 28 & 28 & 15 & 0.0555 & 0.0078 & 0.0037 & 0.0029 & 0.0029 & 0.2912 & 0.5762 & 0.7098 & 0.7710 & 0.7986 & 0.2311 & 0.0922 & 0.0544 & 0.0405 & 0.0343 \\
6 & 28 & 28 & 15 & 0.0042 & 0.0049 & 0.0032 & 0.0029 & 0.0028 & 0.6868 & 0.8057 & 0.8199 & 0.8226 & 0.8239 & 0.1675 & 0.0742 & 0.0560 & 0.0471 & 0.0414 \\
7 & 19 & 19 & 6 & 0.1084 & 0.0375 & 0.0168 & 0.0105 & 0.0065 & 0.2140 & 0.4031 & 0.5136 & 0.5797 & 0.6274 & 1.0043 & 0.4524 & 0.2796 & 0.1987 & 0.1510 \\
8 & 19 & 19 & 6 & 0.0737 & 0.0232 & 0.0099 & 0.0091 & 0.0036 & 0.2866 & 0.4641 & 0.5640 & 0.6285 & 0.6768 & 1.0471 & 0.4060 & 0.2452 & 0.1737 & 0.1331 \\
9 & 4 & 11 & 5 & 0.6841 & 0.4195 & 0.0378 & 0.0054 & 0.0046 & 1.0000 & 1.0000 & 1.0000 & 0.6826 & 0.7216 & 6.38e-18 & 1.54e-18 & 4.68e-18 & 0.1794 & 0.1820 \\
10 & 8 & 12 & 6 & 0.2209 & 0.1517 & 0.1533 & 0.1346 & 0.1308 & 0.3746 & 0.5611 & 0.6546 & 0.7009 & 0.7283 & 1.9799 & 0.8461 & 0.5160 & 0.3672 & 0.2914 \\
11 & 12 & 15 & 6 & 0.1381 & 0.0143 & 0.0050 & 0.0030 & 0.0037 & 0.5386 & 0.6700 & 0.7462 & 0.7829 & 0.8008 & 1.0072 & 0.4008 & 0.2090 & 0.1593 & 0.1270 \\
12 & 12 & 15 & 6 & 0.2110 & 0.0236 & 0.0136 & 0.0039 & 0.0035 & 0.3987 & 0.5732 & 0.6650 & 0.7178 & 0.7559 & 1.0083 & 0.3854 & 0.2248 & 0.1573 & 0.1208 \\
13 & 16 & 16 & 9 & 0.0615 & 0.0087 & 0.0043 & 0.0040 & 0.0029 & 0.4130 & 0.6248 & 0.7255 & 0.7747 & 0.7992 & 0.5640 & 0.1955 & 0.1026 & 0.0694 & 0.0535 \\
14 & 16 & 16 & 9 & 0.0114 & 0.0046 & 0.0033 & 0.0032 & 0.0040 & 0.5511 & 0.7335 & 0.7937 & 0.8127 & 0.8200 & 0.4857 & 0.1555 & 0.0897 & 0.0667 & 0.0558 \\
15 & 8 & 16 & 9 & 0.1767 & 0.1094 & 0.1285 & 0.1286 & 0.1134 & 0.7111 & 0.7650 & 0.7747 & 0.7774 & 0.7786 & 1.1041 & 0.3666 & 0.2492 & 0.1949 & 0.1629 \\
16 & 11 & 11 & 4 & 0.2910 & 0.0933 & 0.0410 & 0.0242 & 0.0137 & 0.1705 & 0.3345 & 0.4402 & 0.5183 & 0.5701 & 1.2047 & 0.5316 & 0.3313 & 0.2369 & 0.1832 \\
17 & 11 & 11 & 4 & 0.1605 & 0.0272 & 0.0144 & 0.0067 & 0.0069 & 0.2956 & 0.4559 & 0.5527 & 0.6210 & 0.6664 & 1.5212 & 0.5500 & 0.3053 & 0.2058 & 0.1548 \\
18 & 12 & 12 & 6 & 0.2279 & 0.0571 & 0.0179 & 0.0066 & 0.0039 & 0.2223 & 0.4190 & 0.5541 & 0.6415 & 0.6977 & 0.8787 & 0.3699 & 0.2094 & 0.1398 & 0.1015 \\
19 & 12 & 12 & 6 & 0.0839 & 0.0039 & 0.0037 & 0.0029 & 0.0025 & 0.3849 & 0.5976 & 0.7061 & 0.7644 & 0.7929 & 1.1059 & 0.3455 & 0.1745 & 0.1161 & 0.0914 \\
$\mathrm{AS}_{\mathrm{Expr}}$ & 9 & 18 & 9 & 0.0047 & -- & -- & -- & -- & -- & -- & -- & -- & -- & -- & -- & -- & -- & -- \\
$\mathrm{AS}_{\mathrm{Ent}}$  & 1 & 5 & 3 & -- & -- & -- & -- & -- & 1.0000 & -- & -- & -- & -- & -- & -- & -- & -- & -- \\
$\mathrm{AS}_{\mathrm{Train}}$ & 1 & 4 & 2 & -- & -- & -- & -- & -- & -- & -- & -- & -- & -- & 16.3666 & -- & -- & -- & -- \\
\bottomrule
\end{tabular}

%% file: pqc/expr_pqc.tex
\scalebox{1.0}{
\Qcircuit @C=1.0em @R=0.2em @!R { \\
	 	\nghost{} & \lstick{} & \gate{\mathrm{R_Y}\,(\theta_{0})} & \qw & \ctrl{2} & \gate{\mathrm{R_X}\,(\theta_{4})} & \qw & \ctrl{1} & \gate{\mathrm{R_X}\,(\theta_{5})} & \gate{\mathrm{R_Z}\,(\theta_{6})} & \ctrl{1} & \gate{\mathrm{R_Y}\,(\theta_{8})} & \qw & \qw\\
	 	\nghost{} & \lstick{} & \qw & \targ & \qw & \gate{\mathrm{H}} & \ctrl{1} & \targ & \ctrl{-1} & \qw & \gate{\mathrm{R_Z}\,(\theta_{1})} & \gate{\mathrm{R_Y}\,(\theta_{2})} & \qw & \qw\\
	 	\nghost{} & \lstick{} & \qw & \qw & \targ & \gate{\mathrm{R_X}\,(\theta_{3})} & \targ & \ctrl{1} & \qw & \ctrl{-2} & \ctrl{1} & \gate{\mathrm{R_Y}\,(\theta_{8})} & \qw & \qw\\
	 	\nghost{} & \lstick{} & \gate{\mathrm{R_X}\,(\theta_{0})} & \ctrl{-2} & \qw & \ctrl{-1} & \qw & \gate{\mathrm{R_X}\,(\theta_{4})} & \gate{\mathrm{R_X}\,(\theta_{6})} & \qw & \gate{\mathrm{R_Z}\,(\theta_{7})} & \qw & \qw & \qw\\
\\ }}

%% file: pqc/train_pqc.tex
\scalebox{1.0}{
\Qcircuit @C=1.0em @R=0.2em @!R { \\
	 	\nghost{} & \lstick{} & \ctrl{1} & \gate{\mathrm{R_X}\,(\theta_{0})} & \qw & \qw\\
	 	\nghost{} & \lstick{} & \control\qw & \control\qw & \qw & \qw\\
	 	\nghost{} & \lstick{} & \ctrl{1} & \ctrl{-1} & \qw & \qw\\
	 	\nghost{} & \lstick{} & \gate{\mathrm{R_X}\,(\theta_{0})} & \qw & \qw & \qw\\
\\ }}

%% file: pqc/entgl_pqc.tex
\scalebox{1.0}{
\Qcircuit @C=1.0em @R=0.2em @!R { \\
	 	\nghost{} & \lstick{} & \gate{\mathrm{H}} & \ctrl{3} & \qw & \ctrl{2} & \qw & \qw\\
	 	\nghost{} & \lstick{} & \gate{\mathrm{R_Z}\,(\theta_{0})} & \qw & \targ & \qw & \qw & \qw\\
	 	\nghost{} & \lstick{} & \qw & \qw & \qw & \targ & \qw & \qw\\
	 	\nghost{} & \lstick{} & \qw & \targ & \ctrl{-2} & \qw & \qw & \qw\\
\\ }}

%% file: tab/iqm_metric.tex
\begin{tabular}{lcccccccc}
\toprule
ID & $|\boldsymbol{\theta}|_{r=1}$ & $G_{r=1}$ & $D_{r=1}$ & $\mathrm{Expr}_{r=1}$ & $\mathrm{Expr}_{r=2}$ & $\mathrm{Expr}_{r=3}$ & $\mathrm{Expr}_{r=4}$ & $\mathrm{Expr}_{r=5}$ \\
\midrule
1 & 8 & 16 & 4 & 0.2736 & 0.1990 & 0.1916 & 0.1911 & 0.2091 \\
2 & 8 & 36 & 20 & 0.2927 & 0.0217 & 0.0058 & 0.0054 & 0.0030 \\
3 & 11 & 60 & 42 & 0.2423 & 0.0821 & 0.0317 & 0.0268 & 0.0155 \\
4 & 11 & 61 & 43 & 0.1354 & 0.0264 & 0.0136 & 0.0101 & 0.0068 \\
5 & 28 & 218 & 136 & 0.0555 & 0.0078 & 0.0037 & 0.0029 & 0.0029 \\
6 & 28 & 230 & 148 & 0.0042 & 0.0049 & 0.0032 & 0.0029 & 0.0028 \\
7 & 19 & 73 & 47 & 0.1084 & 0.0375 & 0.0168 & 0.0105 & 0.0065 \\
8 & 19 & 73 & 47 & 0.0737 & 0.0232 & 0.0099 & 0.0091 & 0.0036 \\
9 & 4 & 24 & 13 & 0.6841 & 0.4195 & 0.0378 & 0.0054 & 0.0046 \\
10 & 8 & 21 & 13 & 0.2209 & 0.1517 & 0.1533 & 0.1346 & 0.1308 \\
11 & 12 & 41 & 25 & 0.1381 & 0.0143 & 0.0050 & 0.0030 & 0.0037 \\
12 & 12 & 32 & 19 & 0.2110 & 0.0236 & 0.0136 & 0.0039 & 0.0035 \\
13 & 16 & 144 & 107 & 0.0615 & 0.0087 & 0.0043 & 0.0040 & 0.0029 \\
14 & 16 & 153 & 115 & 0.0114 & 0.0046 & 0.0033 & 0.0032 & 0.0040 \\
15 & 8 & 80 & 49 & 0.1767 & 0.1094 & 0.1285 & 0.1286 & 0.1134 \\
16 & 11 & 59 & 41 & 0.2910 & 0.0933 & 0.0410 & 0.0242 & 0.0137 \\
17 & 11 & 60 & 42 & 0.1605 & 0.0272 & 0.0144 & 0.0067 & 0.0069 \\
18 & 12 & 78 & 53 & 0.2279 & 0.0571 & 0.0179 & 0.0066 & 0.0039 \\
19 & 12 & 81 & 55 & 0.0839 & 0.0039 & 0.0037 & 0.0029 & 0.0025 \\
$\mathrm{AS}_{\mathrm{Expr}}$ & 12 & 16 & 11 & 0.0835 & -- & -- & -- & -- \\
$\mathrm{AS}_{\mathrm{Expr},10k}$ & 19 & 20 & 10 & 0.0086 & -- & -- & -- & -- \\
\bottomrule
\end{tabular}

%% file: pqc/iqm_expr_pqc.tex
\scalebox{1.0}{
\Qcircuit @C=1.0em @R=0.2em @!R { \\
	 	\nghost{} & \lstick{} & \ctrl{2} & \qw & \qw & \qw & \ctrl{2} & \gate{\mathrm{R}\,(\mathrm{\theta_{0},\theta_{2}})} & \control\qw & \gate{\mathrm{R}\,(\mathrm{\theta_{7},\theta_{8}})} & \qw & \qw & \qw & \qw & \qw\\
	 	\nghost{} & \lstick{} & \qw & \ctrl{1} & \ctrl{1} & \gate{\mathrm{R}\,(\mathrm{\theta_{3},\theta_{2}})} & \qw & \gate{\mathrm{R}\,(\mathrm{\theta_{5},\theta_{0}})} & \qw & \qw & \qw & \qw & \ctrl{1} & \qw & \qw\\
	 	\nghost{} & \lstick{} & \control\qw & \control\qw & \control\qw & \gate{\mathrm{R}\,(\mathrm{\theta_{0},\theta_{1}})} & \control\qw & \qw & \ctrl{-2} & \gate{\mathrm{R}\,(\mathrm{\theta_{6},\theta_{1}})} & \ctrl{1} & \gate{\mathrm{R}\,(\mathrm{\theta_{9},\theta_{5}})} & \control\qw & \qw & \qw\\
	 	\nghost{} & \lstick{} & \gate{\mathrm{R}\,(\mathrm{\theta_{4},\theta_{0}})} & \qw & \qw & \qw & \qw & \qw & \qw & \qw & \control\qw & \gate{\mathrm{R}\,(\mathrm{\theta_{10},\theta_{11}})} & \qw & \qw & \qw\\
\\ }}

%% file: pqc/iqm_sim_expr_pqc.tex
\scalebox{1.0}{
\Qcircuit @C=1.0em @R=0.2em @!R { \\
	 	\nghost{} & \lstick{} & \gate{\mathrm{R}\,(\mathrm{\theta_{6},\theta_{5}})} & \gate{\mathrm{R}\,(\mathrm{\theta_{10},\theta_{4}})} & \qw & \qw & \qw & \qw & \ctrl{2} & \gate{\mathrm{R}\,(\mathrm{\theta_{14},\theta_{15}})} & \qw & \qw & \qw & \qw\\
	 	\nghost{} & \lstick{} & \gate{\mathrm{R}\,(\mathrm{\theta_{3},\theta_{1}})} & \gate{\mathrm{R}\,(\mathrm{\theta_{8},\theta_{5}})} & \qw & \qw & \qw & \ctrl{1} & \qw & \gate{\mathrm{R}\,(\mathrm{\theta_{13},\theta_{11}})} & \qw & \qw & \qw & \qw\\
	 	\nghost{} & \lstick{} & \gate{\mathrm{R}\,(\mathrm{\theta_{0},\theta_{2}})} & \gate{\mathrm{R}\,(\mathrm{\theta_{5},\theta_{0}})} & \gate{\mathrm{R}\,(\mathrm{\theta_{7},\theta_{3}})} & \control\qw & \gate{\mathrm{R}\,(\mathrm{\theta_{2},\theta_{9}})} & \control\qw & \control\qw & \gate{\mathrm{R}\,(\mathrm{\theta_{12},\theta_{7}})} & \control\qw & \gate{\mathrm{R}\,(\mathrm{\theta_{16},\theta_{17}})} & \qw & \qw\\
	 	\nghost{} & \lstick{} & \gate{\mathrm{R}\,(\mathrm{\theta_{0},\theta_{1}})} & \gate{\mathrm{R}\,(\mathrm{\theta_{1},\theta_{4}})} & \qw & \ctrl{-1} & \gate{\mathrm{R}\,(\mathrm{\theta_{11},\theta_{2}})} & \qw & \qw & \qw & \ctrl{-1} & \gate{\mathrm{R}\,(\mathrm{\theta_{18},\theta_{2}})} & \qw & \qw\\
\\ }}

%% file: tab/multi_metric.tex
\begin{tabular}{lcccccccccccccccccc}
\toprule
ID & $|\boldsymbol{\theta}|_{r=1}$ & $G_{r=1}$ & $D_{r=1}$ & $\mathcal{L}_{Expr}^{r=1}$ & $\mathcal{L}_{Expr}^{r=2}$ & $\mathcal{L}_{Expr}^{r=3}$ & $\mathcal{L}_{Expr}^{r=4}$ & $\mathcal{L}_{Expr}^{r=5}$ & $\mathcal{L}_{Entgl}^{r=1}$ & $\mathcal{L}_{Entgl}^{r=2}$ & $\mathcal{L}_{Entgl}^{r=3}$ & $\mathcal{L}_{Entgl}^{r=4}$ & $\mathcal{L}_{Entgl}^{r=5}$ & $\mathcal{L}_{Train}^{r=1}$ & $\mathcal{L}_{Train}^{r=2}$ & $\mathcal{L}_{Train}^{r=3}$ & $\mathcal{L}_{Train}^{r=4}$ & $\mathcal{L}_{Train}^{r=5}$ \\
\midrule
1 & 8 & 8 & 2 & 0.4522 & 0.4203 & 0.4165 & 0.4163 & 0.4253 & 1.0000 & 1.0000 & 1.0000 & 1.0000 & 1.0000 & 0 & 0 & 0.0941 & 0.3210 & 0.4500 \\
2 & 8 & 11 & 5 & 0.4590 & 0.1983 & 0.0652 & 0.0594 & 0.0005 & 0.2680 & 0.1902 & 0.0731 & 0.0616 & 0.0481 & 0.6669 & 0.8601 & 0.9083 & 0.9152 & 0.9349 \\
3 & 11 & 11 & 5 & 0.4400 & 0.3316 & 0.2363 & 0.2195 & 0.1648 & 0.7956 & 0.6093 & 0.4801 & 0.3938 & 0.3306 & 0.5147 & 0.7884 & 0.8670 & 0.9050 & 0.9266 \\
4 & 11 & 11 & 5 & 0.3817 & 0.2179 & 0.1518 & 0.1213 & 0.0816 & 0.6311 & 0.4492 & 0.3451 & 0.2654 & 0.2133 & 0.4427 & 0.7965 & 0.8840 & 0.9190 & 0.9378 \\
5 & 28 & 28 & 15 & 0.2924 & 0.0953 & 0.0203 & 0 & 0 & 0.6575 & 0.3221 & 0.1649 & 0.0929 & 0.0604 & 0.9076 & 0.9631 & 0.9782 & 0.9838 & 0.9863 \\
6 & 28 & 28 & 15 & 0.0340 & 0.0496 & 0.0073 & 0 & 0 & 0.1920 & 0.0521 & 0.0354 & 0.0322 & 0.0308 & 0.9330 & 0.9703 & 0.9776 & 0.9811 & 0.9834 \\
7 & 19 & 19 & 6 & 0.3594 & 0.2532 & 0.1729 & 0.1251 & 0.0778 & 0.7482 & 0.5257 & 0.3958 & 0.3180 & 0.2618 & 0.5983 & 0.8190 & 0.8882 & 0.9205 & 0.9396 \\
8 & 19 & 19 & 6 & 0.3208 & 0.2050 & 0.1197 & 0.1108 & 0.0179 & 0.6629 & 0.4541 & 0.3364 & 0.2606 & 0.2037 & 0.5812 & 0.8376 & 0.9019 & 0.9305 & 0.9468 \\
9 & 4 & 11 & 5 & 0.5440 & 0.4951 & 0.2538 & 0.0594 & 0.0420 & 0 & 0 & 0 & 0.1970 & 0.1511 & 1.0000 & 1.0000 & 1.0000 & 0.9282 & 0.9272 \\
10 & 8 & 12 & 6 & 0.4308 & 0.3931 & 0.3942 & 0.3811 & 0.3783 & 0.5593 & 0.3399 & 0.2299 & 0.1754 & 0.1432 & 0.2080 & 0.6616 & 0.7936 & 0.8531 & 0.8835 \\
11 & 12 & 15 & 6 & 0.3837 & 0.1566 & 0.0507 & 0 & 0.0204 & 0.3664 & 0.2117 & 0.1221 & 0.0790 & 0.0579 & 0.5971 & 0.8397 & 0.9164 & 0.9363 & 0.9492 \\
12 & 12 & 15 & 6 & 0.4262 & 0.2065 & 0.1511 & 0.0268 & 0.0156 & 0.5310 & 0.3257 & 0.2177 & 0.1556 & 0.1107 & 0.5967 & 0.8458 & 0.9101 & 0.9371 & 0.9517 \\
13 & 16 & 16 & 9 & 0.3026 & 0.1071 & 0.0359 & 0.0296 & 0 & 0.5141 & 0.2650 & 0.1464 & 0.0886 & 0.0597 & 0.7744 & 0.9218 & 0.9590 & 0.9723 & 0.9786 \\
14 & 16 & 16 & 9 & 0.1339 & 0.0437 & 0.0085 & 0.0073 & 0.0280 & 0.3517 & 0.1370 & 0.0663 & 0.0438 & 0.0353 & 0.8057 & 0.9378 & 0.9641 & 0.9733 & 0.9777 \\
15 & 8 & 16 & 9 & 0.4084 & 0.3604 & 0.3765 & 0.3766 & 0.3640 & 0.1635 & 0.1000 & 0.0886 & 0.0855 & 0.0840 & 0.5584 & 0.8534 & 0.9003 & 0.9220 & 0.9348 \\
16 & 11 & 11 & 4 & 0.4584 & 0.3444 & 0.2620 & 0.2094 & 0.1518 & 0.7995 & 0.6065 & 0.4821 & 0.3902 & 0.3292 & 0.5181 & 0.7874 & 0.8675 & 0.9052 & 0.9267 \\
17 & 11 & 11 & 4 & 0.3988 & 0.2210 & 0.1571 & 0.0812 & 0.0828 & 0.6523 & 0.4637 & 0.3498 & 0.2694 & 0.2160 & 0.3915 & 0.7800 & 0.8779 & 0.9177 & 0.9381 \\
18 & 12 & 12 & 6 & 0.4339 & 0.2951 & 0.1789 & 0.0788 & 0.0267 & 0.7385 & 0.5071 & 0.3481 & 0.2453 & 0.1791 & 0.6485 & 0.8521 & 0.9163 & 0.9441 & 0.9594 \\
19 & 12 & 12 & 6 & 0.3338 & 0.0261 & 0.0207 & 0 & 0 & 0.5472 & 0.2969 & 0.1693 & 0.1007 & 0.0672 & 0.5577 & 0.8618 & 0.9302 & 0.9535 & 0.9635 \\
$\mathrm{AS}_{\mathrm{Expr}-\mathrm{Ent}}$ & 10 & 23 & 11 & 0.0174 & -- & -- & -- & -- & 0 & -- & -- & -- & -- & -- & -- & -- & -- & -- \\
$\mathrm{AS}_{\mathrm{Expr}-\mathrm{Train}}$ & 6 & 16 & 7 & 0.1817 & -- & -- & -- & -- & -- & -- & -- & -- & -- & 0.0019 & -- & -- & -- & -- \\
\bottomrule
\end{tabular}

%% file: pqc/expr_train_pqc.tex
\scalebox{1.0}{
\Qcircuit @C=1.0em @R=0.2em @!R { \\
	 	\nghost{} & \lstick{} & \gate{\mathrm{R_X}\,(\theta_{2})} & \qw & \targ & \gate{\mathrm{R_Z}\,(\theta_{4})} & \qw & \gate{\mathrm{R_Z}\,(\theta_{4})} & \gate{\mathrm{R_Y}\,(\theta_{1})} & \qw & \qw & \qw\\
	 	\nghost{} & \lstick{} & \gate{\mathrm{R_X}\,(\theta_{0})} & \gate{\mathrm{R_X}\,(\theta_{1})} & \ctrl{-1} & \gate{\mathrm{H}} & \gate{\mathrm{R_Z}\,(\theta_{3})} & \qw & \ctrl{1} & \gate{\mathrm{R_X}\,(\theta_{3})} & \qw & \qw\\
	 	\nghost{} & \lstick{} & \qw & \gate{\mathrm{R_X}\,(\theta_{3})} & \gate{\mathrm{R_Y}\,(\theta_{2})} & \qw & \qw & \ctrl{-2} & \gate{\mathrm{R_X}\,(\theta_{5})} & \qw & \qw & \qw\\
	 	\nghost{} & \lstick{} & \gate{\mathrm{H}} & \ctrl{-1} & \gate{\mathrm{R_Z}\,(\theta_{4})} & \gate{\mathrm{R_X}\,(\theta_{4})} & \qw & \qw & \qw & \qw & \qw & \qw\\
\\ }}

%% file: pqc/expr_entgl_pqc.tex
\scalebox{0.9}{
\Qcircuit @C=1.0em @R=0.2em @!R { \\
	 	\nghost{} & \lstick{} & \gate{\mathrm{R_Y}\,(\theta_{4})} & \gate{\mathrm{R_Z}\,(\theta_{4})} & \qw & \qw & \qw & \ctrl{1} & \gate{\mathrm{R_X}\,(\theta_{5})} & \gate{\mathrm{R_X}\,(\theta_{8})} & \qw & \targ & \gate{\mathrm{R_Z}\,(\theta_{1})} & \control\qw & \qw & \qw & \qw\\
	 	\nghost{} & \lstick{} & \gate{\mathrm{R_Z}\,(\theta_{2})} & \gate{\mathrm{H}} & \gate{\mathrm{R_Y}\,(\theta_{5})} & \gate{\mathrm{R_Z}\,(\theta_{6})} & \ctrl{2} & \gate{\mathrm{R_X}\,(\theta_{7})} & \qw & \ctrl{2} & \gate{\mathrm{R_X}\,(\theta_{0})} & \qw & \qw & \ctrl{-1} & \qw & \qw & \qw\\
	 	\nghost{} & \lstick{} & \gate{\mathrm{R_Y}\,(\theta_{0})} & \gate{\mathrm{R_Z}\,(\theta_{3})} & \targ & \qw & \qw & \qw & \ctrl{-2} & \qw & \ctrl{-1} & \qw & \gate{\mathrm{R_X}\,(\theta_{2})} & \gate{\mathrm{H}} & \gate{\mathrm{R_Y}\,(\theta_{3})} & \qw & \qw\\
	 	\nghost{} & \lstick{} & \gate{\mathrm{H}} & \qw & \ctrl{-1} & \qw & \gate{\mathrm{R_Z}\,(\theta_{5})} & \qw & \qw & \gate{\mathrm{R_X}\,(\theta_{9})} & \qw & \ctrl{-3} & \gate{\mathrm{R_X}\,(\theta_{9})} & \qw & \qw & \qw & \qw\\
\\ }}

%% file: tab/power_of_qnn.tex
\begin{tabular}{lccccccccccccc}
\toprule
ID & $|\boldsymbol{\theta}|_{r=1}$ & $G_{r=1}$ & $D_{r=1}$ & $\mathcal{L}_{Expr}^{r=1}$ & $\mathcal{L}_{Expr}^{r=2}$ & $\mathcal{L}_{Expr}^{r=3}$ & $\mathcal{L}_{Expr}^{r=4}$ & $\mathcal{L}_{Expr}^{r=5}$ & $\mathcal{L}_{Train}^{r=1}$ & $\mathcal{L}_{Train}^{r=2}$ & $\mathcal{L}_{Train}^{r=3}$ & $\mathcal{L}_{Train}^{r=4}$ & $\mathcal{L}_{Train}^{r=5}$ \\
\midrule
1 & 8 & 8 & 2 & 0.7143 & 0.6132 & 0.4072 & 0.3016 & 0.2378 & 0.2181 & 0.2392 & 0.4711 & 0.5117 & 0.6663 \\
2 & 8 & 11 & 5 & 0.7268 & 0.5185 & 0.3738 & 0.2499 & 0.0872 & 0.9524 & 0.9164 & 0.9341 & 0.9375 & 0.9515 \\
3 & 11 & 11 & 5 & 0.6895 & 0.5212 & 0.3521 & 0.2477 & 0.1821 & 0.8515 & 0.8900 & 0.9143 & 0.9227 & 0.9352 \\
4 & 11 & 11 & 5 & 0.7047 & 0.5309 & 0.3316 & 0.2443 & 0.1490 & 0.8514 & 0.8877 & 0.9128 & 0.9338 & 0.9526 \\
5 & 28 & 28 & 15 & 0.4215 & 0.1770 & 0.1279 & 0.0352 & 0.0280 & 0.9562 & 0.9707 & 0.9825 & 0.9773 & 0.9820 \\
6 & 28 & 28 & 15 & 0.4328 & 0.1559 & 0.0381 & 0.0186 & 0.0343 & 0.9399 & 0.9685 & 0.9803 & 0.9788 & 0.9835 \\
7 & 19 & 19 & 6 & 0.5436 & 0.3078 & 0.1260 & 0.0747 & 0.0873 & 0.7720 & 0.8839 & 0.9038 & 0.9214 & 0.9382 \\
8 & 19 & 19 & 6 & 0.5284 & 0.2871 & 0.2117 & 0.1914 & 0.0622 & 0.7962 & 0.8538 & 0.9080 & 0.9211 & 0.9399 \\
9 & 4 & 11 & 5 & 0.8218 & 0.7184 & 0.6101 & 0.5251 & 0.4378 & 0.4206 & 0.8179 & 0.8822 & 0.8998 & 0.9217 \\
10 & 8 & 12 & 6 & 0.7098 & 0.5671 & 0.4008 & 0.3274 & 0.2758 & 0.6112 & 0.7766 & 0.8731 & 0.8880 & 0.9284 \\
11 & 12 & 15 & 6 & 0.6204 & 0.3193 & 0.1734 & 0.0310 & 0.0301 & 0.8918 & 0.9124 & 0.9278 & 0.9401 & 0.9483 \\
12 & 12 & 15 & 6 & 0.6221 & 0.4290 & 0.2873 & 0.0973 & 0.0460 & 0.7679 & 0.8585 & 0.9054 & 0.9262 & 0.9404 \\
13 & 16 & 16 & 9 & 0.6402 & 0.4365 & 0.2604 & 0.1306 & 0.0755 & 0.8652 & 0.9037 & 0.9525 & 0.9514 & 0.9771 \\
14 & 16 & 16 & 9 & 0.6572 & 0.4263 & 0.2150 & 0.1790 & 0.0943 & 0.8805 & 0.9107 & 0.9399 & 0.9576 & 0.9662 \\
15 & 8 & 16 & 9 & 0.7210 & 0.5113 & 0.3170 & 0.1616 & 0.0680 & 0.9524 & 0.9419 & 0.9577 & 0.9601 & 0.9662 \\
16 & 11 & 11 & 4 & 0.6889 & 0.4743 & 0.4009 & 0.2306 & 0.1594 & 0.7938 & 0.8644 & 0.9192 & 0.9256 & 0.9408 \\
17 & 11 & 11 & 4 & 0.6765 & 0.5171 & 0.3342 & 0.2139 & 0.2062 & 0.8797 & 0.8860 & 0.8839 & 0.9279 & 0.9466 \\
18 & 12 & 12 & 6 & 0.6721 & 0.4889 & 0.4003 & 0.2021 & 0.1504 & 0.9352 & 0.9146 & 0.9369 & 0.9471 & 0.9503 \\
19 & 12 & 12 & 6 & 0.6689 & 0.4661 & 0.3081 & 0.1568 & 0.0557 & 0.8241 & 0.8966 & 0.9427 & 0.9483 & 0.9586 \\
QNN & 8 & 14 & 7 & 0.7153 & 0.5309 & 0.3434 & 0.1510 & 0.0676 & 0.7557 & 0.8931 & 0.9354 & 0.9416 & 0.9561 \\
$\mathrm{AS}$ & 7 & 24 & 10 & 0.0266 & -- & -- & -- & -- & 0.3990 & -- & -- & -- & -- \\
\bottomrule
\end{tabular}

%% file: pqc/qnn_pqc.tex
\scalebox{1.0}{
\Qcircuit @C=1.0em @R=0.2em @!R { \\
	 	\nghost{} & \lstick{} & \gate{\mathrm{R_Z}\,(\mathrm{\theta_{1}})} & \gate{\mathrm{R_Z}\,(\mathrm{\theta_{1}})} & \gate{\mathrm{R_Y}\,(\mathrm{\theta_{4}})} & \gate{\mathrm{R_Z}\,(\mathrm{\theta_{3}})} & \qw & \qw & \control\qw & \gate{\mathrm{R_Z}\,(\mathrm{\theta_{3}})} & \qw & \qw & \qw & \qw\\
	 	\nghost{} & \lstick{} & \gate{\mathrm{R_X}\,(\mathrm{\theta_{1}})} & \gate{\mathrm{R_X}\,(\mathrm{\theta_{1}})} & \gate{\mathrm{R_Z}\,(\mathrm{\theta_{2}})} & \gate{\mathrm{R_Z}\,(\mathrm{\theta_{2}})} & \gate{\mathrm{H}} & \gate{\mathrm{R_X}\,(\mathrm{\theta_{2}})} & \ctrl{-1} & \gate{\mathrm{R_Y}\,(\mathrm{\theta_{6}})} & \qw & \qw & \qw & \qw\\
	 	\nghost{} & \lstick{} & \qw & \ctrl{-1} & \qw & \ctrl{1} & \qw & \targ & \ctrl{1} & \gate{\mathrm{R_X}\,(\mathrm{\theta_{4}})} & \gate{\mathrm{R_Z}\,(\mathrm{\theta_{4}})} & \gate{\mathrm{R_X}\,(\mathrm{\theta_{4}})} & \qw & \qw\\
	 	\nghost{} & \lstick{} & \gate{\mathrm{R_Z}\,(\mathrm{\theta_{0}})} & \gate{\mathrm{R_Y}\,(\mathrm{\theta_{0}})} & \gate{\mathrm{R_Y}\,(\mathrm{\theta_{0}})} & \control\qw & \gate{\mathrm{R_Y}\,(\mathrm{\theta_{3}})} & \ctrl{-1} & \gate{\mathrm{R_X}\,(\mathrm{\theta_{5}})} & \gate{\mathrm{R_Y}\,(\mathrm{\theta_{6}})} & \ctrl{-1} & \qw & \qw & \qw\\
\\ }}

%% file: tab/h2.tex
\begin{tabular}{lccccccccccccc}
\toprule
ID & $|\boldsymbol{\theta}|_{r=1}$ & $G_{r=1}$ & $D_{r=1}$ & $\mathcal{L}_{Expr}^{r=1}$ & $\mathcal{L}_{Expr}^{r=2}$ & $\mathcal{L}_{Expr}^{r=3}$ & $\mathcal{L}_{Expr}^{r=4}$ & $\mathcal{L}_{Expr}^{r=5}$ & $\mathcal{L}_{Train}^{r=1}$ & $\mathcal{L}_{Train}^{r=2}$ & $\mathcal{L}_{Train}^{r=3}$ & $\mathcal{L}_{Train}^{r=4}$ & $\mathcal{L}_{Train}^{r=5}$ \\
\midrule
1 & 8 & 16 & 4 & 0.4522 & 0.4203 & 0.4165 & 0.4163 & 0.4253 & 0.5206 & 0.8116 & 0.8819 & 0.9100 & 0.9274 \\
2 & 8 & 28 & 14 & 0.4590 & 0.2020 & 0.0403 & 0.0024 & 0.0164 & 0.8728 & 0.9627 & 0.9793 & 0.9851 & 0.9875 \\
3 & 11 & 52 & 34 & 0.4396 & 0.3425 & 0.2571 & 0.2103 & 0.1632 & 0.9455 & 0.9785 & 0.9872 & 0.9907 & 0.9926 \\
4 & 11 & 52 & 40 & 0.3801 & 0.2526 & 0.1651 & 0.0991 & 0.0366 & 0.9398 & 0.9791 & 0.9875 & 0.9909 & 0.9927 \\
5 & 28 & 179 & 100 & 0.2796 & 0.1199 & 0.0277 & 0.0237 & 0.0604 & 0.9896 & 0.9952 & 0.9967 & 0.9972 & 0.9974 \\
6 & 28 & 182 & 102 & 0.0222 & 0.0290 & 0.0053 & 0.0401 & 0.0110 & 0.9909 & 0.9957 & 0.9965 & 0.9968 & 0.9970 \\
7 & 19 & 68 & 30 & 0.3620 & 0.2514 & 0.1707 & 0.1184 & 0.0886 & 0.9583 & 0.9823 & 0.9886 & 0.9915 & 0.9931 \\
8 & 19 & 68 & 32 & 0.3211 & 0.2024 & 0.1569 & 0.1111 & 0.0593 & 0.9587 & 0.9830 & 0.9889 & 0.9916 & 0.9932 \\
9 & 4 & 15 & 6 & 0.5440 & 0.4947 & 0.2172 & 0.0702 & 0.0034 & 0.9987 & 0.9628 & 0.9869 & 0.9780 & 0.9774 \\
HEA & 8 & 12 & 6 & 0.4308 & 0.3931 & 0.3942 & 0.3811 & 0.3783 & 0.7612 & 0.8974 & 0.9370 & 0.9546 & 0.9636 \\
11 & 12 & 36 & 15 & 0.3919 & 0.1344 & 0.0002 & 0.0620 & 0.0286 & 0.9184 & 0.9692 & 0.9811 & 0.9858 & 0.9884 \\
12 & 12 & 27 & 10 & 0.4237 & 0.2186 & 0.0972 & 0.0381 & 0.0229 & 0.8857 & 0.9629 & 0.9777 & 0.9837 & 0.9869 \\
13 & 16 & 112 & 83 & 0.3065 & 0.1241 & 0.0379 & 0.0257 & 0.0143 & 0.9770 & 0.9914 & 0.9947 & 0.9960 & 0.9966 \\
14 & 16 & 120 & 75 & 0.1714 & 0.0142 & 0.0113 & 0.0399 & 0.0185 & 0.9820 & 0.9928 & 0.9951 & 0.9960 & 0.9964 \\
15 & 8 & 48 & 27 & 0.4187 & 0.3649 & 0.3657 & 0.3615 & 0.3621 & 0.9287 & 0.9640 & 0.9752 & 0.9804 & 0.9836 \\
16 & 11 & 53 & 27 & 0.4553 & 0.3345 & 0.2528 & 0.1986 & 0.1507 & 0.9456 & 0.9788 & 0.9872 & 0.9907 & 0.9926 \\
17 & 11 & 54 & 28 & 0.3910 & 0.2147 & 0.1257 & 0.1016 & 0.0723 & 0.9432 & 0.9796 & 0.9878 & 0.9911 & 0.9929 \\
18 & 12 & 67 & 49 & 0.4388 & 0.2964 & 0.1887 & 0.1077 & 0.0433 & 0.9608 & 0.9848 & 0.9909 & 0.9935 & 0.9949 \\
19 & 6 & 12 & 6 & 0.3219 & 0.0815 & 0.0437 & 0.0321 & 0.0217 & 0.9592 & 0.9861 & 0.9919 & 0.9939 & 0.9949 \\
SU2 HEA & 16 & 44 & 18 & 0.1850 & 0.0294 & 0.0019 & 0.0151 & 0.0152 & 0.9426 & 0.9763 & 0.9845 & 0.9877 & 0.9896 \\
UCCSD & 3 & 244 & 163 & 0.7037 & 0.6977 & 0.6955 & 0.6955 & 0.7011 & 0.7677 & 0.8597 & 0.8843 & 0.8936 & 0.8983 \\
$\mathrm{AS}_\mathrm{bench}$ & 5 & 27 & 18 & 0.3490 & -- & -- & -- & -- & 0 & -- & -- & -- & -- \\
$\mathrm{AS}_\mathrm{excitation}$ & 3 & 87 & 38 & 0.7051 & -- & -- & -- & -- & 0 & -- & -- & -- & -- \\
\bottomrule
\end{tabular}

%% file: pqc/vqe_h2.tex
\scalebox{1.0}{
\Qcircuit @C=1.0em @R=0.2em @!R { \\
	 	\nghost{} & \lstick{} & \qw & \gate{\mathrm{R_X}\,(\mathrm{\theta_{1}})} & \qw & \qw & \targ & \qw & \qw & \qw & \qw & \qw & \qw\\
	 	\nghost{} & \lstick{} & \gate{\mathrm{R_X}\,(\mathrm{\theta_{0}})} & \qw & \gate{\mathrm{R_X}\,(\mathrm{\theta_{0}})} & \gate{\mathrm{R_X}\,(\mathrm{\theta_{0}})} & \ctrl{-1} & \gate{\mathrm{R_X}\,(\mathrm{\theta_{2}})} & \gate{\mathrm{R_Z}\,(\mathrm{\theta_{3}})} & \qw & \qw & \qw & \qw\\
	 	\nghost{} & \lstick{} & \qw & \ctrl{-2} & \gate{\mathrm{H}} & \qw & \qw & \ctrl{-1} & \gate{\mathrm{R_Y}\,(\mathrm{\theta_{3}})} & \ctrl{1} & \qw & \qw & \qw\\
	 	\nghost{} & \lstick{} & \qw & \gate{\mathrm{R_Y}\,(\mathrm{\theta_{0}})} & \qw & \qw & \qw & \qw & \qw & \gate{\mathrm{R_X}\,(\mathrm{\theta_{4}})} & \gate{\mathrm{R_Z}\,(\mathrm{\theta_{1}})} & \qw & \qw\\
\\ }}

%% file: pqc/vqe_h2_excitation.tex
\scalebox{1.0}{
\Qcircuit @C=1.0em @R=1.0em @!R { \\
	 	\nghost{} & \lstick{} & \multigate{1}{\mathrm{SingleExc}\,(\mathrm{\theta_{0}})}_<<<{1} & \multigate{1}{\mathrm{SingleExc}\,(\mathrm{\theta_{0}})}_<<<{1} & \multigate{1}{\mathrm{SingleExc}\,(\mathrm{\theta_{2}})}_<<<{1} & \multigate{3}{\mathrm{DoubleExc}\,(\mathrm{\theta_{1}})}_<<<{2} & \qw & \qw\\
	 	\nghost{} & \lstick{} & \ghost{\mathrm{SingleExc}\,(\mathrm{\theta_{0}})}_<<<{0} & \ghost{\mathrm{SingleExc}\,(\mathrm{\theta_{0}})}_<<<{0} & \ghost{\mathrm{SingleExc}\,(\mathrm{\theta_{2}})}_<<<{0} & \ghost{\mathrm{DoubleExc}\,(\mathrm{\theta_{1}})}_<<<{0} & \qw & \qw\\
	 	\nghost{} & \lstick{} & \multigate{1}{\mathrm{SingleExc}\,(\mathrm{\theta_{0}})}_<<<{1} & \multigate{1}{\mathrm{SingleExc}\,(\mathrm{\theta_{0}})}_<<<{1} & \multigate{1}{\mathrm{SingleExc}\,(\mathrm{\theta_{1}})}_<<<{1} & \ghost{\mathrm{DoubleExc}\,(\mathrm{\theta_{1}})}_<<<{3} & \qw & \qw\\
	 	\nghost{} & \lstick{} & \ghost{\mathrm{SingleExc}\,(\mathrm{\theta_{0}})}_<<<{0} & \ghost{\mathrm{SingleExc}\,(\mathrm{\theta_{0}})}_<<<{0} & \ghost{\mathrm{SingleExc}\,(\mathrm{\theta_{1}})}_<<<{0} & \ghost{\mathrm{DoubleExc}\,(\mathrm{\theta_{1}})}_<<<{1} & \qw & \qw\\
\\ }}

%% file: tab/lih.tex
\begin{tabular}{lccccc}
\toprule
ID & $|\boldsymbol{\theta}|$ & $G$ & $D$ & $\mathcal{L}_{Expr}$ & $\mathcal{L}_{Train}$ \\
\midrule
HEA & 24 & 36 & 14 & 0.4869 & 0.7332 \\
SU2 HEA & 48 & 140 & 42 & 0.3687 & 0.9367 \\
UCCSD & 92 & 30170 & 22336 & 0.8503 & 0.9768 \\
$\mathrm{AS}_{\mathrm{excitation}}$ & 32 & 1964 & 767 & 0.8272 & 0.8522 \\
\bottomrule
\end{tabular}

%% file: pqc/lih_pqc.tex
\scalebox{0.55}{
\Qcircuit @C=1.0em @R=1.0em @!R { \\
		\nghost{} & \lstick{} & \qw & \qw & \qw & \multigate{8}{\mathrm{DoubleExc}\,(\mathrm{\theta_{0}})}_<<<{2} & \multigate{2}{\mathrm{SingleExc}\,(\mathrm{\theta_{1}})}_<<<{1} & \multigate{4}{\mathrm{DoubleExc}\,(\mathrm{\theta_{2}})}_<<<{0} & \multigate{9}{\mathrm{DoubleExc}\,(\mathrm{\theta_{4}})}_<<<{2} & \multigate{5}{\mathrm{DoubleExc}\,(\mathrm{\theta_{5}})}_<<<{2} & \multigate{10}{\mathrm{DoubleExc}\,(\mathrm{\theta_{6}})}_<<<{2} & \qw & \qw & \qw & \qw\\
		\nghost{} & \lstick{} & \qw & \qw & \qw & \ghost{\mathrm{DoubleExc}\,(\mathrm{\theta_{0}})}_<<<{0} & \ghost{\mathrm{SingleExc}\,(\mathrm{\theta_{1}})} & \ghost{\mathrm{DoubleExc}\,(\mathrm{\theta_{2}})}_<<<{2} & \ghost{\mathrm{DoubleExc}\,(\mathrm{\theta_{4}})} & \ghost{\mathrm{DoubleExc}\,(\mathrm{\theta_{5}})}_<<<{3} & \ghost{\mathrm{DoubleExc}\,(\mathrm{\theta_{6}})}_<<<{0} & \qw & \qw & \qw & \qw\\
		\nghost{} & \lstick{} & \qw & \qw & \qw & \ghost{\mathrm{DoubleExc}\,(\mathrm{\theta_{0}})} & \ghost{\mathrm{SingleExc}\,(\mathrm{\theta_{1}})}_<<<{0} & \ghost{\mathrm{DoubleExc}\,(\mathrm{\theta_{2}})} & \ghost{\mathrm{DoubleExc}\,(\mathrm{\theta_{4}})} & \ghost{\mathrm{DoubleExc}\,(\mathrm{\theta_{5}})} & \ghost{\mathrm{DoubleExc}\,(\mathrm{\theta_{6}})} & \qw & \qw & \qw & \qw\\
		\nghost{} & \lstick{} & \qw & \qw & \qw & \ghost{\mathrm{DoubleExc}\,(\mathrm{\theta_{0}})} & \qw & \ghost{\mathrm{DoubleExc}\,(\mathrm{\theta_{2}})}_<<<{3} & \ghost{\mathrm{DoubleExc}\,(\mathrm{\theta_{4}})}_<<<{0} & \ghost{\mathrm{DoubleExc}\,(\mathrm{\theta_{5}})} & \ghost{\mathrm{DoubleExc}\,(\mathrm{\theta_{6}})} & \qw & \qw & \qw & \qw\\
		\nghost{} & \lstick{} & \qw & \multigate{6}{\mathrm{DoubleExc}\,(\mathrm{\theta_{0}})}_<<<{2} & \qw & \ghost{\mathrm{DoubleExc}\,(\mathrm{\theta_{0}})} & \qw & \ghost{\mathrm{DoubleExc}\,(\mathrm{\theta_{2}})}_<<<{1} & \ghost{\mathrm{DoubleExc}\,(\mathrm{\theta_{4}})} & \ghost{\mathrm{DoubleExc}\,(\mathrm{\theta_{5}})}_<<<{0} & \ghost{\mathrm{DoubleExc}\,(\mathrm{\theta_{6}})} & \qw & \qw & \qw & \qw\\
		\nghost{} & \lstick{} & \qw & \ghost{\mathrm{DoubleExc}\,(\mathrm{\theta_{0}})}_<<<{0} & \qw & \ghost{\mathrm{DoubleExc}\,(\mathrm{\theta_{0}})} & \qw & \qw & \ghost{\mathrm{DoubleExc}\,(\mathrm{\theta_{4}})} & \ghost{\mathrm{DoubleExc}\,(\mathrm{\theta_{5}})}_<<<{1} & \ghost{\mathrm{DoubleExc}\,(\mathrm{\theta_{6}})} & \qw & \qw & \qw & \qw\\
		\nghost{} & \lstick{} & \multigate{2}{\mathrm{SingleExc}\,(\mathrm{\theta_{0}})}_<<<{1} & \ghost{\mathrm{DoubleExc}\,(\mathrm{\theta_{0}})} & \qw & \ghost{\mathrm{DoubleExc}\,(\mathrm{\theta_{0}})}_<<<{3} & \qw & \qw & \ghost{\mathrm{DoubleExc}\,(\mathrm{\theta_{4}})} & \qw & \ghost{\mathrm{DoubleExc}\,(\mathrm{\theta_{6}})} & \qw & \qw & \qw & \qw\\
		\nghost{} & \lstick{} & \ghost{\mathrm{SingleExc}\,(\mathrm{\theta_{0}})} & \ghost{\mathrm{DoubleExc}\,(\mathrm{\theta_{0}})} & \qw & \ghost{\mathrm{DoubleExc}\,(\mathrm{\theta_{0}})} & \multigate{4}{\mathrm{DoubleExc}\,(\mathrm{\theta_{3}})}_<<<{2} & \multigate{4}{\mathrm{DoubleExc}\,(\mathrm{\theta_{0}})}_<<<{2} & \ghost{\mathrm{DoubleExc}\,(\mathrm{\theta_{4}})}_<<<{3} & \qw & \ghost{\mathrm{DoubleExc}\,(\mathrm{\theta_{6}})} & \qw & \multigate{3}{\mathrm{DoubleExc}\,(\mathrm{\theta_{0}})}_<<<{2} & \multigate{4}{\mathrm{DoubleExc}\,(\mathrm{\theta_{7}})}_<<<{2} & \qw\\
		\nghost{} & \lstick{} & \ghost{\mathrm{SingleExc}\,(\mathrm{\theta_{0}})}_<<<{0} & \ghost{\mathrm{DoubleExc}\,(\mathrm{\theta_{0}})} & \multigate{3}{\mathrm{SingleExc}\,(\mathrm{\theta_{0}})}_<<<{1} & \ghost{\mathrm{DoubleExc}\,(\mathrm{\theta_{0}})}_<<<{1} & \ghost{\mathrm{DoubleExc}\,(\mathrm{\theta_{3}})} & \ghost{\mathrm{DoubleExc}\,(\mathrm{\theta_{0}})} & \ghost{\mathrm{DoubleExc}\,(\mathrm{\theta_{4}})} & \qw & \ghost{\mathrm{DoubleExc}\,(\mathrm{\theta_{6}})}_<<<{3} & \qw & \ghost{\mathrm{DoubleExc}\,(\mathrm{\theta_{0}})}_<<<{3} & \ghost{\mathrm{DoubleExc}\,(\mathrm{\theta_{7}})}_<<<{3} & \qw\\
		\nghost{} & \lstick{} & \qw & \ghost{\mathrm{DoubleExc}\,(\mathrm{\theta_{0}})}_<<<{3} & \ghost{\mathrm{SingleExc}\,(\mathrm{\theta_{0}})} & \qw & \ghost{\mathrm{DoubleExc}\,(\mathrm{\theta_{3}})}_<<<{0} & \ghost{\mathrm{DoubleExc}\,(\mathrm{\theta_{0}})}_<<<{3} & \ghost{\mathrm{DoubleExc}\,(\mathrm{\theta_{4}})}_<<<{1} & \qw & \ghost{\mathrm{DoubleExc}\,(\mathrm{\theta_{6}})} & \qw & \ghost{\mathrm{DoubleExc}\,(\mathrm{\theta_{0}})}_<<<{0} & \ghost{\mathrm{DoubleExc}\,(\mathrm{\theta_{7}})}_<<<{0} & \qw\\
		\nghost{} & \lstick{} & \qw & \ghost{\mathrm{DoubleExc}\,(\mathrm{\theta_{0}})}_<<<{1} & \ghost{\mathrm{SingleExc}\,(\mathrm{\theta_{0}})} & \qw & \ghost{\mathrm{DoubleExc}\,(\mathrm{\theta_{3}})}_<<<{3} & \ghost{\mathrm{DoubleExc}\,(\mathrm{\theta_{0}})}_<<<{0} & \qw & \qw & \ghost{\mathrm{DoubleExc}\,(\mathrm{\theta_{6}})}_<<<{1} & \multigate{1}{\mathrm{SingleExc}\,(\mathrm{\theta_{0}})}_<<<{1} & \ghost{\mathrm{DoubleExc}\,(\mathrm{\theta_{0}})}_<<<{1} & \ghost{\mathrm{DoubleExc}\,(\mathrm{\theta_{7}})} & \qw\\
		\nghost{} & \lstick{} & \qw & \qw & \ghost{\mathrm{SingleExc}\,(\mathrm{\theta_{0}})}_<<<{0} & \qw & \ghost{\mathrm{DoubleExc}\,(\mathrm{\theta_{3}})}_<<<{1} & \ghost{\mathrm{DoubleExc}\,(\mathrm{\theta_{0}})}_<<<{1} & \qw & \qw & \qw & \ghost{\mathrm{SingleExc}\,(\mathrm{\theta_{0}})}_<<<{0} & \qw & \ghost{\mathrm{DoubleExc}\,(\mathrm{\theta_{7}})}_<<<{1} & \qw\\
\\ }}
\vspace{0.6em}
\scalebox{0.55}{
\Qcircuit @C=1.0em @R=1.0em @!R { \\
		\nghost{} & \lstick{} & \multigate{10}{\mathrm{DoubleExc}\,(\mathrm{\theta_{8}})}_<<<{0} & \multigate{3}{\mathrm{SingleExc}\,(\mathrm{\theta_{10}})}_<<<{1} & \multigate{11}{\mathrm{DoubleExc}\,(\mathrm{\theta_{0}})}_<<<{0} & \qw & \qw & \qw & \multigate{3}{\mathrm{DoubleExc}\,(\mathrm{\theta_{10}})}_<<<{2} & \qw & \multigate{2}{\mathrm{SingleExc}\,(\mathrm{\theta_{13}})}_<<<{1} & \multigate{5}{\mathrm{DoubleExc}\,(\mathrm{\theta_{14}})}_<<<{0} & \qw & \multigate{4}{\mathrm{DoubleExc}\,(\mathrm{\theta_{12}})}_<<<{2} & \qw\\
		\nghost{} & \lstick{} & \ghost{\mathrm{DoubleExc}\,(\mathrm{\theta_{8}})} & \ghost{\mathrm{SingleExc}\,(\mathrm{\theta_{10}})} & \ghost{\mathrm{DoubleExc}\,(\mathrm{\theta_{0}})}_<<<{2} & \qw & \qw & \multigate{8}{\mathrm{DoubleExc}\,(\mathrm{\theta_{5}})}_<<<{0} & \ghost{\mathrm{DoubleExc}\,(\mathrm{\theta_{10}})}_<<<{0} & \multigate{10}{\mathrm{DoubleExc}\,(\mathrm{\theta_{11}})}_<<<{2} & \ghost{\mathrm{SingleExc}\,(\mathrm{\theta_{13}})} & \ghost{\mathrm{DoubleExc}\,(\mathrm{\theta_{14}})}_<<<{2} & \qw & \ghost{\mathrm{DoubleExc}\,(\mathrm{\theta_{12}})}_<<<{3} & \qw\\
		\nghost{} & \lstick{} & \ghost{\mathrm{DoubleExc}\,(\mathrm{\theta_{8}})} & \ghost{\mathrm{SingleExc}\,(\mathrm{\theta_{10}})} & \ghost{\mathrm{DoubleExc}\,(\mathrm{\theta_{0}})} & \qw & \multigate{7}{\mathrm{DoubleExc}\,(\mathrm{\theta_{5}})}_<<<{0} & \ghost{\mathrm{DoubleExc}\,(\mathrm{\theta_{5}})} & \ghost{\mathrm{DoubleExc}\,(\mathrm{\theta_{10}})}_<<<{3} & \ghost{\mathrm{DoubleExc}\,(\mathrm{\theta_{11}})} & \ghost{\mathrm{SingleExc}\,(\mathrm{\theta_{13}})}_<<<{0} & \ghost{\mathrm{DoubleExc}\,(\mathrm{\theta_{14}})} & \multigate{2}{\mathrm{SingleExc}\,(\mathrm{\theta_{5}})}_<<<{1} & \ghost{\mathrm{DoubleExc}\,(\mathrm{\theta_{12}})}_<<<{0} & \qw\\
		\nghost{} & \lstick{} & \ghost{\mathrm{DoubleExc}\,(\mathrm{\theta_{8}})}_<<<{2} & \ghost{\mathrm{SingleExc}\,(\mathrm{\theta_{10}})}_<<<{0} & \ghost{\mathrm{DoubleExc}\,(\mathrm{\theta_{0}})} & \qw & \ghost{\mathrm{DoubleExc}\,(\mathrm{\theta_{5}})} & \ghost{\mathrm{DoubleExc}\,(\mathrm{\theta_{5}})} & \ghost{\mathrm{DoubleExc}\,(\mathrm{\theta_{10}})}_<<<{1} & \ghost{\mathrm{DoubleExc}\,(\mathrm{\theta_{11}})} & \qw & \ghost{\mathrm{DoubleExc}\,(\mathrm{\theta_{14}})} & \ghost{\mathrm{SingleExc}\,(\mathrm{\theta_{5}})} & \ghost{\mathrm{DoubleExc}\,(\mathrm{\theta_{12}})} & \qw\\
		\nghost{} & \lstick{} & \ghost{\mathrm{DoubleExc}\,(\mathrm{\theta_{8}})} & \qw & \ghost{\mathrm{DoubleExc}\,(\mathrm{\theta_{0}})} & \multigate{7}{\mathrm{DoubleExc}\,(\mathrm{\theta_{6}})}_<<<{2} & \ghost{\mathrm{DoubleExc}\,(\mathrm{\theta_{5}})} & \ghost{\mathrm{DoubleExc}\,(\mathrm{\theta_{5}})}_<<<{2} & \qw & \ghost{\mathrm{DoubleExc}\,(\mathrm{\theta_{11}})} & \qw & \ghost{\mathrm{DoubleExc}\,(\mathrm{\theta_{14}})}_<<<{3} & \ghost{\mathrm{SingleExc}\,(\mathrm{\theta_{5}})}_<<<{0} & \ghost{\mathrm{DoubleExc}\,(\mathrm{\theta_{12}})}_<<<{1} & \qw\\
		\nghost{} & \lstick{} & \ghost{\mathrm{DoubleExc}\,(\mathrm{\theta_{8}})} & \qw & \ghost{\mathrm{DoubleExc}\,(\mathrm{\theta_{0}})} & \ghost{\mathrm{DoubleExc}\,(\mathrm{\theta_{6}})}_<<<{0} & \ghost{\mathrm{DoubleExc}\,(\mathrm{\theta_{5}})}_<<<{2} & \ghost{\mathrm{DoubleExc}\,(\mathrm{\theta_{5}})} & \qw & \ghost{\mathrm{DoubleExc}\,(\mathrm{\theta_{11}})}_<<<{0} & \qw & \ghost{\mathrm{DoubleExc}\,(\mathrm{\theta_{14}})}_<<<{1} & \qw & \qw & \qw\\
		\nghost{} & \lstick{} & \ghost{\mathrm{DoubleExc}\,(\mathrm{\theta_{8}})}_<<<{3} & \qw & \ghost{\mathrm{DoubleExc}\,(\mathrm{\theta_{0}})} & \ghost{\mathrm{DoubleExc}\,(\mathrm{\theta_{6}})} & \ghost{\mathrm{DoubleExc}\,(\mathrm{\theta_{5}})} & \ghost{\mathrm{DoubleExc}\,(\mathrm{\theta_{5}})} & \qw & \ghost{\mathrm{DoubleExc}\,(\mathrm{\theta_{11}})}_<<<{3} & \qw & \qw & \qw & \qw & \qw\\
		\nghost{} & \lstick{} & \ghost{\mathrm{DoubleExc}\,(\mathrm{\theta_{8}})} & \qw & \ghost{\mathrm{DoubleExc}\,(\mathrm{\theta_{0}})}_<<<{3} & \ghost{\mathrm{DoubleExc}\,(\mathrm{\theta_{6}})} & \ghost{\mathrm{DoubleExc}\,(\mathrm{\theta_{5}})}_<<<{3} & \ghost{\mathrm{DoubleExc}\,(\mathrm{\theta_{5}})}_<<<{3} & \multigate{2}{\mathrm{SingleExc}\,(\mathrm{\theta_{4}})}_<<<{1} & \ghost{\mathrm{DoubleExc}\,(\mathrm{\theta_{11}})} & \qw & \qw & \multigate{4}{\mathrm{SingleExc}\,(\mathrm{\theta_{0}})}_<<<{1} & \qw & \qw\\
		\nghost{} & \lstick{} & \ghost{\mathrm{DoubleExc}\,(\mathrm{\theta_{8}})} & \qw & \ghost{\mathrm{DoubleExc}\,(\mathrm{\theta_{0}})} & \ghost{\mathrm{DoubleExc}\,(\mathrm{\theta_{6}})} & \ghost{\mathrm{DoubleExc}\,(\mathrm{\theta_{5}})} & \ghost{\mathrm{DoubleExc}\,(\mathrm{\theta_{5}})} & \ghost{\mathrm{SingleExc}\,(\mathrm{\theta_{4}})} & \ghost{\mathrm{DoubleExc}\,(\mathrm{\theta_{11}})} & \multigate{3}{\mathrm{SingleExc}\,(\mathrm{\theta_{12}})}_<<<{1} & \qw & \ghost{\mathrm{SingleExc}\,(\mathrm{\theta_{0}})} & \qw & \qw\\
		\nghost{} & \lstick{} & \ghost{\mathrm{DoubleExc}\,(\mathrm{\theta_{8}})} & \multigate{2}{\mathrm{SingleExc}\,(\mathrm{\theta_{9}})}_<<<{1} & \ghost{\mathrm{DoubleExc}\,(\mathrm{\theta_{0}})} & \ghost{\mathrm{DoubleExc}\,(\mathrm{\theta_{6}})} & \ghost{\mathrm{DoubleExc}\,(\mathrm{\theta_{5}})}_<<<{1} & \ghost{\mathrm{DoubleExc}\,(\mathrm{\theta_{5}})}_<<<{1} & \ghost{\mathrm{SingleExc}\,(\mathrm{\theta_{4}})}_<<<{0} & \ghost{\mathrm{DoubleExc}\,(\mathrm{\theta_{11}})} & \ghost{\mathrm{SingleExc}\,(\mathrm{\theta_{12}})} & \qw & \ghost{\mathrm{SingleExc}\,(\mathrm{\theta_{0}})} & \qw & \qw\\
		\nghost{} & \lstick{} & \ghost{\mathrm{DoubleExc}\,(\mathrm{\theta_{8}})}_<<<{1} & \ghost{\mathrm{SingleExc}\,(\mathrm{\theta_{9}})} & \ghost{\mathrm{DoubleExc}\,(\mathrm{\theta_{0}})} & \ghost{\mathrm{DoubleExc}\,(\mathrm{\theta_{6}})}_<<<{3} & \qw & \qw & \qw & \ghost{\mathrm{DoubleExc}\,(\mathrm{\theta_{11}})} & \ghost{\mathrm{SingleExc}\,(\mathrm{\theta_{12}})} & \multigate{1}{\mathrm{SingleExc}\,(\mathrm{\theta_{15}})}_<<<{1} & \ghost{\mathrm{SingleExc}\,(\mathrm{\theta_{0}})} & \qw & \qw\\
		\nghost{} & \lstick{} & \qw & \ghost{\mathrm{SingleExc}\,(\mathrm{\theta_{9}})}_<<<{0} & \ghost{\mathrm{DoubleExc}\,(\mathrm{\theta_{0}})}_<<<{1} & \ghost{\mathrm{DoubleExc}\,(\mathrm{\theta_{6}})}_<<<{1} & \qw & \qw & \qw & \ghost{\mathrm{DoubleExc}\,(\mathrm{\theta_{11}})}_<<<{1} & \ghost{\mathrm{SingleExc}\,(\mathrm{\theta_{12}})}_<<<{0} & \ghost{\mathrm{SingleExc}\,(\mathrm{\theta_{15}})}_<<<{0} & \ghost{\mathrm{SingleExc}\,(\mathrm{\theta_{0}})}_<<<{0} & \qw & \qw\\
\\ }}
\vspace{0.6em}
\scalebox{0.55}{
\Qcircuit @C=1.0em @R=1.0em @!R { \\
		\nghost{} & \lstick{} & \multigate{2}{\mathrm{SingleExc}\,(\mathrm{\theta_{5}})}_<<<{1} & \multigate{10}{\mathrm{DoubleExc}\,(\mathrm{\theta_{16}})}_<<<{0} & \qw & \qw & \qw & \multigate{9}{\mathrm{DoubleExc}\,(\mathrm{\theta_{17}})}_<<<{2} & \multigate{2}{\mathrm{SingleExc}\,(\mathrm{\theta_{5}})}_<<<{1} & \multigate{1}{\mathrm{SingleExc}\,(\mathrm{\theta_{20}})}_<<<{1} & \qw & \multigate{11}{\mathrm{DoubleExc}\,(\mathrm{\theta_{21}})}_<<<{0} & \multigate{4}{\mathrm{SingleExc}\,(\mathrm{\theta_{24}})}_<<<{1} & \qw & \qw\\
		\nghost{} & \lstick{} & \ghost{\mathrm{SingleExc}\,(\mathrm{\theta_{5}})} & \ghost{\mathrm{DoubleExc}\,(\mathrm{\theta_{16}})}_<<<{2} & \qw & \qw & \qw & \ghost{\mathrm{DoubleExc}\,(\mathrm{\theta_{17}})} & \ghost{\mathrm{SingleExc}\,(\mathrm{\theta_{5}})} & \ghost{\mathrm{SingleExc}\,(\mathrm{\theta_{20}})}_<<<{0} & \qw & \ghost{\mathrm{DoubleExc}\,(\mathrm{\theta_{21}})} & \ghost{\mathrm{SingleExc}\,(\mathrm{\theta_{24}})} & \qw & \qw\\
		\nghost{} & \lstick{} & \ghost{\mathrm{SingleExc}\,(\mathrm{\theta_{5}})}_<<<{0} & \ghost{\mathrm{DoubleExc}\,(\mathrm{\theta_{16}})} & \qw & \qw & \qw & \ghost{\mathrm{DoubleExc}\,(\mathrm{\theta_{17}})} & \ghost{\mathrm{SingleExc}\,(\mathrm{\theta_{5}})}_<<<{0} & \multigate{3}{\mathrm{SingleExc}\,(\mathrm{\theta_{25}})}_<<<{1} & \qw & \ghost{\mathrm{DoubleExc}\,(\mathrm{\theta_{21}})} & \ghost{\mathrm{SingleExc}\,(\mathrm{\theta_{24}})} & \qw & \qw\\
		\nghost{} & \lstick{} & \qw & \ghost{\mathrm{DoubleExc}\,(\mathrm{\theta_{16}})} & \qw & \qw & \qw & \ghost{\mathrm{DoubleExc}\,(\mathrm{\theta_{17}})}_<<<{0} & \qw & \ghost{\mathrm{SingleExc}\,(\mathrm{\theta_{25}})} & \qw & \ghost{\mathrm{DoubleExc}\,(\mathrm{\theta_{21}})} & \ghost{\mathrm{SingleExc}\,(\mathrm{\theta_{24}})} & \qw & \qw\\
		\nghost{} & \lstick{} & \qw & \ghost{\mathrm{DoubleExc}\,(\mathrm{\theta_{16}})} & \qw & \qw & \multigate{7}{\mathrm{DoubleExc}\,(\mathrm{\theta_{0}})}_<<<{2} & \ghost{\mathrm{DoubleExc}\,(\mathrm{\theta_{17}})} & \qw & \ghost{\mathrm{SingleExc}\,(\mathrm{\theta_{25}})} & \qw & \ghost{\mathrm{DoubleExc}\,(\mathrm{\theta_{21}})}_<<<{2} & \ghost{\mathrm{SingleExc}\,(\mathrm{\theta_{24}})}_<<<{0} & \qw & \qw\\
		\nghost{} & \lstick{} & \qw & \ghost{\mathrm{DoubleExc}\,(\mathrm{\theta_{16}})} & \qw & \qw & \ghost{\mathrm{DoubleExc}\,(\mathrm{\theta_{0}})}_<<<{0} & \ghost{\mathrm{DoubleExc}\,(\mathrm{\theta_{17}})} & \qw & \ghost{\mathrm{SingleExc}\,(\mathrm{\theta_{25}})}_<<<{0} & \qw & \ghost{\mathrm{DoubleExc}\,(\mathrm{\theta_{21}})} & \qw & \qw & \qw\\
		\nghost{} & \lstick{} & \qw & \ghost{\mathrm{DoubleExc}\,(\mathrm{\theta_{16}})}_<<<{3} & \multigate{5}{\mathrm{DoubleExc}\,(\mathrm{\theta_{13}})}_<<<{0} & \multigate{4}{\mathrm{DoubleExc}\,(\mathrm{\theta_{2}})}_<<<{0} & \ghost{\mathrm{DoubleExc}\,(\mathrm{\theta_{0}})} & \ghost{\mathrm{DoubleExc}\,(\mathrm{\theta_{17}})}_<<<{3} & \multigate{4}{\mathrm{DoubleExc}\,(\mathrm{\theta_{19}})}_<<<{2} & \qw & \qw & \ghost{\mathrm{DoubleExc}\,(\mathrm{\theta_{21}})} & \qw & \multigate{3}{\mathrm{DoubleExc}\,(\mathrm{\theta_{22}})}_<<<{2} & \qw\\
		\nghost{} & \lstick{} & \qw & \ghost{\mathrm{DoubleExc}\,(\mathrm{\theta_{16}})} & \ghost{\mathrm{DoubleExc}\,(\mathrm{\theta_{13}})}_<<<{2} & \ghost{\mathrm{DoubleExc}\,(\mathrm{\theta_{2}})}_<<<{2} & \ghost{\mathrm{DoubleExc}\,(\mathrm{\theta_{0}})} & \ghost{\mathrm{DoubleExc}\,(\mathrm{\theta_{17}})} & \ghost{\mathrm{DoubleExc}\,(\mathrm{\theta_{19}})}_<<<{3} & \qw & \qw & \ghost{\mathrm{DoubleExc}\,(\mathrm{\theta_{21}})}_<<<{3} & \multigate{4}{\mathrm{SingleExc}\,(\mathrm{\theta_{0}})}_<<<{1} & \ghost{\mathrm{DoubleExc}\,(\mathrm{\theta_{22}})}_<<<{0} & \qw\\
		\nghost{} & \lstick{} & \qw & \ghost{\mathrm{DoubleExc}\,(\mathrm{\theta_{16}})} & \ghost{\mathrm{DoubleExc}\,(\mathrm{\theta_{13}})} & \ghost{\mathrm{DoubleExc}\,(\mathrm{\theta_{2}})} & \ghost{\mathrm{DoubleExc}\,(\mathrm{\theta_{0}})} & \ghost{\mathrm{DoubleExc}\,(\mathrm{\theta_{17}})} & \ghost{\mathrm{DoubleExc}\,(\mathrm{\theta_{19}})}_<<<{0} & \qw & \qw & \ghost{\mathrm{DoubleExc}\,(\mathrm{\theta_{21}})} & \ghost{\mathrm{SingleExc}\,(\mathrm{\theta_{0}})} & \ghost{\mathrm{DoubleExc}\,(\mathrm{\theta_{22}})}_<<<{3} & \qw\\
		\nghost{} & \lstick{} & \qw & \ghost{\mathrm{DoubleExc}\,(\mathrm{\theta_{16}})} & \ghost{\mathrm{DoubleExc}\,(\mathrm{\theta_{13}})}_<<<{3} & \ghost{\mathrm{DoubleExc}\,(\mathrm{\theta_{2}})}_<<<{3} & \ghost{\mathrm{DoubleExc}\,(\mathrm{\theta_{0}})} & \ghost{\mathrm{DoubleExc}\,(\mathrm{\theta_{17}})}_<<<{1} & \ghost{\mathrm{DoubleExc}\,(\mathrm{\theta_{19}})} & \qw & \qw & \ghost{\mathrm{DoubleExc}\,(\mathrm{\theta_{21}})} & \ghost{\mathrm{SingleExc}\,(\mathrm{\theta_{0}})} & \ghost{\mathrm{DoubleExc}\,(\mathrm{\theta_{22}})}_<<<{1} & \qw\\
		\nghost{} & \lstick{} & \qw & \ghost{\mathrm{DoubleExc}\,(\mathrm{\theta_{16}})}_<<<{1} & \ghost{\mathrm{DoubleExc}\,(\mathrm{\theta_{13}})} & \ghost{\mathrm{DoubleExc}\,(\mathrm{\theta_{2}})}_<<<{1} & \ghost{\mathrm{DoubleExc}\,(\mathrm{\theta_{0}})}_<<<{3} & \multigate{1}{\mathrm{SingleExc}\,(\mathrm{\theta_{18}})}_<<<{1} & \ghost{\mathrm{DoubleExc}\,(\mathrm{\theta_{19}})}_<<<{1} & \multigate{1}{\mathrm{SingleExc}\,(\mathrm{\theta_{0}})}_<<<{1} & \multigate{1}{\mathrm{SingleExc}\,(\mathrm{\theta_{0}})}_<<<{1} & \ghost{\mathrm{DoubleExc}\,(\mathrm{\theta_{21}})} & \ghost{\mathrm{SingleExc}\,(\mathrm{\theta_{0}})} & \qw & \qw\\
		\nghost{} & \lstick{} & \qw & \qw & \ghost{\mathrm{DoubleExc}\,(\mathrm{\theta_{13}})}_<<<{1} & \qw & \ghost{\mathrm{DoubleExc}\,(\mathrm{\theta_{0}})}_<<<{1} & \ghost{\mathrm{SingleExc}\,(\mathrm{\theta_{18}})}_<<<{0} & \qw & \ghost{\mathrm{SingleExc}\,(\mathrm{\theta_{0}})}_<<<{0} & \ghost{\mathrm{SingleExc}\,(\mathrm{\theta_{0}})}_<<<{0} & \ghost{\mathrm{DoubleExc}\,(\mathrm{\theta_{21}})}_<<<{1} & \ghost{\mathrm{SingleExc}\,(\mathrm{\theta_{0}})}_<<<{0} & \qw & \qw\\
\\ }}
\vspace{0.6em}
\scalebox{0.55}{
\Qcircuit @C=1.0em @R=1.0em @!R { \\
		\nghost{} & \lstick{} & \qw & \qw & \qw & \multigate{9}{\mathrm{DoubleExc}\,(\mathrm{\theta_{26}})}_<<<{2} & \multigate{2}{\mathrm{SingleExc}\,(\mathrm{\theta_{27}})}_<<<{1} & \multigate{2}{\mathrm{SingleExc}\,(\mathrm{\theta_{28}})}_<<<{1} & \qw & \qw & \qw & \multigate{11}{\mathrm{DoubleExc}\,(\mathrm{\theta_{5}})}_<<<{0} & \multigate{3}{\mathrm{SingleExc}\,(\mathrm{\theta_{30}})}_<<<{1} & \multigate{1}{\mathrm{SingleExc}\,(\mathrm{\theta_{31}})}_<<<{1} & \qw\\
		\nghost{} & \lstick{} & \qw & \qw & \qw & \ghost{\mathrm{DoubleExc}\,(\mathrm{\theta_{26}})} & \ghost{\mathrm{SingleExc}\,(\mathrm{\theta_{27}})} & \ghost{\mathrm{SingleExc}\,(\mathrm{\theta_{28}})} & \qw & \qw & \qw & \ghost{\mathrm{DoubleExc}\,(\mathrm{\theta_{5}})} & \ghost{\mathrm{SingleExc}\,(\mathrm{\theta_{30}})} & \ghost{\mathrm{SingleExc}\,(\mathrm{\theta_{31}})}_<<<{0} & \qw\\
		\nghost{} & \lstick{} & \qw & \qw & \qw & \ghost{\mathrm{DoubleExc}\,(\mathrm{\theta_{26}})}_<<<{0} & \ghost{\mathrm{SingleExc}\,(\mathrm{\theta_{27}})}_<<<{0} & \ghost{\mathrm{SingleExc}\,(\mathrm{\theta_{28}})}_<<<{0} & \qw & \qw & \qw & \ghost{\mathrm{DoubleExc}\,(\mathrm{\theta_{5}})} & \ghost{\mathrm{SingleExc}\,(\mathrm{\theta_{30}})} & \qw & \qw\\
		\nghost{} & \lstick{} & \qw & \qw & \qw & \ghost{\mathrm{DoubleExc}\,(\mathrm{\theta_{26}})} & \qw & \qw & \qw & \qw & \qw & \ghost{\mathrm{DoubleExc}\,(\mathrm{\theta_{5}})} & \ghost{\mathrm{SingleExc}\,(\mathrm{\theta_{30}})}_<<<{0} & \qw & \qw\\
		\nghost{} & \lstick{} & \qw & \qw & \qw & \ghost{\mathrm{DoubleExc}\,(\mathrm{\theta_{26}})} & \qw & \qw & \qw & \qw & \qw & \ghost{\mathrm{DoubleExc}\,(\mathrm{\theta_{5}})}_<<<{2} & \qw & \qw & \qw\\
		\nghost{} & \lstick{} & \qw & \qw & \qw & \ghost{\mathrm{DoubleExc}\,(\mathrm{\theta_{26}})} & \qw & \qw & \qw & \qw & \qw & \ghost{\mathrm{DoubleExc}\,(\mathrm{\theta_{5}})} & \qw & \qw & \qw\\
		\nghost{} & \lstick{} & \qw & \qw & \qw & \ghost{\mathrm{DoubleExc}\,(\mathrm{\theta_{26}})} & \qw & \qw & \qw & \qw & \qw & \ghost{\mathrm{DoubleExc}\,(\mathrm{\theta_{5}})} & \qw & \multigate{4}{\mathrm{DoubleExc}\,(\mathrm{\theta_{0}})}_<<<{2} & \qw\\
		\nghost{} & \lstick{} & \multigate{3}{\mathrm{DoubleExc}\,(\mathrm{\theta_{23}})}_<<<{0} & \multigate{2}{\mathrm{SingleExc}\,(\mathrm{\theta_{6}})}_<<<{1} & \multigate{1}{\mathrm{SingleExc}\,(\mathrm{\theta_{7}})}_<<<{1} & \ghost{\mathrm{DoubleExc}\,(\mathrm{\theta_{26}})}_<<<{3} & \qw & \qw & \qw & \qw & \multigate{3}{\mathrm{DoubleExc}\,(\mathrm{\theta_{5}})}_<<<{0} & \ghost{\mathrm{DoubleExc}\,(\mathrm{\theta_{5}})} & \qw & \ghost{\mathrm{DoubleExc}\,(\mathrm{\theta_{0}})} & \qw\\
		\nghost{} & \lstick{} & \ghost{\mathrm{DoubleExc}\,(\mathrm{\theta_{23}})}_<<<{2} & \ghost{\mathrm{SingleExc}\,(\mathrm{\theta_{6}})} & \ghost{\mathrm{SingleExc}\,(\mathrm{\theta_{7}})}_<<<{0} & \ghost{\mathrm{DoubleExc}\,(\mathrm{\theta_{26}})} & \multigate{3}{\mathrm{SingleExc}\,(\mathrm{\theta_{4}})}_<<<{1} & \qw & \qw & \qw & \ghost{\mathrm{DoubleExc}\,(\mathrm{\theta_{5}})}_<<<{2} & \ghost{\mathrm{DoubleExc}\,(\mathrm{\theta_{5}})}_<<<{3} & \qw & \ghost{\mathrm{DoubleExc}\,(\mathrm{\theta_{0}})}_<<<{3} & \qw\\
		\nghost{} & \lstick{} & \ghost{\mathrm{DoubleExc}\,(\mathrm{\theta_{23}})}_<<<{3} & \ghost{\mathrm{SingleExc}\,(\mathrm{\theta_{6}})}_<<<{0} & \qw & \ghost{\mathrm{DoubleExc}\,(\mathrm{\theta_{26}})}_<<<{1} & \ghost{\mathrm{SingleExc}\,(\mathrm{\theta_{4}})} & \multigate{2}{\mathrm{SingleExc}\,(\mathrm{\theta_{29}})}_<<<{1} & \multigate{2}{\mathrm{SingleExc}\,(\mathrm{\theta_{22}})}_<<<{1} & \qw & \ghost{\mathrm{DoubleExc}\,(\mathrm{\theta_{5}})}_<<<{3} & \ghost{\mathrm{DoubleExc}\,(\mathrm{\theta_{5}})} & \multigate{2}{\mathrm{SingleExc}\,(\mathrm{\theta_{6}})}_<<<{1} & \ghost{\mathrm{DoubleExc}\,(\mathrm{\theta_{0}})}_<<<{0} & \qw\\
		\nghost{} & \lstick{} & \ghost{\mathrm{DoubleExc}\,(\mathrm{\theta_{23}})}_<<<{1} & \multigate{1}{\mathrm{SingleExc}\,(\mathrm{\theta_{1}})}_<<<{1} & \qw & \qw & \ghost{\mathrm{SingleExc}\,(\mathrm{\theta_{4}})} & \ghost{\mathrm{SingleExc}\,(\mathrm{\theta_{29}})} & \ghost{\mathrm{SingleExc}\,(\mathrm{\theta_{22}})} & \multigate{1}{\mathrm{SingleExc}\,(\mathrm{\theta_{22}})}_<<<{1} & \ghost{\mathrm{DoubleExc}\,(\mathrm{\theta_{5}})}_<<<{1} & \ghost{\mathrm{DoubleExc}\,(\mathrm{\theta_{5}})} & \ghost{\mathrm{SingleExc}\,(\mathrm{\theta_{6}})} & \ghost{\mathrm{DoubleExc}\,(\mathrm{\theta_{0}})}_<<<{1} & \qw\\
		\nghost{} & \lstick{} & \qw & \ghost{\mathrm{SingleExc}\,(\mathrm{\theta_{1}})}_<<<{0} & \qw & \qw & \ghost{\mathrm{SingleExc}\,(\mathrm{\theta_{4}})}_<<<{0} & \ghost{\mathrm{SingleExc}\,(\mathrm{\theta_{29}})}_<<<{0} & \ghost{\mathrm{SingleExc}\,(\mathrm{\theta_{22}})}_<<<{0} & \ghost{\mathrm{SingleExc}\,(\mathrm{\theta_{22}})}_<<<{0} & \qw & \ghost{\mathrm{DoubleExc}\,(\mathrm{\theta_{5}})}_<<<{1} & \ghost{\mathrm{SingleExc}\,(\mathrm{\theta_{6}})}_<<<{0} & \qw & \qw\\
\\ }}